\documentclass[12pt,reqno]{amsart}

\usepackage[margin = 1.25in]{geometry}
\usepackage{amssymb,amsmath,amsthm,amsfonts}
\usepackage{float}
\usepackage{afterpage}
\usepackage{times}
\usepackage[flushleft]{threeparttable}
\usepackage[title]{appendix}

\newtheorem{assumption}{Assumption}
\newtheorem{proposition}{Proposition}
\newtheorem{example}{Example}
\newtheorem{remark}{Remark}
\usepackage[]{pdfpages}
\newtheorem{theorem}{Theorem}
\newtheorem*{theorem*}{Theorem}
\newtheorem*{proposition*}{Proposition}
\newtheorem{lemma}{Lemma}
\usepackage[ruled,vlined]{algorithm2e}
\usepackage{mathrsfs}
\newtheorem{corollary}{Corollary}
\newtheorem*{claim*}{Claim}
\newtheorem*{remark*}{Remark}
\providecommand{\dist}{\operatorname{dist}}
\providecommand{\conv}{\operatorname{conv}}

\LinesNumbered
\DeclareMathOperator*{\argmin}{arg\,min}	
\DeclareMathOperator*{\argmax}{arg\,max}

\usepackage{graphicx}
\usepackage[inline]{enumitem}
\setlist[enumerate]{label=\arabic*.}
\usepackage[colorlinks,
linkcolor = black,
            urlcolor  = blue,
            citecolor = black]{hyperref}
\usepackage[round]{natbib}
\usepackage{booktabs}
\usepackage{multirow}
\usepackage[all=normal, paragraphs=tight, bibbreaks=tight, floats=tight,
mathdisplays=tight,bibnotes=tight]{savetrees}
\makeatletter
\renewcommand\paragraph{\@startsection{paragraph}{4}%
  \z@{.5em}{-1em}%
  {\normalfont\bfseries}}
  \renewcommand\@addpunct[1]{}
\makeatother
\title{Nonparametric Bayesian Inference for Partially Identified Discrete Response Models}
\author{Elie Tamer}
\author{Christopher D. Walker}
\thanks{\today. \@ Tamer: Harvard University, Department of Economics, \href{mailto:elietamer@fas.harvard.edu}{elietamer@fas.harvard.edu}; Walker: Duke University, Department of Economics, \href{mailto:christopher.walker@duke.edu}{christopher.walker@duke.edu}. This work is preliminary and any comments are welcome. We are grateful to Aslihan Asil, Allan Collard-Wexler, and Adam Rosen for helpful discussions. We thank Mengyuan Jiang for research assistance. All errors are our own.}
\begin{document}
\pagestyle{plain}
\maketitle
\begin{abstract}
This paper proposes a nonparametric Bayesian inference framework for partially identified discrete response models. The key observation is that these models map a reduced-form conditional choice probability to an identified set. Consequently, nonparametric Bayesian inference for the conditional probability mass function leads to Bayesian inference for the identified set. The inference framework nests conditional moment inequalities and linear systems with unknown coefficients as special cases. Importantly, our proposal does not require converting conditional moments into unconditional moments or discretizing covariates. We show that the posterior is consistent for the true identified set when the model is correctly specified, show that the posterior can consistently detect model misspecification, and show posterior consistency for a pseudo-identified set that is valid under misspecification. We also verify the assumptions for a class of priors based on Gaussian processes that we use to implement our proposal. These priors offer similar flexibility to frequentist partial identification methods, and are computationally attractive because posterior sampling can be performed in closed-form. We also show that many of the ideas in this paper extend to continuous responses and aggregated discrete responses (e.g., market shares).
\end{abstract}
\newpage
\section{Introduction}
\subsection{Overview.} Many economic models with discrete responses lead to set-identified structural parameters. These include games with multiple equilibria \citep{tamer2003incomplete,ciliberto2009market}, semiparametric static binary/multinomial choice models with fixed effects \citep{manski1987semiparametric,shi2018estimating,gao2024identification,pakes2024moment}, dynamic models with state dependence \citep{Heckman1981a,honore2006bounds,torgovitsky2019nonparametric,khan2021inference,pakes2021unobserved,khan2023identification}, limited consideration discrete choice models \citep{barseghyan2021heterogeneous,LU2022368}, ordered choice models \citep{CHESHER201233,pakes2015moment}, network formation models \citep{sheng2020structural}, and binary regression models with interval-censored covariates \citep{manski2002inference}, to name a few. A common statistical theme across these examples is that the identified set is known up to an unrestricted reduced-form conditional probability mass function (PMF). Consequently, uncertainty about the identified set entirely reflects that of the PMF, and, as a result, statistical inference for the PMF should translate to the identified set.

Using this observation, our paper proposes a nonparametric Bayesian inference framework for identified sets defined by conditional PMFs. Concretely, let the observed data be $(Y,X')'$, where $Y$ is a discrete random variable and $X$ is a (possibly) continuous random vector, and let the target parameter be an identified set $\Gamma_{I} = \{\gamma: f(X,p(X),\gamma) \geq 0 \ a.s.\}$, where $\gamma$ is a finite-dimensional economic parameter (e.g., payoff parameters), $p(X)$ is a reduced-form conditional PMF for $Y$ given $X$, and $f(X,p(X),\gamma)$ is a known vector of possibly nonseparable functions (with inequalities evaluated elementwise). The central idea of our proposal is that $\Gamma_{I}$ can be viewed as the output of a correspondence that takes $p$ as an input. Consequently, a Bayesian inference framework for $\Gamma_{I}$ arises in which a marginal posterior $\Gamma_{I}|Y,X$ is obtained from a posterior $p|Y,X$ for $p$. Beyond this, we provide implementation details for Gaussian process (GP) priors, prove a general posterior consistency theorem for $\Gamma_{I}|Y,X$ (and for a pseudo-identified set that accommodates misspecification), and verify the conditions of the theorem for GPs. The examples below highlight the breadth of our framework; Section \ref{sec:extensions} also demonstrates how many of our results extend to continuous $Y$ and aggregated discrete choice (e.g., $Y$ is a market share).
\begin{example}[Conditional Moment Inequalities]\label{ex:momentsineq}
Let $g(Y,X,\gamma)$ be a vector of known functions. The conditional moment inequality (CMI) model says $\gamma$ is in the identified set iff $E[g(Y,X,\gamma)|X] \geq 0$ a.s. CMIs arise in discrete choice, for instance, due to model incompleteness \citep{tamer2003incomplete,ciliberto2009market,sheng2020structural,barseghyan2021heterogeneous,LU2022368} or weak assumptions about unobserved heterogeneity \citep{manski2002inference,shi2018estimating,pakes2021unobserved,pakes2024moment}. Defining $f(X,p(X),\gamma) = \sum_{y}g(y,X,\gamma)p_{y}(X),$ where $p_{y}(X)$ is the conditional probability that $Y=y$ under $p(X)$, CMIs fall within our framework.
    
\end{example}
\begin{example}[Linear Systems]
Some discrete choice models lead to identifying restrictions of the form $\mathbf{1}\{p(X) \in R\}c_{R}(X)\gamma \geq 0 \ \forall \ R \in \mathcal{R}$ a.s., where $R$ is some restriction on $p(X)$, $c_{R}(X)$ is a known transformation that depends on $R$, $\mathcal{R}$ is a set of restrictions, and $\gamma$ is a structural parameter.\footnote{For instance, in the dynamic binary panel data model of \cite{khan2023identification}, inequalities like $\mathbf{1}\{p(Y_{t-1}=1,Y_{t}=1|X) +  P(Y_{s-1}=1,Y_{s} = 0|X) \geq 1\}\{(X_{t}-X_{s})'\beta + \theta\} \geq 0$, where $s,t \in \{1,...,T\}$ are time periods, and $X=(X_{t})_{t=1}^{T}$ is a vector of covariates, characterize the sharp identified set.} These are a special case of a class of linear systems $A(X,p(X))\gamma \geq b(X,p(X))$ a.s., where $A(X,p(X))$ is a matrix and $b(X,p(X))$ is a vector, which is cast within our framework by defining $f(X,p(X),\gamma) = A(X,p(X))\gamma-b(X,p(X))$.
\end{example}

Our proposal has several appealing features. First, our approach accommodates continuous $X$ without requiring the conversion of conditional moments into unconditional moments or explicitly discretizing the support of $X$ into bins. In CMI models, the one-sided nature of inequality restrictions means that making these choices in a way that preserves identifying information is challenging \citep{khan2009inference, andrews2013inference}. This has resulted in a widespread empirical practice of basing inference on ad hoc choices. For instance, in many empirical implementations of CMIs, the conditioning variables are discretized into finite cells or groups. The resulting cell-level inequalities can then be treated as a finite set of unconditional moment inequalities, typically by multiplying the conditional moment by cell indicators or other nonnegative `instruments' and averaging (see, for example, \cite{kline2016bayesian,chen2018monte,ciliberto2021market, pakes2021unobserved}, among many others). Ad hoc selections of unconditional moments may result in substantial loss of sharpness under correct specification, while conclusions drawn from outer sets based on unconditional moments may be sensitive to instrument choice when the underyling CMIs are misspecified \citep{li2024discordant}. Similarly, empirical applications of linear systems typically involve coarsened covariates. For instance, the linear programming formulation \cite{honore2006bounds} presumes that the covariates have finite support (so as to obtain a finite set of inequalities), while the empirical implementation of the identified set in \cite{khan2023identification} binarizes the covariates, despite the finding that rich covariate support can lead to much sharper identified sets (and even point identification in some cases). Similar discretizing steps were taken in \cite{ciliberto2009market} and \cite{kline2016bayesian}. In both CMIs and linear systems, our procedure obviates the need for such discretization and/or selecting unconditional moments by targeting the conditional PMF directly. 

Second, our approach is flexible and amenable to computation. We provide implementation details for a class of priors that combine stick-breaking and Gaussian process (GP) methods.\footnote{Stick-breaking is a technique that dates back to \cite{halmos1944random} and can be thought of as representing a discrete distribution in terms of its hazard functions. It was famously employed by \cite{sethuraman1994constructive} as a computational device for the \cite{ferguson1973bayesian} Dirichlet process, however these techniques are increasingly used for Bayesian conditional distribution estimation \citep{dunson2008kernel,chung2009nonparametric,rodriguez2011nonparametric,ren2011logistic}.} GP priors allow a researcher to flexibly incorporate nonparametric properties, such as smoothness, and produce nonparametric Bayesian procedures with similar flexibility to frequentist nonparametric estimators \citep{vaart2008rates}. Consequently, our implementation shares the advantageous flexibility of frequentist inference for identified sets, while enjoying all of the niceties of Bayesian inference. These priors are also computationally attractive because the stick-breaking characterization enables the use of the P\'{o}lya-Gamma data augmentation technique of \cite{polson2013bayesian} to derive Gibbs sampling algorithms in which GPs form conditionally conjugate priors. As a result, sampling from the posterior for the reduced-form PMFs amounts to generating draws from standard parametric distributions. Given draws from $p|Y,X$, the researcher can use \textit{any} appropriate method for computing $\Gamma_{I}$ to obtain the corresponding draws from $\Gamma_{I}|Y,X$, while the two-step nature of our proposal means that computing the $\Gamma_{I}$ draws is parallelizable.

The main theoretical results in the paper are posterior consistency guarantees for the identified set. A possible concern about Bayesian inference for partially identified models is that, even with infinite data, the prior does not fully revise \citep{poirier1998revising,moon2012bayesian,giacomini2021robust}. In our setting, the prior is placed directly on an identifiable reduced-form parameter $p$ and the identified set $\Gamma_{I}$ is viewed as a point identified, set-valued parameter. Consequently, under suitable continuity restrictions on the mapping from reduced form PMFs to identified sets, fully revised beliefs for the PMF should translate to fully revised beliefs about the identified set. Our core posterior consistency theorem confirms this intuition, providing high-level conditions under which posterior consistency for the conditional choice probabilities leads to the posterior for the identified set to concentrate around the true identified set. In this sense, the prior for $p$ has an asymptotically negligible effect on inferences about $\Gamma_{I}$. We also prove posterior consistency for a pseudo-identified set $\tilde{\Gamma}_{I}$ that is always nonempty (and coincides with $\Gamma_{I}$ when $\Gamma_{I}$ is not empty), and show that the posterior probability that the identified set is empty is a consistent diagnostic for model misspecification. This is important because several of the examples that fit within our framework are prone to model misspecification, for instance, due to parametric restrictions on unobserved heterogeneity \citep{ciliberto2009market,barseghyan2021heterogeneous,ciliberto2021market}.

Our posterior consistency theorems are based on assumptions that the posterior concentrates on sets of $p$ for which certain uniform convergence and lower hemicontinuity conditions are satisfied. While these conditions are straightforward to verify in finite-dimensional models, posterior consistency in infinite-dimensional models can be a delicate issue, evidenced by canonical inconsistency results, such as those in \cite{freedman1963asymptotic, freedman1965asymptotic} and \cite{diaconis1986inconsistent,diaconis1986consistency,diaconis1993nonparametric}, and more recent work highlighting the intricacies of posterior consistency in norms relevant to statistical functionals \citep{gine2011rates,castillo2014bayesian,ho2024bayesian}. Since we cast the identified set as a statistical functional, this second set of concerns is relevant for our framework. Conscious of this, we enhance the credibility of our general theory with a set low-level sufficient conditions for an important class of identified sets (i.e., they apply to linear systems and some conditional moment inequality models) under both correct specification (i.e., $\Gamma_{I}$) and misspecification (i.e., $\tilde{\Gamma}_{I}$). Then, noticing a common underlying supremum norm posterior consistency requirement for $p$, we derive conditions under which the stick-breaking GP prior used for the implementation meets this condition, and thereby provides a complete verification of posterior consistency. As an intermediate technical result, we find conditions under which the stick-breaking GP posterior for $p$ converges to the truth $p_{0}$ in empirical mean-square at the minimax optimal rate, which may be of independent interest.
\subsection{Literature.}
This paper is related to several literatures in econometrics. The first is Bayesian analysis of partially identified models. We build on \cite{kline2016bayesian} by performing Bayesian inference for an identified set by evaluating it at draws from the posterior of an identifiable reduced-form parameter (i.e., the PMF). A key difference, however, is that \cite{kline2016bayesian} explicitly restrict the reduced-form parameter to be finite-dimensional, thereby ruling out partially identified discrete response models with continuous covariates. \cite{LIAO2019338} and \cite{florens2021revisiting} also derive reduced-form Bayesian inference frameworks for identified sets using support functions and Dirichlet processes, respectively. These papers similarly restrict attention to reduced-form parameters that are estimable at the parametric rate, precluding settings where the identified set is indexed by a reduced-form conditional PMF with continuous conditioning variables. \cite{norets2014semiparametric} propose Bayesian inference for \cite{rust1987optimal}-type structural dynamic discrete choice models by assuming a joint prior for the reduced-form conditional choice probabilities and observable state transition distributions, but restrict attention to finite observable state spaces. More recently, taking a posterior for an identified set as given, \cite{kline2024counterfactual} propose a Bayesian inference framework for counterfactual prediction sets, and establish continuous mapping theorems under which posterior consistency for the identified set implies predictive consistency for the counterfactuals. Since our posterior consistency results for $\Gamma_{I}$ match those required by \cite{kline2024counterfactual}, an implication of our theory is posterior consistency for counterfactual prediction sets without requiring discretization of covariates. Other papers on Bayesian inference under partial identification include \cite{poirier1998revising}, \cite{liao2010bayesian}, \cite{moon2012bayesian}, \cite{giacomini2021robust}, and \cite{christensen2026optimal}.\footnote{There is also a related  literature in statistics on nonparametric Bayesian level set estimation. See for example \cite{gayraud2005rates,gayraud2007consistency}'s work on density level sets and  \cite{li2021posterior} who focus on level sets of a general unknown function.}

More broadly, our paper fits within the literature on reduced-form Bayesian inference for econometric models, which refers to situations in which a parameter of interest is viewed as a functional of an identifiable reduced-form parameter (for which a prior is assumed).\footnote{The reduced-form Bayesian literature can be connected to a large literature in statistics on Bayesian estimation of statistical functionals. A nonexhaustive list of contributions includes \cite{ferguson1973bayesian}, \cite{rubin1981bayesian}, \cite{newton1994approximate}, \cite{rivoirard2012bernstein}, \cite{castillo2015bernstein}, \cite{lyddon2019general}, \cite{nickl2020nonparametric}, \cite{monard2021statistical}, \cite{ray2021bernstein}, and \cite{li2026empiricallikelihoodgenerativeai}.} \cite{walker2026semiparametric} considers Bayesian inference for parameters point identified by conditional moment equalities via a nonparametric prior for the conditional density of the endogenous variables given the exogenous variables. Our paper follows a similar thought experiment, except that it assumes a nonparametric prior for a conditional PMF and focuses on identified sets determined by inequality restrictions on the conditional distribution. In this sense, the conceptual relationship between our paper and \cite{walker2026semiparametric} parallels how \cite{kline2016bayesian} connects the seminal inference framework of \cite{chamberlain2003nonparametric} to identified sets. \cite{norets2022adaptiveconditional,norets2022adaptive} develop a nonparametric Bayesian frameworks for reduced-form mixed continuous-discrete distributions, with a key empirical motivation being that firm entry parameters (e.g., \cite{pakes2007simple}) are functionals of these distributions. \cite{ray2018semiparametric}, \cite{breunig2025double,breunig2026semiparametricbayesiandifferenceindifferences}, \cite{ditraglia2025bayesian}, \cite{yiu2025semiparametric}, and \cite{ye2026nonparametric} propose reduced-form Bayesian approaches for treatment effect parameters.

Finally, our paper contributes to a large literature on inference under partial identification. Since we characterize identified sets as level sets of criterion functions, our approach offers a Bayesian analogue to frequentist inference procedures proposed in \cite{manski2002inference}, \cite{chernozhukov2007estimation}, \cite{bugni2010bootstrap}, \cite{romano2010inference}, \cite{menzel2014consistent}, \cite{chernozhukov2015inference}, and \cite{chen2018monte}. Our two leading examples, conditional moment inequalities and linear systems, connects our Bayesian framework to frequentist inference frameworks for these models, such as \cite{khan2009inference}, \cite{andrews2013inference,andrews2014nonparametric,andrews2017inference}, \cite{chernozhukov2013intersection}, \cite{armstrong2014weighted,armstrong2015asymptotically}, \cite{armstrong2016multiscale}, \cite{chetverikov2018adaptive}, \cite{andrews2023inference}, and \cite{chernozhukov2023constrained} for conditional moment inequalities, and \cite{bai2022testing}, \cite{fang2023inference}, \cite{goff2025inference}, and \cite{bai2026inference} for linear systems. There are many other frequentist approaches to inference under partial identification; see \cite{canay2017practical}, \cite{Ho_Rosen_2017}, \cite{MOLINARI2020355}, \cite{KLINE2021345}, and \cite{CANAY2023105558} and the references therein.

\subsection{Outline.} The paper is organized as follows. Section \ref{sec:generalsetup} formalizes the Bayesian inference framework, Section \ref{sec:GPimplementation} puts forward a flexible class of priors for $p$ and illustrates their use in an application to state dependence, Section \ref{sec:posteriorconsistency} presents our posterior consistency results, Section \ref{sec:extensions} offers several extensions of the framework, and Section \ref{sec:conclusion} concludes. Notation is introduced when appropriate and proofs are in the Appendix.

\section{Bayesian Inference Framework}\label{sec:generalsetup}
This section formalizes how Bayesian inferences about identified sets of the form $\Gamma_{I} = \{\gamma: f(X,p(X),\gamma) \geq 0 \ a.s.\}$ can be obtained via nonparametric Bayesian inference for the unrestricted reduced-form conditional PMF $p$ of $Y$ given $X$.
\subsection{Nonparametric Bayesian Inference for PMFs.} Let $\{(Y_{i},X_{i}')'\}_{i\geq 1}$, where $Y_{i} \in \mathcal{Y}$, $\mathcal{Y} = \{y_{1},...,y_{K}\}$, and $X_{i} \in \mathcal{X} \subseteq \mathbb{R}^{d_{x}}$ for each $i$, be a sequence of random vectors for which $\{(Y_{i},X_{i}')'\}_{i=1}^{n}$, $n \geq 1$, forms the observed data. There are no restrictions on $\mathcal{X}$ (i.e., $X_{i}$ can have continuous components). Our sampling model conditions on the realization $\{x_{i}\}_{i \geq 1}$ of $\{X_{i}\}_{i \geq 1}$ and imposes that $Y^{(n)} = (Y_{1},...,Y_{n})$ satisfies 
\begin{align*}
Y^{(n)}|p \sim  P^{(n)}, \quad P^{(n)} = \bigotimes_{i=1}^{n}Categorical(p(x_{i})),
\end{align*}
where $Categorical(p(x))$ denotes a discrete distribution over $\{y_{1},...,y_{K}\}$ with probability mass function $p(x) = (p_{1}(x),...,p_{K}(x))$ for $x \in \mathcal{X}$. For Bayesian inference, we endow the infinite-dimensional space $\mathcal{P}$ of PMFs $p$ with a prior $\Pi$, a conditional (on $\{X_{i}\}_{i \geq 1}$) probability measure over $(\mathcal{P},\mathscr{P})$, where $\mathscr{P}$ is a $\sigma$-algebra over $\mathcal{P}$. Section \ref{sec:GPimplementation} presents an example.
\begin{assumption}\label{as:samplingmodel}
    For each $n \geq 1$ and almost every fixed realization $\{x_{i}\}_{i \geq 1}$ of $\{X_{i}\}_{i \geq 1}$, $(y^{(n)},p) \mapsto L_{n}(p) :=\prod_{i=1}^{n}\prod_{k=1}^{K}p_{k}(x_{i})^{\mathbf{1}\{y_{i}=y_{k}\}}$ is measurable on $(\mathcal{Y}^{n} \times \mathcal{P}, 2^{\mathcal{Y}^{n}}\otimes\mathscr{P})$.
\end{assumption}
\begin{assumption}\label{as:marginallikelihood}
The following conditions hold:
\begin{enumerate}
\item For each $n \geq 1$ and almost every fixed realization $\{x_{i}\}_{i\geq 1}$ of $\{X_{i}\}_{i \geq 1}$, the marginal distribution $\Pi_{n}$ of $(p(x_{1}),...,p(x_{n}))$ under $\Pi$ admits a density $\pi_{n}$.
\item Let $p_{0}$ denote the true conditional PMF. For each $n \geq 1$ and almost every fixed realization $\{x_{i}\}_{i \geq 1}$ of $\{X_{i}\}_{i \geq 1}$, $\int_{\mathcal{P}}L_{n}(p)\pi_{n}(p)dp \in (0,\infty)$ a.e. $[P_{0}^{(n)}]$.
\end{enumerate}
\end{assumption}
Assumptions \ref{as:samplingmodel} and \ref{as:marginallikelihood} are standard regularity conditions \citep{ghosal2017fundamentals}, and imply that the vector of in-sample PMFs $(p(x_{1}),...,p(x_{n}))$ has a well-defined posterior distribution $\Pi_{n}((p(x_{1}),...,p(x_{n})) \in \cdot | Y^{(n)})$ with density that satisfies
\begin{align*}
\pi_{n}(p(x_{1}),...,p(x_{n})|Y^{(n)}) \propto L_{n}(p)\pi_{n}(p(x_{1}),...,p(x_{n})).    
\end{align*}
for each $p \in \mathcal{P}$. The posterior contains a full inference theory for the vector $(p(x_{1}),...,p(x_{n}))$ and any transformation of it. Consequently, by treating the identified set as a function of this vector, we obtain a Bayesian inference theory for the identified set.

\subsection{Bayesian Inference for Identified Sets.} We now show how $\pi_{n}(\cdot|Y^{(n)})$ leads to a Bayesian inference framework for an identified set. Let $\Gamma \subseteq \mathbb{R}^{d_{\gamma}}$, $d_{\gamma} < \infty$, be the structural parameter space, and, for a fixed realization $\{x_{i}\}_{i \geq 1}$ of $\{X_{i}\}_{i \geq 1}$, define the \textit{conditional identified set} $\Gamma_{n,I}(p)$ at $p$ as 
\begin{align*}
\Gamma_{n,I}(p) = \{\gamma \in \Gamma: f(x_{i},p(x_{i}),\gamma) \in \mathbb{R}_{+}^{d_{f}} \ \forall \ i=1,...,n\},    
\end{align*}
where $f: \mathcal{X} \times \Delta^{K-1} \times \Gamma \rightarrow \mathbb{R}^{d_{f}}$, $d_{f} < \infty$, is a known function.\footnote{Theorem \ref{thm:fullsupport} resolves the apparent discrepancy between $\Gamma_{n,I}$ and $\Gamma_{I}$.} The conditional identified set can be equivalently expressed as the level set $\Gamma_{n,I}(p) = \{\gamma \in \Gamma: Q_{n}(\gamma,p) = 0\}$, where $Q_{n}: \Gamma \times \mathcal{P} \rightarrow \mathbb{R}_{+}$ is any criterion function that only depends on $p$ through $(p(x_{1}),...,p(x_{n}))$, and, importantly, satisfies the property\footnote{Section \ref{sec:examples} demonstrates $Q_{n}$ may depend on $\{x_{i}\}_{i=1}^{n}$, but we suppress for ease of exposition.} 
\begin{align*}
Q_{n}(\gamma,p) = 0 \iff f(x_{i},p(x_{i}),\gamma)\in \mathbb{R}_{+}^{d_{f}} \ \forall \ i=1,...,n .   
\end{align*}
Crucially, since $Q_{n}$ only depends on the in-sample response probabilities $(p(x_{1}),...,p(x_{n}))$, posterior probability statements about $(p(x_{1}),...,p(x_{n}))$ translate to statements about $\Gamma_{n,I}(p)$, provided that the correspondence $\Gamma_{n,I}: \mathcal{P} \rightarrow 2^{\Gamma}$ is suitably measurable.

For this purpose, we maintain that $\Gamma_{n,I}: \mathcal{P} \rightarrow 2^{\Gamma}$ forms a random closed set \citep{molchanov2017theory,molchanov2018random}. That is, $\Gamma_{n,I}(p)$ is a closed set for each $p$ and $\Gamma_{n,I}^{-}(A) := \{p: \Gamma_{n,I}(p) \cap A \neq \emptyset\} \in \mathscr{P}$ for each compact subset $A$ of $\Gamma$. In such case, the probability law of $\Gamma_{n,I}|Y^{(n)}$ is uniquely determined by the \textit{posterior capacity functional} $T_{\Gamma_{n,I}|Y^{(n)}}$, which satisfies $T_{\Gamma_{n,I}|Y^{(n)}}(A) = \Pi_{n}(\Gamma_{n,I}(p) \cap A \neq \emptyset | Y^{(n)})$ for $A \subseteq \Gamma$ compact. The posterior capacity functional formalizes $\Gamma_{I}|Y,X$. Assumptions \ref{as:parameterspace} and \ref{as:criterion}, both standard, are sufficient for this characterization (with Theorem \ref{thm:randomset} serving as validation).\footnote{Two comments. First, compactness of $\Gamma$ is not used in the proof of Theorem \ref{thm:randomset}, however, for the application of random set theory, the ambient space (i.e., $\Gamma$) is typically restricted to be locally compact second countable Hausdorff (LCSCH) (see \cite{molchanov2017theory,molchanov2018random}). Assumption \ref{as:parameterspace} is sufficient for this and will be used in the subsequent results of the paper, so we impose it outright but explicitly acknowledge that it is not necessary for Theorem \ref{thm:randomset} (i.e, we can impose that $\Gamma$ is LCSCH only). Second, if Assumption \ref{as:criterion}.1 is strengthened to continuity and Assumption \ref{as:parameterspace} holds, then Assumption \ref{as:criterion}.2 holds if $Q_{n}(\gamma,\cdot): \mathcal{P} \rightarrow \mathbb{R}_{+}$ is $\mathscr{P}$-measurable for each $\gamma \in \Gamma$ (see, for example, Theorem 17.18 in \cite{aliprantis1999infinite}). For sufficient conditions based on lower semicontinuity, see \cite{stinchcombe1992some}.}
\begin{assumption}\label{as:parameterspace}
$\Gamma$ is a compact subset of $(\mathbb{R}^{d_{\gamma}},||\cdot||_{2})$, where $||\cdot||_{2}$ is the Euclidean norm.
\end{assumption}
\begin{assumption}\label{as:criterion}
For each $n \geq 1$ and almost every fixed realization $\{x_{i}\}_{i \geq 1}$ of $\{X_{i}\}_{i \geq 1}$,
\begin{enumerate}
    \item The mapping $\gamma \mapsto Q_{n}(\gamma,p)$ is lower semicontinuous for each $p \in \mathcal{P}$.
    \item The mapping $p \mapsto \inf_{\gamma \in A}Q_{n}(\gamma,p)$ is measurable function from $(\mathcal{P},\mathscr{P})$ to $(\mathbb{R}_{+},\mathcal{B}(\mathbb{R}_{+}))$ for each compact $A \subseteq \Gamma$.\footnote{The notation $\mathcal{B}(A)$ denotes the Borel $\sigma$-algebra generated by the open subsets of $A$.}
\end{enumerate}   
\end{assumption}
\begin{theorem}\label{thm:randomset}
Suppose that Assumptions \ref{as:parameterspace} and \ref{as:criterion} hold. Then, for each $n \geq 1$ and almost every fixed realization $\{x_{i}\}_{i \geq 1}$ of $\{X_{i}\}_{i \geq 1}$, $\Gamma_{n,I}$ defines a random closed set over $(\mathcal{P},\mathscr{P})$.
\end{theorem}

There are several aspects of our framework that should be emphasized. First, the posterior for the identified set implies a marginal posterior for the identified set of functions of the structural parameter. Let $g: \mathbb{R}^{d_{\gamma}} \rightarrow \mathbb{R}^{d_{h}}$, $d_{g} < \infty$, be a known continuous transformation. The conditional identified set for $g(\gamma)$ is $G_{n,I}(p) = \{g(\gamma): \gamma \in \Gamma_{n,I}(p)\}$. Corollary \ref{cor:functionrandomset} shows that $G_{n,I}: \mathcal{P} \rightarrow 2^{g(\Gamma)}$ forms a random closed set. Consequently, $\pi_{n}(\cdot|Y^{(n)})$ leads to a posterior for $G_{n,I}(p)$ that is similarly described its capacity functional. Defining $g$ as the coordinate projection, this nests the identified set for subvectors of $\gamma$.
\begin{corollary}\label{cor:functionrandomset}
Suppose Assumptions \ref{as:parameterspace} and \ref{as:criterion} hold, and $g: \Gamma \rightarrow \mathbb{R}^{d_{g}}$ is continuous. Then, for each $n \geq 1$ and almost every fixed realization $\{x_{i}\}_{i \geq 1}$ of $\{X_{i}\}_{i \geq 1}$, $G_{n,I}$ defines a random closed set over $(\mathcal{P},\mathscr{P})$.
\end{corollary}

Second, our framework also accommodates model misspecification. A partially identified model is misspecified at $p \in \mathcal{P}$ if $\Gamma_{n,I}(p) = \emptyset$. Setting $A = \Gamma$, $\{p: \Gamma_{n,I}(p)= \emptyset\} $ is $\mathscr{P}$-measurable and  $1-T_{\Gamma_{n,I}|Y^{(n)}}(\Gamma)$ is a misspecification diagnostic because a high posterior probability of an empty identified set provides strong evidence that the identifying restrictions are not compatible with the data. Alternatively, one can report the posterior for a pseudo-identified set $\tilde{\Gamma}_{n,I}(p) = \{\gamma \in \Gamma: \tilde{Q}_{n}(\gamma,p) = 0\}$, where $\tilde{Q}_{n}(\gamma,p) = Q_{n}(\gamma,p) - \inf_{\tilde{\gamma} \in \Gamma}Q_{n}(\tilde{\gamma},p)$ (i.e., $\tilde{\Gamma}_{n,I}(p)$ is the set of minimizers of $Q_{n}(\cdot,p)$). Like $\Gamma_{n,I}(p)$, $\tilde{\Gamma}_{n,I}(p)$ is a set-valued functional of $(p(x_{1}),...,p(x_{n}))$ so it has a posterior under $\pi_{n}(\cdot|Y^{(n)})$ as long as it is a random closed set.\footnote{There are, of course, conceptual issues associated with reporting pseudo-true parameters \citep{white1982maximum,muller2013risk,hansen2021inference,andrews2026true}. Our paper does not take a stance on whether one \textit{should} report $\tilde{\Gamma}_{n,I}(p)$, but simply provides it as an option for those concerned about misspecification at $p$.} Similarly, a valid inference framework for functions of $\gamma$ is obtained via the posterior for $\tilde{G}_{n,I}(p) = \{g(\gamma): \gamma \in \tilde{\Gamma}_{n,I}(p)\}$.
\begin{theorem}\label{thm:pseudoset}
Suppose that Assumptions \ref{as:parameterspace} and \ref{as:criterion} hold. Then, for each $n \geq 1$ and almost every fixed realization $\{x_{i}\}_{i \geq 1}$ of $\{X_{i}\}_{i \geq 1}$, $\tilde{\Gamma}_{n,I}$ defines a random closed set over $(\mathcal{P},\mathscr{P})$.
\end{theorem}
Finally, most literature on partial identification with covariates assumes that $\{(Y_{i},X_{i}')'\}_{i \geq 1}$ forms an independent and identically distributed (i.i.d) sequence, and defines the identified set at $p$ as $\Gamma_{I}(p) = \{\gamma \in \Gamma: f(x,p(x),\gamma) \geq 0 \ \forall \ x \in \mathcal{S}_{X}\}$ to define the identified set at $p$, where $\mathcal{S}_{X}$ is the support of the marginal distribution $P_{0,X}$ of $X$ (see, for example, \cite{Ho_Rosen_2017}, \cite{MOLINARI2020355}, \cite{KLINE2021345}, and \cite{CANAY2023105558}). This creates an apparent discrepancy between our target parameter $\Gamma_{n,I}(p)$ and the focus of much of the partial identification literature (an exception being \cite{rosen2025finite}, who define the identified set for the \cite{manski1975maximum,manski1985semiparametric} maximum score model conditionally). Fortunately, the next result establishes that, under a well-separatedness condition for $\Gamma_{I}(p)$ and other standard regularity conditions, there is virtually no difference between $\Gamma_{n,I}(p)$ and $\Gamma_{I}(p)$ for large $n$. For notation, $f_{j}(x,p(x),\gamma)$ is the $j$th element of $f(x,p(x),\gamma)$, $\mathscr{X}$ is the $\sigma$-algebra over $\mathcal{X}$ for which $P_{0,X}$ is defined, $\mathscr{X}^{\infty}$ is the associated product $\sigma$-algebra, and $d_{\mathcal{H}}$ is the Hausdorff distance.\footnote{Given two sets $\Gamma_{1},\Gamma_{2}$ in $(\mathbb{R}^{d_{\gamma}},||\cdot||_{2})$, the Hausdorff distance is defined as $d_{\mathcal{H}}(\Gamma_{1},\Gamma_{2}) = \max\{\sup_{\gamma \in \Gamma_{1}}\inf_{\tilde{\gamma} \in \Gamma_{2}}||\gamma-\tilde{\gamma}||_{2},\sup_{\gamma \in \Gamma_{2}}\inf_{\tilde{\gamma} \in \Gamma_{1}}||\gamma-\tilde{\gamma}||_{2}\}$.}
\begin{theorem}\label{thm:fullsupport}
Suppose that the following hold for each $p \in \mathcal{P}$: 1. $\gamma \mapsto Q_{n}(\gamma,p)$ is lower semicontinuous for each $\{x_{i}\}_{i  \geq 1} \in \mathcal{X}^{\infty}$ and each $n \geq 1$, 2. $\{x_{i}\}_{i \geq 1} \mapsto \inf_{\gamma \in A}Q_{n}(\gamma,p)$ is a measurable function from $(\mathcal{X}^{\infty},\mathscr{X}^{\infty})$ to $(\mathbb{R}_{+},\mathcal{B}(\mathbb{R}_{+}))$ for each compact $A \subseteq \Gamma$ and each $n \geq 1$, 3. $\Gamma$ is a compact subset of $(\mathbb{R}^{d_{\gamma}},||\cdot||_{2})$, 4. the sequence $\{X_{i}\}_{i\geq 1}$ is i.i.d with common distribution $P_{0,X}$, 5. the mapping $x \mapsto \sup_{\gamma \in A}\min_{1 \leq j \leq d_{f}}f_{j}(x,p(x),\gamma) $ is a measurable function from $(\mathcal{X},\mathscr{X})$ to $(\mathbb{R},\mathcal{B}(\mathbb{R}))$ for every compact $A \subseteq \Gamma$, 6. $\Gamma_{I}(p) \neq \emptyset$, and 7. for each compact $A\subseteq \Gamma \setminus \Gamma_{I}(p)$, $P_{0,X}(\sup_{\gamma \in A}\min_{1 \leq j \leq d_{f}}f_{j}(X,p(X),\gamma) < 0)>0$. Then for each $p \in \mathcal{P}$: $\lim_{n\rightarrow \infty}d_{\mathcal{H}}\left(\Gamma_{n,I}(p),\Gamma_{I}(p)\right) = 0$ a.e. $[P_{0,X}^{(\infty)}]$, where $P_{0,X}^{(\infty)}= \bigotimes_{i \geq 1}P_{0,X}$.
\end{theorem}
\subsection{Examples.}\label{sec:examples}
This section revisits Examples \ref{ex:momentsineq} and \ref{ex:linsyst}. Let $||\cdot||_{n,r}$ satisfy $||h||_{n,r}= (n^{-1}\sum_{i=1}^{n}|h(x_{i})|^{r})^{1/r}$ for $r \in [1,\infty)$ and $||h||_{n,\infty} = \max_{1 \leq i \leq n}|h(x_{i})|$ for real-valued $h$ (replacing $|h(x)|$ with $||h(x)||_{2}$ and $||h(x)||_{op}$ for vector valued and matrix valued $h$, respectively, where $||\cdot||_{op}$ is the operator norm). Let $\text{dist}_{W}(\cdot,A)$ satisfy $\text{dist}_{W}(t,A) = \inf_{\tilde{t} \in A}||t-\tilde{t}||_{W}$ for each $t$, where $||t||_{W}^{2} = t'Wt$ for some symmetric, positive-definite matrix $W$.
\setcounter{example}{0}
\begin{example}[Continued]
Recall that, for CMIs, $f(x_{i},p(x_{i}),\gamma) = E[g(Y,X,\gamma)|X=x_{i}]$, where $g(Y,X,\gamma)$ be a $d_{g} \times 1$ vector of known functions and the expectation is taken with respect to $p(x_{i})$. A class of criterion functions for CMIs are of the mimimum distance form,
\begin{align*}
 Q_{n,r}(\gamma,p) = \left|\left|\text{dist}_{W(\cdot,p(\cdot),\gamma)}(f(\cdot,p(\cdot),\gamma),\mathbb{R}_{+}^{d_{g}})\right| \right|_{n,r},   
\end{align*} where $W(x_{i},p(x_{i}),\gamma)$ is some symmetric, positive-definite matrix and $r \in [1,\infty]$. Since $\{Q_{n}(\gamma,p): \gamma \in \Gamma\}$ is fully determined by $(p(x_{1}),...,p(x_{n}))$, the posterior $\pi_{n}(\cdot|Y^{(n)})$ leads to a posterior over the values of $Q_{n}(\gamma,p)$, $\gamma \in \Gamma$, and, by extension, the set of minimizers. These criterions are common in the moment inequality literature (see \cite{canay2017practical} and the references therein). Indeed, setting $r = 2$ and $W(x_{i},p(x_{i}),\gamma) = \Sigma^{-1}(x_{i},p(x_{i}),\gamma)$, where $\Sigma(x_{i},p(x_{i}),\gamma) = Var(g(Y,X,\gamma)|X=x_{i})$, one obtains a quasi-likelihood ratio (QLR) criterion,
\begin{align*}
    &Q_{n,QLR}(\gamma,p) 
=\sqrt{\frac{1}{n}\sum_{i=1}^{n}\min_{t \in \mathbb{R}^{d_{g}}_{+}}\left\{(f(x_{i},p(x_{i}),\gamma)-t)'\Sigma^{-1}(x_{i},\gamma)(f(x_{i},p(x_{i}),\gamma)-t)\right\}}
\end{align*}while setting $r = 2$ and $W(x_{i},p(x_{i}),\gamma) = D^{-1}(x_{i},p(x_{i}),\gamma)$, with $D(x_{i},p(x_{i}),\gamma) = \text{diag}\Sigma(x_{i},p(x_{i}),\gamma)$, yields a modified method of moments (MMM) criterion
\begin{align*}
 Q_{n,MMM}(\gamma,p) =    \sqrt{\frac{1}{n}\sum_{i=1}^{n}||(D^{-1/2}(x_{i},p(x_{i})\gamma)f(x_{i},p(x_{i}),\gamma))_{-}||_{2}^{2}},
\end{align*}
where $(t)_{-} = (\min\{t_{1},0\},...,\min\{t_{d},0\})'$ for $t \in \mathbb{R}^{d}$. In specific relation to the frequentist conditional moment inequality literature (e.g., \cite{andrews2013inference} and \cite{armstrong2014weighted}), $Q_{n,r}(\gamma,p)$ with $r \in [1,\infty)$ and $r= \infty$ can be viewed as Cr\'{a}mer-von Mises-type (CvM) and Kolmogorov-Smirnov-type (KS) statistics, respectively. A key distinction is that we aggregate over the observed covariates $\{x_{i}\}_{i=1}^{n}$, whereas the papers listed above aggregate over a user-specified class of implied unconditional moment restrictions. Consequently, applied researchers willing to adopt a Bayesian perspective can avoid selecting unconditional moments.
\end{example}
\begin{example}[Continued]\label{ex:linsyst}
Recall that, for linear systems, $f(x_{i},p(x_{i}),\gamma) = A(x_{i},p(x_{i}))\gamma - b(x_{i},p(x_{i}))$, where $A(x_{i},p(x_{i})) \in \mathbb{R}^{d_{a}\times d_{\gamma}}$ and $b(x_{i},p(x_{i})) \in \mathbb{R}^{d_{a}}$. Like Example \ref{ex:momentsineq}, a possible criterion function is the minimum distance objective
\begin{align*}
Q_{n,r}(\gamma,p) = ||\text{dist}_{W(\cdot,p(\cdot),\gamma)}(A(\cdot,p(\cdot))\gamma-b(\cdot,p(\cdot)),\mathbb{R}_{+}^{d_{a}})||_{n,r},    
\end{align*}
where $W(x_{i},p(x_{i}),\gamma)$ is a symmetric, positive-definite weight matrix and $r \in [1,\infty]$. Alternatively, one could report
\begin{align*}
    Q_{n}^{LP}(\gamma,p) = \frac{1}{n}\sum_{i=1}^{n}||(A(x_{i},p(x_{i}))\gamma-b(x_{i},p(x_{i})))_{-}||_{1},
\end{align*}
where $||\cdot||_{1}$ is the taxicab norm.\footnote{The taxicab norm is $||z||_{1} = \sum_{j=1}^{d_{z}}|z_{j}|$.}$^{,}$\footnote{This statistic is not specific to linear systems because, for any candidate $f$, one can define $Q_{n}^{LP}(\gamma,p) = n^{-1}\sum_{i=1}^{n}||(f(x_{i},p(x_{i}),\gamma))_{-}||_{1}$.} This criterion function is particularly compelling for linear systems because if $\Gamma$ is a convex polytope (i.e., $\Gamma=\{\gamma\in\mathbb R^{d_\gamma}:A_{\Gamma}\gamma\leq b_{\Gamma}\}$ for some matrix $A_{\Gamma}$ and vector $b_{\Gamma}$, and is bounded), then the set of minimizers of $Q_n^{LP}(\cdot,p)$ is the projection onto the $\gamma$-coordinates of the optimal-solution set of the following linear program:\footnote{If $\Gamma$ is merely compact (with a relevant distinguishing case being the spherical normalization $\Gamma =\{\gamma \in \mathbb{R}^{d_{\gamma}}: ||\gamma||_{2}= 1\}$), then the same optimization problem can be written with the constraint $\gamma\in\Gamma$, although it is not necessarily a linear program.}
\begin{align*} 
\min_{\gamma,s_1,\cdots,s_n} & \frac{1}{n}\sum_{i=1}^n\mathbf 1_{d_a}'s_i\\ \text{s.t.}\quad& A(x_i,p(x_i))\gamma+s_i-b(x_i,p(x_i))\geq 0, \ i=1,\ldots,n,\\
&A_{\Gamma}\gamma\leq b_{\Gamma},\quad s_i\geq 0,\ i=1,\ldots,n . \end{align*}
A similar story to Example \ref{ex:momentsineq} applies: $(p(x_{1}),...,p(x_{n}))$ determines $Q_{n}^{LP}(\cdot,p)$ and $Q_{n,r}(\cdot,p)$, meaning that a posterior $\pi_{n}(\cdot|Y^{(n)})$ leads to a posterior for the set of minimizers. Importantly, our approach does not require discretizing the covariates.
\end{example}
\section{Implementation with Gaussian Processes.}\label{sec:GPimplementation}
This section describes implementation for priors based on Gaussian processes (GPs), and presents a Monte Carlo simulation.
\subsection{A Flexible Class of Priors.}\label{sec:GPimplementationdetails}
Our implementation is based on the `stick-breaking' parametrization of a categorical distribution from \cite*{linderman2015dependent}. Define the stick-breaking parametrization of the PMF $p$ as follows,
\begin{align}\label{eq:stickbreak}
p_{1} = q_{1}, \quad p_{k} = q_{k}\prod_{j = 1}^{k-1}(1-q_{j}), \  j=2,...,K-1, \quad p_{K} = \prod_{j=1}^{K-1}(1-q_{j}),
\end{align}
where $q_{k}: \mathcal{X} \rightarrow [0,1]$ for each $k=1,...,K-1$.\footnote{For any discrete distribution $p$, there exists a well-defined stick-breaking representation because $q_{k}= p_{k}/\sum_{j \geq k}p_{j}$ for all $k=1,...,K-1$, so this parametrization does not entail any loss of generality.} By substituting the stick-breaking weights, the likelihood function
\begin{align*}
L_{n}(p) = \prod_{i=1}^{n}\prod_{k=1}^{K}p_{k}(x_{i})^{\mathbf{1}\{Y_{i}=y_{k}\}}    
\end{align*}
satisfies $L_{n}(p) = L_{n}^{*}(q)$, where\footnote{$L_{n}^{*}(q)$ presumes the support of $Y$ has been enumerated (with $y_{k}$ being an enumerated support point) so that the event $Y > y_{k}$ has meaning.}
\begin{align*}
    L_{n}^{*}(q) = \prod_{i=1}^{n}\prod_{k=1}^{K-1}q_{k}(x_{i})^{\mathbf{1}\{Y_{i}=y_{k}\}}(1-q_{k}(x_{i}))^{\mathbf{1}\{Y_{i} > y_{k}\}}.
\end{align*}
Consequently, by specifying a prior for $q := (q_{1},...,q_{K-1})$ and computing the posterior for $(q(x_{1}),...,q(x_{n}))$ using $L_{n}^{*}(q)$, Bayesian inference for $(p(x_{1}),...,p(x_{n}))$ is obtained via the induced probability distribution for $p$ under the stick-breaking parametrization (\ref{eq:stickbreak}).

Since the stick-breaking functions $q$ \textit{only} need to satisfy $0 \leq q_{k} \leq 1$ for each $k=1,...,K-1$, we specify the prior for $q$ as the probability law of $(\Lambda(B_{1}),...,\Lambda(B_{K-1}))$, where $\Lambda(z) = \exp(z)/(1+\exp(z))$ for each $z \in \mathbb{R}$, and
\begin{align*}
B_{k}\overset{ind}{\sim} GP(\mu_{k},\kappa_{k}), \quad k=1,...,K-1,
\end{align*}
with $GP(\mu_{k},\kappa_{k})$ denoting a Gaussian process with mean function $\mu_{k}: \mathcal{X} \rightarrow \mathbb{R}$ and covariance function $\kappa_{k}: \mathcal{X} \times \mathcal{X} \rightarrow \mathbb{R}$.\footnote{The parametrization $\Lambda(B_{k})$ is fully nonparametric because $B_{k} = \Lambda^{-1}(q_{k})$.} Following similar constructions in the conditional density estimation literature (e.g., \cite{ren2011logistic}), we call this the logistic stick-breaking Gaussian process prior, and use $LSBGP(\{\mu_{k}\}_{k=1}^{K-1},\{\kappa_{k}\}_{k=1}^{K-1})$ to denote this class of priors. The specification of $\kappa_{k}$ allows for flexible incorporation of nonparametric properties like smoothness, while the choice of $\mu_{k}$ enables shrinkage towards a particular parametric model for $q$, such as a linear index model \citep{williams2006gaussian}. Moreover, the prior is compatible with our general assumptions. Indeed, under the GP, the vector $(B_{n,1},...,B_{n,K-1})$, with $B_{n,k} :=(B_{k}(x_{i}))_{i=1}^{n}$, satisfies 
\begin{align*}
B_{n,k} \overset{ind}{\sim}\mathcal{N}(\mu_{n,k},\kappa_{n,k}), \quad k=1,...,K-1    
\end{align*} 
where $\mu_{n,k}:=(\mu_{k}(x_{i}))_{i=1}^{n}$ and $\kappa_{n,k}:=(\kappa_{k}(x_{i},x_{j}))_{i,j=1}^{n}$ for $k=1,...,K-1$. This implies that Assumption \ref{as:marginallikelihood}.1 holds because $B \mapsto p$ has a differentiable inverse. In practice, we set $\mu_{k} = 0$ for each $k=1,...,K-1$, so prior elicitation amounts to specifying $\{\kappa_{k}\}_{k=1}^{K-1}$.
\subsection{Posterior Sampling}\label{sec:samplerdescribe} We describe sampling from the posterior for $(B_{n,1},...,B_{n,K-1})$, the quantities relevant for computing $(p(x_{1}),...,p(x_{n}))$. Under the independent GP prior, the posterior density for $(B_{n,1},...,B_{n,K-1})$ satisifes
\begin{align*}
    \pi_{n}(B_{n,1},...,B_{n,K-1}|Y^{(n)}) &\propto \prod_{k=1}^{K-1}\prod_{i \in \mathcal{I}_{k}}\Lambda(B_{k}(x_{i}))^{\mathbf{1}\{Y_{i}=y_{k}\}}(1-\Lambda(B_{k}(x_{i})))^{\mathbf{1}\{Y_{i} > y_{k}\}}   \\
    &\quad \times \prod_{k=1}^{K-1}\pi_{n,k}(B_{n,k}) \\
    &\propto \prod_{k=1}^{K-1}\pi_{n,k}(B_{n,k}|Y^{(n)}),
\end{align*}
where $\pi_{n,k}$ is the density of $\mathcal{N}(\mu_{n,k},\kappa_{n,k})$, $\mathcal{I}_{k} = \{i:Y_{i} \geq y_{k}\}$, and
\begin{align}\label{eq:binarylogit}
   \pi_{n,k}(B_{n,k}|Y^{(n)}) \propto \prod_{i \in \mathcal{I}_{k}}\Lambda(B_{k}(x_{i}))^{\mathbf{1}\{Y_{i}=y_{k}\}}(1-\Lambda(B_{k}(x_{i})))^{\mathbf{1}\{Y_{i} > y_{k}\}}\pi_{n,k}(B_{n,k}).
\end{align}
This means that the blocks $B_{n,1}$,...,$B_{n,K-1}$ are independent under the posterior. Consequently, a draw from the posterior for $(p(x_{1}),...,p(x_{n}))$ is obtained as follows:
\begin{enumerate}
    \item Draw $B_{n,k} \overset{ind}{\sim} \pi_{n,k}(\cdot|Y^{(n)})$ for $k=1,...,K-1$.
    \item Compute $p(x_{i})$ for $i=1,...,n$ using (\ref{eq:stickbreak}).
\end{enumerate}
Given a draw $(p(x_{1}),...,p(x_{n}))$, a draw from $\Gamma_{n,I}|Y^{(n)}$ is acquired by computing $\Gamma_{n,I}(p)$ using \textit{any} appropriate method (and, as a post-processing exercise, can be made parallel). Step 1 can also be parallelized across $k$, which is an attractive feature for settings with large $K$.

We briefly outline sampling from $\pi_{n,k}(\cdot|Y^{(n)})$; Appendix \ref{ap:posteriordraws} contains the precise details. For a given $k$, let $B_{\mathcal{I}_{k},k}=(B_{k}(x_{i}))_{i \in \mathcal{I}_{k}}$ and $B_{\mathcal{I}_{k}^{c},k} = (B_{k}(x_{i}))_{i \in \mathcal{I}_{k}^{c}}$. Partitioning $B_{n,k} $ into $(B_{\mathcal{I}_{k},k},B_{\mathcal{I}_{k}^{c},k})$ and factorizing $\pi_{n,k}(B_{n,k}|Y^{(n)})$ as
\begin{align*}
    \pi_{n,k}\left(B_{n,k}\middle |Y^{(n)}\right) = \pi_{n,k}\left(B_{\mathcal{I}_{k}^{c},k}\middle |Y^{(n)},B_{\mathcal{I}_{k},k}\right) \pi_{n,k}\left(B_{\mathcal{I}_{k},k}\middle | Y^{(n)}\right),
\end{align*}
we sample from $B_{n,k}|Y^{(n)}$ by first sampling from $\pi_{n,k}(B_{\mathcal{I}_{k},k}|Y^{(n)})$ and then sampling from $\pi_{n,k}(B_{\mathcal{I}_{k}^{c},k}|Y^{(n)},B_{\mathcal{I}_{k},k})$. Expression (\ref{eq:binarylogit}) indicates that $\pi_{n,k}(B_{\mathcal{I}_{k},k}|Y^{(n)})$ is the posterior of a nonparametric binary logit model for the subsample $(Y_{i})_{i \in \mathcal{I}_{k}}$. Consequently, P\'{o}lya-Gamma data augmentation \citep{polson2013bayesian} leads to a Gibbs sampler with \textit{closed-form} sweeps. For $\pi_{n}(B_{\mathcal{I}_{k}^{c},k}|Y^{(n)},B_{\mathcal{I}_{k},k})$, the conditional likelihood in (\ref{eq:binarylogit}) is a constant function of $B_{\mathcal{I}_{k}^{c},k}$, which implies that $\pi_{n,k}(B_{\mathcal{I}_{k}^{c},k}|Y^{(n)},B_{\mathcal{I}_{k},k})=\pi_{n}(B_{\mathcal{I}_{k}^{c},k}|B_{\mathcal{I}_{k},k})$. Since $B_{n,k} \sim \mathcal{N}(\mu_{n,k},\kappa_{n,k})$, properties of multivariate normal distributions implies that sampling from $\pi_{n,k}(B_{\mathcal{I}_{k}^{c},k}|Y^{(n)},B_{\mathcal{I}_{k},k})$ simply amounts to generating draws from a Gaussian distribution.
\begin{remark}[Prior Hyperparameters]
The mean and covariance functions often depend on hyperparameters. That is, $\mu_{k}(x)=\mu_{k,\tau_{k,1}}(x)$ and $\kappa_{k}(x,\tilde{x}) = \kappa_{k,\tau_{k,2}}(x,\tilde{x})$ for some $\tau_{k}:=(\tau_{k,1}',\tau_{k,2}')' \in \mathbb{R}^{d_{\tau_{k}}}$. Factorizing the posterior into the product 
\begin{align*}
\pi_{n,k}(B_{n,k},\tau_{k}|Y^{(n)}) = \pi_{n,k}(B_{\mathcal{I}_{k}^{c},k}|Y^{(n)},B_{\mathcal{I}_{k},k},\tau_{k})\pi_{n,k}(B_{\mathcal{I}_{k},k},\tau_{k}|Y^{(n)}),    
\end{align*} 
accommodating hyperparameter selection requires sampling from $\pi_{n,k}(B_{\mathcal{I}_{k},k},\tau_{k}|Y^{(n)})$ and then $\pi_{n,k}(B_{\mathcal{I}_{k}^{c},k}|Y^{(n)},B_{\mathcal{I}_{k},k},\tau_{k})$. The former requires a minor modification to the Gibbs sampler that, importantly, preserves the tractability of P\'{o}lya-Gamma data augmentation (see Appendix \ref{ap:hyperparameters}). Generating draws from $\pi_{n,k}(B_{\mathcal{I}_{k}^{c},k}|Y^{(n)},B_{\mathcal{I}_{k},k},\tau_{k})$ is identical to $\pi_{n,k}(B_{\mathcal{I}_{k}^{c},k}|Y^{(n)},B_{\mathcal{I}_{k},k})$, except that $\mu_{n,k}$ and $\kappa_{n,k}$ are adjusted to reflect the sampled $\tau_{k}$.
\end{remark}
\begin{remark}[Discrete Covariates]\label{remark:discrete}
Suppose that $x_{i} = (x_{i,c}',x_{i,d}')'$, where $x_{i,d} \in \{x_{1},...,x_{G}\}$ with $G < \infty$. In such case, $q_{k}(x_{i}) = \prod_{g=1}^{G}q_{k,g}(x_{i,c})^{\mathbf{1}\{x_{i,d}=x_{g}\}}$ for each $k=1,...,K-1$, and, $L_{n}^{*}(q) =\prod_{i=1}^{n}\prod_{k=1}^{K-1}\prod_{g=1}^{G}q_{k,g}(x_{i,c})^{\mathbf{1}\{x_{i,d}=g,Y_{i}= k\}}(1-q_{k,g}(x_{i,c})^{\mathbf{1}\{x_{i,d}=g,Y_{i} > k\}}$. Consequently, by writing $q_{k,g}(x_{i,c}) = \Lambda(B_{k,g}(x_{i,c}))$ and assuming that $B_{k,g} \overset{ind}{\sim} GP(\mu_{k,g},\kappa_{k,g})$ for each $k,g$, we can apply the same algorithm as above, except now with an additional factorization across discrete covariate values $\{x_{g}\}_{g=1}^{G}$.
\end{remark}
\subsection{Simulation.}
We demonstrate the logistic stick-breaking process in a simulated dynamic panel binary choice model.
\subsubsection{Data-Generating Process and Identified Sets.} The data-generating process imposes that a binary outcome $Y_{t}$ at time $t$ is related to a contemporaneous covariate $X_{t}$ and the previous period's outcome $Y_{t-1}$ via the threshold-crossing model,
\begin{align*}
    Y_{t} = \mathbf{1}\{\beta t + X_{t} + \theta Y_{t-1} \geq \alpha + U_{t}\}, \quad t=1,2
\end{align*}
where $\beta = 1$, $\theta = 0.5$, $\alpha | X_{1},X_{2} \sim \mathcal{N}(0,1)$, $(U_{1},U_{2})|\alpha,X_{1},X_{2} \sim \mathcal{N}(0_{2},\Sigma_{U})$ with correlation matrix $\Sigma_{U}$ having off-diagonal equal to $0.5$, and $Y_{0}|X_{1},X_{2}\sim Bernoulli(0.5)$. We generate covariates according to the rule $X_{1},X_{2} \overset{iid}{\sim}U[-2,2]$. Our simulations focus on the posterior for the conditional identified set of $\gamma := (\beta,\theta)'$ under the stationarity assumption
\begin{align*}
U_{1}|X_{1},X_{2},\alpha \sim U_{2}|X_{1},X_{2},\alpha.    
\end{align*}
Conditional stationarity is a common approach to index coefficient identification in semiparametric panel discrete choice models \citep{manski1987semiparametric,shi2018estimating,khan2021inference,pakes2021unobserved,khan2023identification,pakes2024moment}.

The identified set for $\gamma$ is determined by the conditional probability mass function $p$ of $(Y_{0},Y_{1},Y_{2})'$ given $X = (X_{1},X_{2})$. Specifically, Theorem 1 of \cite{khan2023identification} show that the sharp identified set is characterized by inequalities of the form $p(x) \in R \implies c_{R}(x,\gamma) \geq 0$, where $R$ is some restriction and $c_{R}(x,\gamma)$ is a known transformation of $X$ that depends on $R$. One illustrative inequality is
\begin{align*}
p(Y_{2}=1|X=x) \geq p(Y_{1}=1|X=x) \implies \beta + (x_{2}-x_{1}) + |\theta| \geq 0 
\end{align*}
Hence, the dynamic panel binary choice model falls within our framework by defining $f(x,p(x),\gamma) = (f_{1}(x,p(x),\gamma),...,f_{18}(x,p(x),\gamma))'$, where $f_{j}(x,p(x),\gamma) = \mathbf{1}\{p(x) \in R_{j}\}c_{R_{j}}(x,\gamma)$, the $R_{j}$s are the restrictions that define the identified set, and $J=18$ is the number of restrictions for each $x$.\footnote{Appendix \ref{ap:simulation} lists the full set of inequalities. Importantly, by partitioning the parameter space based on the sign of $\theta$, the identified set can be viewed as the union of convex polytopes (of which fall into Example \ref{ex:linsyst}).} Figure \ref{fig:conditionalidset} presents the conditional identified set for $n=100,500,1000,2000$ based on a sequence of covariates $\{x_{i}\}_{i=1}^{n}$ that we keep fixed throughout the simulations.\footnote{Specifically, we first generate $\{(X_{i1},X_{i2})\}_{i=1}^{2000}$ i.i.d $U([-2,2])^{\otimes 2} $ and then compute $\Gamma_{n,I}(p_{0})$ using the truncated sequence $\{(X_{i1},X_{i2})\}_{i=1}^{n}$. Since the sequence of covariates are i.i.d, Theorem \ref{thm:fullsupport} implies that the conditional identified sets converge to the identified set based on the full support of the uniform distribution. In Appendix \ref{ap:comparisonwithothersets}, we compare the conditional identified set at $n=2000$ with the full support identified set. The latter is modestly smaller.}
\begin{figure}
    \centering
    \includegraphics[width=\linewidth]{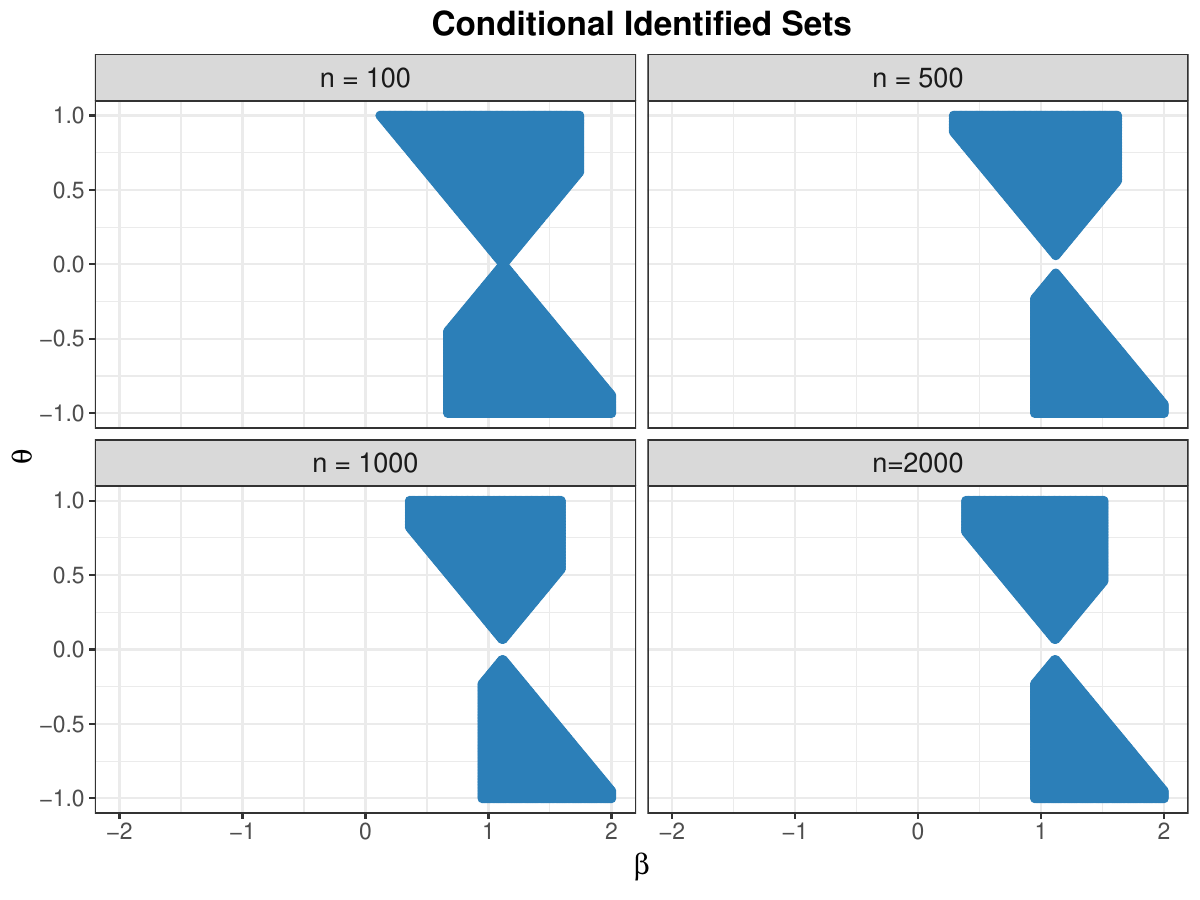}
    \caption{Conditional Identified Sets for Different Sample Sizes}
    \label{fig:conditionalidset}
\end{figure}
\subsubsection{Simulation Results} Conditional on the fixed set of covariates $\{x_{i}\}_{i=1}^{n}$, we simulate $200$ independent panels comprised of $n$ independent cross-sectional units and implement the logistic Gaussian process Gibbs sampler to generate draws from the posterior for the reduced form PMF $p$. Since there are eight possible outcome paths (i.e., $K = |\{0,1\}^{3}| = 8$), we need to specify seven independent GPs to implement the logistic Gaussian process stick-breaking Gibbs sampler. For each $k$, we set $\mu_{k} = 0$ and specify that $\kappa_{k}$ belongs to the Mat\'{e}rn class,
\begin{align}\label{eq:materncovariance}
    \kappa_{k}(x,\tilde{x}) = \frac{2^{1-\alpha_{k}}}{\Gamma(\alpha)}\left(\sqrt{2\alpha_{k}}\frac{||x-\tilde{x}||_{2}}{l_{k}}\right)^{\alpha_{k}}C_{\alpha_{k}}\left(\sqrt{2\alpha_{k}}\frac{||x-\tilde{x}||_{2}}{l_{k}}\right),
\end{align}
where $x = (x_{1},x_{2})'$, $\alpha_{k} > 0$ is a smoothness parameter, $l_{k}>0$ is a length-scale parameter, and $C_{\alpha_{k}}$ is the modified Bessel function of the second kind.\footnote{We standardize the covariates when evaluating $\kappa_{k}$.} The hyperparameter $\alpha_{k}$ indexes smoothness because the sample paths $B$ are $\lfloor\alpha_{k} \rfloor$-times differentiable \citep{JMLR:v12:vandervaart11a}, while the length-scale hyperparameter $l_{k}$ describes the rate at which the correlation between function values decays across the covariate space. We follow the common practice of setting $\alpha_{k} = 1.5$ for each $k=1,...,7$ \citep{williams2006gaussian}. To select the length scale, we update $l_{k}$ during the first $10,000$ iterations of the sampler (assuming $\log l_{k} \overset{iid}{\sim} \log N(0,1)$), and then fix $(l_{1},...,l_{7})$ at the medians of the draws obtained during the second half of this exploration phase.\footnote{A similar approach to length-scale selection is employed in \cite{kankanala2025generalized}.} We perform $20,000$ iterations of the Gibbs sampling algorithm, discarding the first $16,000$ iterations as burn-in and then thinning every four draws. This results in $1,000$ total draws from the posterior.\footnote{Even if updating $l_{k}$ for the entire chain, our prior is specified so that the product posterior from Section \ref{sec:samplerdescribe} arises, which means that if a user's computer has at least $7$ cores, then the posteriors $B_{n,k}|Y^{(n)}$, $k=1,...,7$, can be sampled completely in parallel, effectively only \textit{one} Gibbs sampling chain is performed.}

Using the output from the Gibbs sampler, we compute choice probabilities $(p(x_{1}),...,p(x_{n}))$, and then $\argmin_{\gamma \in [-2,2] \times [-1,1]}Q_{n}^{LP}(\gamma,p)$ to obtain identified set posterior draws, where
\begin{align*}
    Q_{n}^{LP}(\gamma,p) = \frac{1}{n}\sum_{i=1}^{n}\sum_{j=1}^{18}\left|\min\{\mathbf{1}\{p(x_{i}) \in R_{j}\}c_{R_{j}}(x_{i},\gamma),0\}\right|,
\end{align*}
In other words, we report the posterior for $\tilde{\Gamma}_{n,I}(p)=\{\gamma \in \Gamma: \tilde{Q}_{n}^{L}(\gamma,p) = 0\}$. The set of minimizers is found by evaluating $Q_{n}^{LP}(\gamma,p)$ over a grid in $[-2,2]\times [-1,1]$ with step size $10^{-2}$. Figure \ref{fig:inclusionprobabilities} plot the posterior probabilities that $\beta \in \tilde{B}_{n,I}(p)$ and $\theta \in \tilde{\Theta}_{n,I}(p)$ (averaged over the $200$ simulated datasets), where $\tilde{B}_{n,I}(p)$ and $\tilde{\Theta}_{n,I}(p)$ are the coordinate projections of $\tilde{\Gamma}_{n,I}(p)$ in the $\beta$ and $\theta$ directions, respectively. Comparing this with Figure \ref{fig:conditionalidset}, the identified set places increasing probability on the elements of the conditional identified set as $n$ grows. Table \ref{tab:mad} presents the mean absolute deviation (MAD) of $E[d_{\mathcal{H}}(\tilde{\Gamma}_{n,I}(p),\Gamma_{n,I}(p_{0}))|Y^{(n)}]$ from zero, which, by Markov's inequality, summarizes the amount of mass the posterior for $\tilde{\Gamma}_{n,I}(p)$ places on $d_{\mathcal{H}}$ neighborhoods of $\Gamma_{n,I}(p_{0})$. The results indicate that the posterior becomes more concentrated around $\Gamma_{n,I}(p_{0})$ as $n$ grows.
\begin{figure}[h!]
  \centering
    \includegraphics[width=\linewidth]{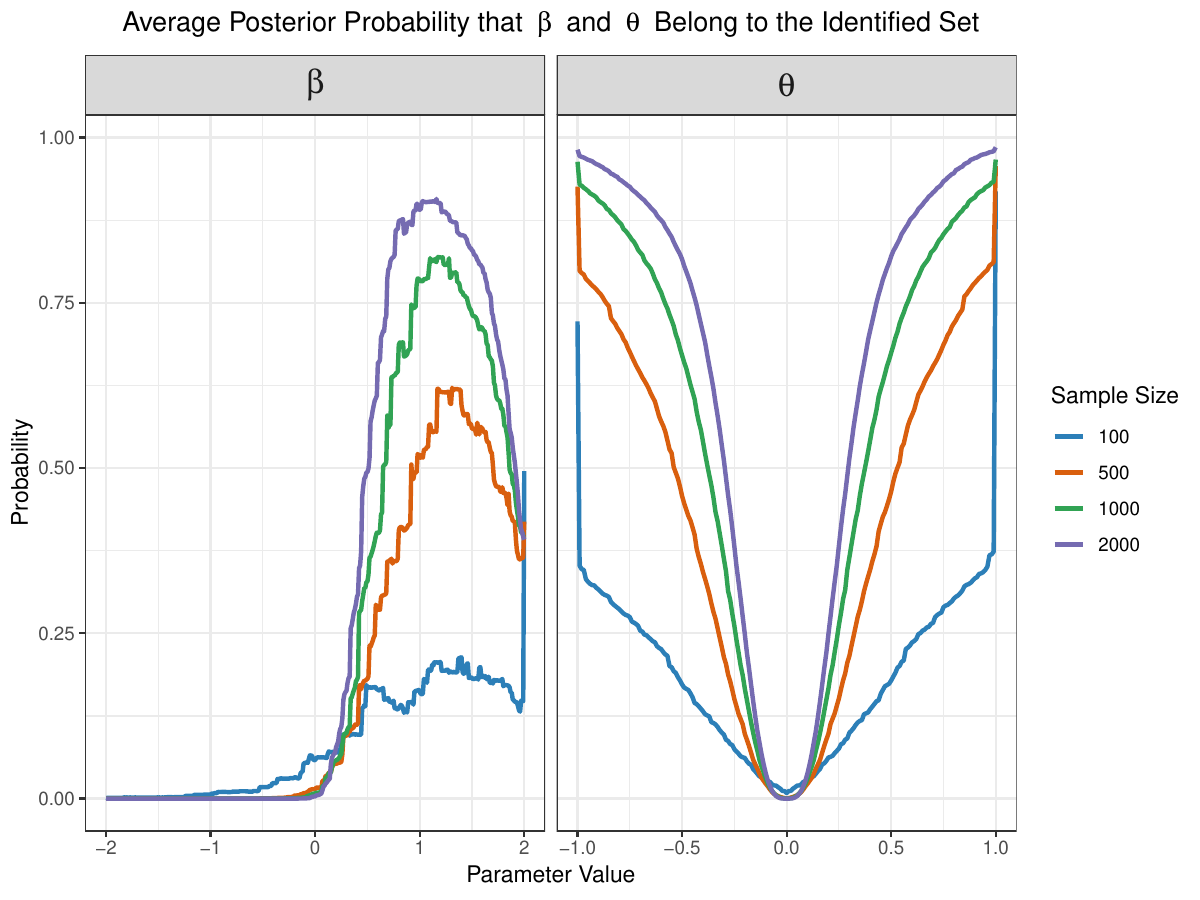}
    \caption{Posterior Inclusion Probabilities}
    \label{fig:inclusionprobabilities}
\end{figure}
\begin{table}[h!]
\centering
\label{tab:hausdorff}
\begin{tabular}{lllll}
& $n=100$ & $n=500$ & $n=1000$ & $n=2000$ \\ \cline{2-5}
\multicolumn{1}{l|}{MAD} &   $1.733$      & $0.995$        &    $0.705$      &  $0.548$\\
 & & & &
\end{tabular}
\caption{MAD for Different Sample Sizes}\label{tab:mad}
\end{table}

\section{Posterior Consistency for the Identified Set}\label{sec:posteriorconsistency}
This section derives conditions under which $\Gamma_{n,I}|Y^{(n)}$ concentrates around $\Gamma_{n,I}(p_{0})$ in large samples, where $p_{0}$ is the true conditional distribution of $Y$ given $X$.
\subsection{General Theorem.} This subsection presents a general posterior consistency theorem for $\Gamma_{n,I}|Y^{(n)}$ that applies to any prior $\Pi$ and criterion $Q_{n}$ satisfying some conditions. We start with some assumptions. For notation, let $P_{0,X}^{(\infty)}$ denote the probability law of $\{X_{i}\}_{i \geq 1}$, let $\text{dist}(x,A)$ denote the Euclidean distance between a point $x$ and a set $A$, and let $\Gamma_{n,I}^{\varepsilon} (p)= \{\gamma \in \Gamma: \text{dist}(\gamma,\Gamma_{n,I}(p)) < \varepsilon\}$ for each $\varepsilon > 0$ and $p \in \mathcal{P}$.
\begin{assumption}[DGP]\label{as:dgp}
Conditional on $P_{0,X}^{(\infty)}$-almost every fixed realization $\{x_{i}\}_{i\geq 1}$ of $\{X_{i}\}_{i \geq 1}$, the following hold:
\begin{enumerate}
    \item Independent Outcomes: the sequence $\{Y_{i}\}_{i \geq 1}$ satisfies $Y_{i} \overset{ind}{\sim} Categorical(p_{0}(x_{i}))$ for every $i \geq 1$, where $p_{0} \in \mathcal{P}$.
    \item Correct Specification and Well-Separated Criterion:
    \begin{enumerate}
        \item The true conditional identified set $\Gamma_{n,I}(p_{0})$ satisfies $\Gamma_{n,I}(p_{0}) \neq \emptyset$ for all $n \geq 1$.
        \item For every $\varepsilon >0$, there exists $\xi > 0$ and $N \geq 1$ such that $Q_{n}(\gamma,p_{0}) \geq \xi$ for all $\gamma \in \Gamma \setminus \Gamma_{n,I}^{\varepsilon}(p_{0})$ and all $n \geq N$ and $\{\gamma \in \Gamma: Q_{n}(\gamma,p_{0}) \geq \xi\}$ is compact subset of $\Gamma$ for all $n \geq N$. 
    \end{enumerate}
\end{enumerate}
\end{assumption}
\begin{assumption}[Prior]\label{as:prior}
For $P_{0,X}^{(\infty)}$-almost every fixed realization $\{x_{i}\}_{i \geq 1}$ of $\{X_{i}\}_{ i\geq 1}$, there is a sequence of semimetrics $\{d_{n}\}_{n \geq 1}$ over $\mathcal{P}$, sequences $\{\delta_{n}\}_{n \geq 1}$ with $\delta_{n} \downarrow 0$ as $n\rightarrow \infty$, and sieves $\{\mathcal{P}_{n}\}_{n \geq 1} \subseteq \mathcal{P}$ such that
\begin{enumerate}
    \item \textit{Posterior Concentration:}
    \begin{enumerate}
            \item $\Pi_{n}(p \in \mathcal{P}_{n}|Y^{(n)}) \rightarrow 1$ in $P^{(n)}_{0}$-probability as $n\rightarrow \infty$.
        \item $\Pi_{n}(d_{n}(p,p_{0})< \delta_{n}|Y^{(n)}) \rightarrow 1$ in $P^{(n)}_{0}$-probability as $n\rightarrow \infty$.
    \end{enumerate}
    \item \textit{Uniform Convergence:} the sequence of criterion function $\{Q_{n}\}_{n \geq 1}$ is such that
    \begin{align*}
    \limsup_{n\rightarrow \infty}\sup_{(\gamma,p) \in \Gamma \times V_{n,\delta_{n}}}|Q_{n}(\gamma,p)-Q_{n}(\gamma,p_{0})| =0,
    \end{align*}
    where $V_{n,\delta_{n}} = \mathcal{P}_{n} \cap \{p: d_{n}(p,p_{0}) < \delta_{n}\}$.
  \item \textit{Lower Hemicontinuity:} for every $\varepsilon >0$, there exists $N \geq 1$ such that $p \in V_{n,\delta_{n}}$ and $n \geq N$ implies $\Gamma_{n,I}(p) \cap B(\gamma,\varepsilon) \neq \emptyset$ for all $\gamma \in \Gamma_{n,I}(p_{0})$.
\end{enumerate} 
\end{assumption}
We discuss the assumptions. Assumption \ref{as:dgp}.1 requires that the conditional sampling model be correctly specified. Assumption \ref{as:dgp}.2 imposes that the identified set $\Gamma_{n,I}(p_{0})$ is nonempty and, for large $n$, forms a well-separated set of minimizers of the criterion function $Q_{n}(\cdot,p_{0})$.\footnote{The compactness of $\{\gamma \in \Gamma: Q_{n}(\gamma,p_{0}) \geq \xi\}$ is a technical measurability condition. It holds, for instance, if lower semicontinuity of $Q_{n}(\gamma,p_{0})$ is strengthened to continuity in $\gamma$.} Separation conditions are often imposed for consistency of estimators under both point identification \citep{newey1994large} and partial identification \citep{chernozhukov2007estimation}. We treat misspecification (i.e., $\Gamma_{n,I}(p_{0}) = \emptyset$) momentarily.  Assumption \ref{as:prior}.1 requires that the posterior concentrates on (possibly) infinite-dimensional sieves $\mathcal{P}_{n}$ and neighborhoods of $p_{0}$ with respect to some semimetric $d_{n}$ at rate $\delta_{n}$. Assumption \ref{as:prior}.2 states the criterion functions $Q_{n}(\gamma,p)$ converges uniformly to $Q_{n}(\gamma,p_{0})$ over $\Gamma \times V_{n,\delta_{n}}$, while Assumption \ref{as:prior}.3 is a lower hemicontinuity-type condition on the correspondence $p \mapsto \Gamma_{n,I}(p)$. Section \ref{sec:suff1} provides low-level sufficient conditions for Assumptions \ref{as:prior}.2 and \ref{as:prior}.3, while Section \ref{sec:GPverification} verifies Assumption \ref{as:prior}.1 for the logistic stick-breaking Gaussian process prior.
\begin{theorem}\label{thm:consistency}
Suppose Assumptions \ref{as:samplingmodel}--\ref{as:prior} hold. Then, for $P_{0,X}^{(\infty)}$-almost every fixed realization $\{x_{i}\}_{i \geq 1}$ of $\{X_{i}\}_{i \geq 1}$ and for every $\varepsilon > 0$,
\begin{align*}
\Pi_{n}\left(d_{\mathcal{H}}(\Gamma_{n,I}(p),\Gamma_{n,I}(p_{0})) \geq \varepsilon \middle|Y^{(n)}\right) \overset{P_{0}^{(n)}}{\longrightarrow} 0,    
\end{align*} 
as $n\rightarrow \infty$, where $d_{\mathcal{H}}$ is the Hausdorff distance and $P_{0}^{(n)} = \bigotimes_{i=1}^{n}Categorical(p_{0}(x_{i}))$.
\end{theorem}

There are two immediate corollaries of Theorem \ref{thm:consistency}. First, the posterior is consistent for the identified set $G_{n,I}(p_{0})$ of continuous transformations of $\gamma$ (of which a special case is the identified set for a subvector of $\gamma$). Second, the posterior is consistent for the full support identified set $\Gamma_{I}(p_{0})$ if, additionally, the conditions of Theorem \ref{thm:fullsupport} hold.
\begin{corollary}\label{cor:subvectors}
Suppose Assumptions \ref{as:samplingmodel} -- \ref{as:prior} and $g: \Gamma \rightarrow \mathbb{R}^{d_{g}}$ is continuous. Then, for $P_{0,X}^{(\infty)}$-almost every fixed realization $\{x_{i}\}_{i \geq 1}$ of $\{X_{i}\}_{i \geq 1}$ and every $\varepsilon > 0$, the following is true:
\begin{align*}
\Pi_{n}\left(d_{\mathcal{H}}(G_{n,I}(p),G_{n,I}(p_{0})) \geq \varepsilon \middle|Y^{(n)}\right) \overset{P_{0}^{(n)}}{\longrightarrow} 0,    
\end{align*} 
as $n\rightarrow \infty$, where $d_{\mathcal{H}}$ denotes the Hausdorff distance.
\end{corollary}
\begin{corollary}\label{cor:fullsupportidset}
Suppose Assumptions \ref{as:samplingmodel}--\ref{as:prior} hold and the assumptions of Theorem \ref{thm:fullsupport} hold too. Then, for $P_{0,X}^{(\infty)}$-almost every fixed realization $\{x_{i}\}_{i \geq 1}$ of $\{X_{i}\}_{i \geq 1}$ and every $\varepsilon >0$,
\begin{align*}
    \Pi_{n}\left(d_{\mathcal{H}}\left(\Gamma_{n,I}(p),\Gamma_{I}(p_{0})\right) \geq \varepsilon \middle | Y^{(n)}\right) \overset{P_{0}^{(n)}}{\longrightarrow} 0
\end{align*}
as $n\rightarrow \infty$. If, additionally, $g: \mathbb{R}^{d_{\gamma}}\rightarrow \mathbb{R}$ is continuous, then, for $P_{0,X}^{(\infty)}$-almost almost every fixed realization $\{x_{i}\}_{i \geq 1}$ of $\{X_{i}\}_{i \geq 1}$ and every $\varepsilon > 0$,
\begin{align*}
\Pi_{n}\left(d_{\mathcal{H}}(G_{n,I}(p),G_{I}(p_{0})) \geq \varepsilon \middle|Y^{(n)}\right) \overset{P_{0}^{(n)}}{\longrightarrow} 0,    
\end{align*} 
as $n\rightarrow \infty$, where $d_{\mathcal{H}}$ denotes Hausdorff distance.
\end{corollary}

Theorem \ref{thm:consistency}'s assumptions (or slight modifications of them) also imply some asymptotic properties under misspecification. Specifically, the posterior probability that $\Gamma_{n,I}(p) \neq \emptyset$ is a consistent diagnostic for model misspecification, and the posterior for the pseudo-identified set $\tilde{\Gamma}_{n,I}(p) = \{\gamma \in \Gamma: \tilde{Q}_{n}(\gamma,p)=0\}$ is consistent at $\tilde{\Gamma}_{n,I}(p_{0})$.
\begin{theorem}\label{thm:test}
Suppose Assumptions \ref{as:samplingmodel} -- \ref{as:prior} hold. Then, the following holds for $P_{0,X}^{(\infty)}$-almost every fixed realization $\{x_{i}\}_{i \geq 1}$ of $\{X_{i}\}_{i \geq 1}$,
\begin{enumerate}
    \item If $\liminf_{n\rightarrow \infty}\inf_{\gamma \in \Gamma}Q_{n}(\gamma,p_{0}) > 0$, then $\Pi_{n}(\Gamma_{n,I}(p) = \emptyset | Y^{(n)}) = 1 + o_{P_{0}^{(n)}}(1)$. 
    \item If $\Gamma_{n,I}(p_{0}) \neq \emptyset$ for all $n \geq 1$, then $\Pi_{n}(\Gamma_{n,I}(p) \neq \emptyset | Y^{(n)}) = 1 + o_{P_{0}^{(n)}}(1)$.
\end{enumerate}  
\end{theorem}
\begin{theorem}\label{thm:misspec_consistency}
Suppose Assumptions \ref{as:samplingmodel}--\ref{as:prior} hold, except with the following modifications:
\begin{enumerate}
    \item Modification of Assumption \ref{as:dgp}.2: for every $\varepsilon >0$, there exists $\xi > 0$ and $N \geq 1$ such that $\tilde{Q}_{n}(\gamma,p_{0}) \geq \xi$ for all $\gamma \in \Gamma \setminus \tilde{\Gamma}_{n,I}^{\varepsilon}(p_{0})$ and all $n \geq N$, and $\{\gamma \in \Gamma: \tilde{Q}_{n}(\gamma,p_{0}) \geq \xi\}$ is a compact subset of $\Gamma$ for all $n \geq N$.
    \item Modification of Assumption \ref{as:prior}.3: for every $\varepsilon >0$, there exists $N \geq 1$ such that $p \in V_{n,\delta_{n}}$ and $n \geq N$ implies $\tilde{\Gamma}_{n,I}(p) \cap B(\gamma,\varepsilon) \neq \emptyset$ for all $\gamma \in \tilde{\Gamma}_{n,I}(p_{0})$..
\end{enumerate}
Then, for $P_{0,X}^{(\infty)}$-almost every fixed realization $\{x_{i}\}_{i \geq 1}$ of $\{X_{i}\}_{i \geq 1}$ and every $\varepsilon > 0$, the following is true:
\begin{align*}
\Pi_{n}\left(d_{\mathcal{H}}(\tilde{\Gamma}_{n,I}(p),\tilde{\Gamma}_{n,I}(p_{0})) \geq \varepsilon \middle|Y^{(n)}\right) \overset{P_{0}^{(n)}}{\longrightarrow} 0,    
\end{align*} 
as $n\rightarrow \infty$, where $d_{\mathcal{H}}$ denotes the Hausdorff distance.
\end{theorem}
\subsection{Verifying Assumption \ref{as:prior}.}\label{sec:verifyingconditions}
This section offers sufficient conditions for Assumption \ref{as:prior}. There are three subsections: the first two present low-level sufficient conditions for Assumptions \ref{as:prior}.2 and \ref{as:prior}.3 (and its misspecified counterpart). Then, noticing a common uniform convergence requirement, we verify uniform convergence for the GP priors from Section \ref{sec:GPimplementation} in the third subsection to find a sufficient condition for the existence of sieved neighborhoods $\{V_{n,\delta_{n}}\}_{n \geq 1}$.
\subsubsection{Sufficient Conditions for Assumption \ref{as:prior}.2}\label{sec:suff1} The first proposition gives sufficient conditions for uniform convergence of $Q_{n}(\gamma,p)$ (Assumption \ref{as:prior}.2) in the case where 
\begin{align*}
Q_{n}(\gamma,p) = ||\text{dist}_{W(\cdot,p(\cdot),\gamma)}(f(\cdot,p(\cdot),\gamma),\mathbb{R}_{+}^{d_{f}})||_{n,r}    
\end{align*} 
for some $r \in [1,\infty]$. For notation, let $\lambda_{min}(A)$ and $\lambda_{max}(A)$ denote the minimum and maximum eigenvalues, respectively, of a matrix $A$.
\begin{proposition}\label{prop:ucon}
Suppose that, for $P_{0,X}^{(\infty)}$-almost every fixed realization $\{x_{i}\}_{i \geq 1}$ of $\{X_{i}\}_{i \geq 1}$, the following conditions hold:
\begin{enumerate}
\item Bounded criterion:
\begin{align*}
    \sup_{\gamma \in \Gamma}||\text{dist}_{W(\cdot,p_{0}(\cdot),\gamma)}(f(\cdot,p_{0}(\cdot),\gamma),\mathbb{R}_{+}^{d_{f}})||_{n,r} < \infty.
\end{align*}
\item Bounded weighting matrix: there exists constants $0 < \underline{\lambda}_{W} \leq \overline{\lambda}_{W} < \infty$ such that
\begin{align*}
  \underline{\lambda}_{W} \leq  \inf_{\gamma \in \Gamma} \inf_{n \geq 1}\min_{1 \leq i \leq n} \lambda_{min}(W(x_{i},p_{0}(x_{i}),\gamma))
  \end{align*}
  and
  \begin{align*}
  \sup_{\gamma \in \Gamma} \sup_{n \geq 1}\max_{1 \leq i \leq n} \lambda_{max}(W(x_{i},p_{0}(x_{i}),\gamma)) \leq \overline{\lambda}_{W}.    
  \end{align*}
\item Uniform convergence: $\{\mathcal{P}_{n}\}_{n \geq 1}$, $\{d_{n}\}_{n \geq 1}$, and $\{\delta_{n}\}_{n \geq 1}$ from Assumption \ref{as:prior} are such that
\begin{align}\label{eq:supf}
\sup_{(\gamma,p) \in \Gamma \times V_{n,\delta_{n}}}\left| \left| f(\cdot,p(\cdot),\gamma)- f(\cdot,p_{0}(\cdot),\gamma) \right| \right|_{n,\infty} = o(1)
\end{align}
and
\begin{align*}
 \sup_{(\gamma,p) \in \Gamma \times V_{n,\delta_{n}}}\left| \left| W(\cdot,p(\cdot),\gamma)-W(\cdot,p_{0}(\cdot),\gamma)\right| \right|_{n,\infty} = o(1)   
\end{align*}
as $n\rightarrow \infty$.
\end{enumerate}
Then Assumption \ref{as:prior}.2 holds for $Q_{n}(\gamma,p) = ||\text{dist}_{W(\cdot,p(\cdot),\gamma)}(f(\cdot,p(\cdot),\gamma),\mathbb{R}_{+}^{d_{f}})||_{n,r}$.
\end{proposition}
Proposition \ref{prop:ucon} states that Assumption \ref{as:prior}.2 holds if $Q_{n}(\cdot,p_{0}): \Gamma \rightarrow \mathbb{R}_{+}$ is bounded, the class of weight matrices $\{W(x_{i},p_{0}(x_{i}),\gamma): 1 \leq i \leq n, \ n \geq 1, \ \gamma \in \Gamma\}$ forms a compact set of positive-definite matrices, and the sieve neighborhoods $V_{n,\delta_{n}}$ are structured enough to ensure convergence of $W(\cdot,p(\cdot),\gamma)$ and $f(\cdot,p(\cdot),\gamma)$ to $W(\cdot,p_{0}(\cdot),\gamma)$ and $f(\cdot,p_{0}(\cdot),\gamma)$, respectively, in the empirical supremum norm, and uniformly over $\Gamma \times V_{n,\delta_{n}}$.\footnote{It should be noted convergence of $f$ in $||\cdot||_{\infty}$ is sufficient for simultaneous control over all $r \in [1,\infty]$, however one can show that Assumption \ref{as:prior}.2 still holds pointwise across $r$ if (\ref{eq:supf}) is weakened to $||\cdot||_{n,r}$; we prefer to state the proposition in terms of $||\cdot||_{n,\infty}$ as this convergence is important for verifying Assumption \ref{as:prior}.3.} The first two conditions hold if $x$ takes values in a compact set, $(x,\gamma) \mapsto f_{0}(x,p_{0}(x),\gamma)$ is continuous, and $(x,\gamma) \mapsto W(x,p_{0}(x),\gamma)$ is continuous and $W(x,p_{0}(x),\gamma)$ is positive-definite for each $(x,\gamma) \in \mathcal{X} \times \Gamma$. Condition 3 is typically satisfied if $V_{n,\delta_{n}}$ is contained in a $||\cdot||_{n,\infty}$-ball centered at $p_{0}$ with shrinking radius.
\setcounter{example}{0}
\begin{example}[Continued]
Suppose that $\mathcal{X}$ is compact and $g(y,x,\gamma)$ is continuous in $(x,\gamma)$ for each $y$ (alternatively, if $g(y,x,\gamma)$ is bounded in $(x,\gamma)$), then, coupled with Assumption \ref{as:parameterspace}, there is a constant $C \in (0,\infty)$ that is independent of $\{x_{i}\}_{i=1}^{n}$ such that
\begin{align*}
||f(\cdot,p(\cdot),\gamma) - f(\cdot,p_{0}(\cdot),\gamma)||_{n,\infty} \leq C||p-p_{0}||_{n,\infty}.     
\end{align*}
Consequently, the first uniform convergence requirement is met if $V_{n,\delta_{n}}$ is contained in an empirical supremum norm ball centered at $p_{0}$ with slowly shrinking radius. Similarly, under positive definiteness of $\{W(\cdot,p_{0}(\cdot),\gamma): \gamma \in \Gamma\}$, the second uniform convergence condition is met for $W(x,p(x),\gamma) = \Sigma^{-1}(x,\gamma)$ and $W(x,p(x),\gamma) = D^{-1}(x,\gamma)$ if $V_{n,\delta_{n}}$ is contained in an empirical supremum norm ball centered at $p_{0}$ with slowly shrinking radius.
\end{example}
\begin{example}[Continued]
Since $\Gamma$ is compact, $\sup_{\gamma \in \Gamma}||f(\cdot,p(\cdot),\gamma) - f(\cdot,p_{0}(\cdot),\gamma)||_{n,\infty}$ is bounded from above by a constant multiple of the maximum of $||A(\cdot,p(\cdot))-A(\cdot,p_{0}(\cdot))||_{n,\infty}$ and $||b(\cdot,p(\cdot))-b(\cdot,p_{0}(\cdot))||_{n,\infty}$. Hence, uniform convergence is satisfied if $A(\cdot,p(\cdot))$ and $b(\cdot,p(\cdot))$ is continuous in $p$ with respect to $||\cdot||_{n,\infty}$ and $V_{n,\delta_{n}}$ is contained in an empirical supremum norm ball centered at $p_{0}$ with slowly shrinking radius. For concreteness, consider \cite{khan2023identification}-type inequalities $\mathbf{1}\{p(x_{i}) \in R\}c_{R}(x_{i})'\gamma \geq 0$ for all $i=1,...,n$ and for all $ R \in \mathcal{R}$, where $R$ is a polyhedral restriction on $p(x)$, $\mathcal{R}$ is a finite set of restrictions $R$, and $c_{R}(x)$ is a known vector. Let $r_{n}\downarrow 0$ as $n\rightarrow \infty$ be such that $V_{n,\delta_{n}} \subseteq \{p: ||p-p_{0}||_{n,\infty} < r_{n}\}$ and let $\rho_{n}:=\min_{1 \leq i \leq n}\min_{R \in \mathcal{R}}\text{dist}(p_{0}(x_{i}),\partial R)$. If there is an $N \geq 1$ such that $r_{n} < \rho_{n}$, then $||A(\cdot,p(\cdot))-A(\cdot,p_{0}(\cdot))||_{n,\infty} = 0$ for all $n \geq N$, which guarantees the desired uniform convergence of $f$ (we can ignore $b(x,p(x))$ as it is zero for all $x$ and $p$).
\end{example}
Since Proposition \ref{prop:ucon} applies to minimum distance criterion functions based on weighted Euclidean norms, it does not immediately apply the criterion $Q_{n}^{LP}(\gamma,p)$ introduced in Example \ref{ex:linsyst}. Fortunately, $Q_{n}^{LP}(\gamma,p)$ satisfies Assumption \ref{as:prior}.2 under empirical $L^{1}$ convergence of $f$ (implied by the same uniform convergence assumption on $f$ as Proposition \ref{prop:ucon}).
\begin{proposition}\label{prop:taxicab}
Suppose that, for $P_{0,X}^{(\infty)}$-almost every fixed realization $\{x_{i}\}_{i \geq 1}$ of $\{X_{i}\}_{i \geq 1}$, $\{\mathcal{P}_{n}\}_{n \geq 1}$, $\{d_{n}\}_{n \geq 1}$, and $\{\delta_{n}\}_{n \geq 1}$ from Assumption \ref{as:prior} are such that
\begin{align*}
\sup_{(\gamma,p) \in \Gamma \times V_{n,\delta_{n}}}\left| \left| f(\cdot,p(\cdot),\gamma)- f(\cdot,p_{0}(\cdot),\gamma) \right| \right|_{n,1} = o(1)
\end{align*}
as $n\rightarrow \infty$. Then Assumption \ref{as:prior}.2 holds for $Q_{n}^{LP}(\gamma,p) = ||(f(\cdot,p(\cdot),\gamma))_{-}||_{n,1}$.
\end{proposition}
\subsubsection{Verifying Assumption \ref{as:prior}.3 and its Misspecified Counterpart}\label{sec:suffhemi} We provide low-level sufficient conditions for Assumption \ref{as:prior}.3. Proposition \ref{prop:hemi} states Assumption \ref{as:prior}.3 holds if, additionally, $\Gamma_{n,I}(p)$ is `strongly identified' for all $p \in V_{n,\delta_{n}}$.
\begin{proposition}\label{prop:hemi}
If, for $P_{0,X}^{(\infty)}$-almost every fixed realization $\{x_{i}\}_{i \geq 1}$ of $\{X_{i}\}_{i \geq 1}$, the sequences $\{\mathcal{P}_{n}\}_{n \geq 1}$, $\{d_{n}\}_{n \geq 1}$, and $\{\delta_{n}\}_{n \geq 1}$ from Assumption \ref{as:prior} are such that
\begin{align*}
\sup_{(\gamma,p) \in \Gamma \times V_{n,\delta_{n}}}\left| \left| f(\cdot,p(\cdot),\gamma)- f(\cdot,p_{0}(\cdot),\gamma) \right| \right|_{n,\infty} = o(1),
\end{align*}
and there is a constant $c \in (0,\infty)$ and $\tilde{N} \geq 1$ such that, for each $\gamma \in \Gamma$,
\begin{align}\label{eq:errorbound}
\max_{1 \leq i \leq n} ||(f(x_{i},p(x_{i}),\gamma))_{-}||_{2} \geq c \text{dist}(\gamma,\Gamma_{n,I}(p))    
\end{align} 
holds for all $p \in V_{n,\delta_{n}}$ and all $n \geq \tilde{N}$, then Assumption \ref{as:prior}.3 holds.
\end{proposition}
Conditions like (\ref{eq:errorbound}) are regularly encountered in the literature on inference for identified sets. For example, it is similar to a linear version of the polynomial minorant condition (4.5) in \cite{chernozhukov2007estimation} and an unstudentized version of Assumption 3 in \cite{kaido2022constraint}. We call it a strong identification condition on $\Gamma_{n,I}(p)$ because it states that, as $\gamma$ moves away from $\Gamma_{n,I}(p)$, the maximal violation of the inequality restrictions grows sufficiently fast. The result readily extends to `union bounds' in which $\Gamma_{n,I}(p) = \bigcup_{j=1}^{J}\Gamma_{n,I}^{(j)}(p)$ with $\Gamma_{n,I}^{(j)} = \{\gamma \in \Gamma^{(j)}: f^{(j)}(x_{i},p(x_{i}),\gamma) \in \mathbb{R}_{+}^{d_{f^{(j)}}} \ \forall \ i=1,...,n\}$ satisfying the conditions of Proposition \ref{prop:hemi}. This result, combined with Proposition \ref{prop:convexhemi} below, is useful for the dynamic binary panel model studied in Section \ref{sec:GPimplementation} (we elaborate on this after stating Proposition \ref{prop:convexhemi}).
\begin{corollary}\label{cor:unions}
    Suppose that 
    \begin{align*}
        \Gamma_{n,I}(p) = \bigcup_{j=1}^{J}\Gamma_{n,I}^{(j)}(p),
    \end{align*} 
    where, for each $j=1,...,J$,
    \begin{align*}
    \Gamma_{n,I}^{(j)} = \left\{\gamma \in \Gamma^{(j)}: f^{(j)}(x_{i},p(x_{i}),\gamma) \in \mathbb{R}_{+}^{d_{f^{(j)}}} \ \forall \ i=1,...,n\right\}    
    \end{align*} satisfies the conditions of Proposition \ref{prop:hemi}. Then Assumption \ref{as:prior}.3 holds.
\end{corollary}

Condition (\ref{eq:errorbound}) (and Corollary \ref{cor:unions}) is still quite high-level, and results from the mathematical optimization literature indicate that its satisfaction may depend subtly on the geometry of $\Gamma_{n,I}$ (see \cite{IOFFE_2016} for a recent survey). For this reason, we provide an interpretable set of conditions for (\ref{eq:errorbound}) for an important class of identified sets.
\begin{proposition}\label{prop:convexhemi}
Suppose that the following conditions hold:
\begin{enumerate}
    \item Parameter Space: $\Gamma$ is a compact and convex subset of $\mathbb{R}^{d_{\gamma}}$.
    \item Smooth concavity: $f(x,p(x),\gamma)$ is differentiable and concave in $\gamma$ for all $(x,p)$.
\end{enumerate}
Further suppose that, for $P_{0,X}^{(\infty)}$-almost every fixed realization $\{x_{i}\}_{i\geq 1}$ of $\{X_{i}\}_{i \geq 1}$, the following conditions hold:
\begin{enumerate}
\setcounter{enumi}{2}
    \item Interior true identified set: there exists $\bar{N}_{1} \geq 1$ and $\xi>0$ such that 
    \begin{align*}
     \Gamma_{n,I}(p_{0}) \subseteq \text{int}(\Gamma)   
    \end{align*} 
    and 
    \begin{align*}
     \sup_{\gamma \in \partial \Gamma}\max_{1\leq i \leq n}\max_{1 \leq j \leq d_{f}}f_{j}(x_{i},p_{0}(x_{i}),\gamma) \leq -\xi   
    \end{align*} 
    for all $n \geq \bar{N}_{1}$.
    \item Slater condition: there exists $\gamma^{\circ} \in \Gamma$, $\eta > 0$, and $\bar{N}_{2} \geq 1$ such that 
    \begin{align*}
    \min_{1\leq i \leq n}\min_{1 \leq j \leq d_{f}}f_{j}(x_{i},p_{0}(x_{i}),\gamma^{\circ}) \geq \eta     
    \end{align*}
    for all $n \geq \bar{N}_{2}$.
    \item Uniform convergence: the sequences $\{\mathcal{P}_{n}\}_{n \geq 1}$, $\{d_{n}\}_{n \geq 1}$, and $\{\delta_{n}\}_{n \geq 1}$ from Assumption \ref{as:prior} are such that
\begin{align*}
\sup_{(\gamma,p) \in \Gamma \times V_{n,\delta_{n}}}\left| \left| f(\cdot,p(\cdot),\gamma)- f(\cdot,p_{0}(\cdot),\gamma) \right| \right|_{n,\infty} = o(1).
\end{align*}
\end{enumerate}
Then Assumption \ref{as:prior}.3 holds.
\end{proposition}
Proposition \ref{prop:convexhemi} gives conditions under which Assumption \ref{as:prior}.3 holds when $\Gamma_{n,I}(p)$ is convex. Convex $\Gamma_{n,I}(p)$ is satisfied by several discrete response models, such as interval censored binary choice models \citep{manski2002inference}, ordered choice models \citep{pakes2015moment}, and panel multinomial choice models \citep{shi2018estimating}. Provided the components $\Gamma_{n,I}^{(j)}(p)$ of $\Gamma_{n,I}(p) = \bigcup_{j=1}^{J}\Gamma_{n,I}^{(j)}(p)$ satisfy Conditions 1--5 of Proposition \ref{prop:convexhemi}, Corollary \ref{cor:unions} implies that it readily extends to identified sets that are unions of convex sets. Related to Section \ref{sec:GPimplementation}, this case conforms to the identification results from \cite{khan2023identification}, where the identified sets are unions of convex polytopes determined by the sign of a state dependence parameter. In terms of the conditions themselves, the key condition on the data-generating process $p_{0}$ in Proposition \ref{prop:convexhemi} is the existence of a uniform slater point $\gamma^{\circ}$ because, combined with the other conditions, it allows us to find a constant $c > 0$ through an asymptotic uniform boundedness (over $\Gamma \times V_{n,\delta_{n}}$) property for the Lagrange multipliers associated with the convex program $\inf_{\tilde{\gamma} \in \Gamma_{n,I}(p)}||\gamma-\tilde{\gamma}||_{2}$.\footnote{The requirements on $p_{0}$ that $\Gamma_{n,I}(p_{0})$ is in the interior of $\Gamma$ is a standard condition (e.g., it is similar to Assumption 1(b) in \cite{kaido2022constraint}), while the separation from the boundary of $\Gamma$ is merely a technical condition that provides enough uniformity to account for $p$ being random.} The uniform convergence for $f(\cdot,p(\cdot),\gamma)$ is the same as Propositions \ref{prop:ucon}-\ref{prop:hemi}, and it is the \textit{only} role for the prior $\Pi$.

Propositions \ref{prop:hemi} and \ref{prop:convexhemi} are catered towards the correctly specified case where $\Gamma_{n,I}(p_{0}) \neq \emptyset$. The next result provides a sufficient condition for the modified Assumption \ref{as:prior}.3 that appears in Theorem \ref{thm:misspec_consistency} (i.e., the posterior consistency for $\tilde{\Gamma}_{n,I}$); it is a counterpart to Proposition \ref{prop:hemi}. \begin{proposition}\label{prop:minorantmisspec}
Suppose that for $P_{0,X}^{(\infty)}$-almost every fixed realization $\{x_{i}\}_{i \geq 1}$ of $\{X_{i}\}_{i \geq 1}$, the following holds:
\begin{align}\label{eq:hemi_misspec_ucon}
  u_{n}:=\sup_{(\gamma,p) \in \Gamma \times V_{n,\delta_{n}}}\left|\tilde{Q}_{n}(\gamma,p)-\tilde{Q}_{n}(\gamma,p_{0})\right| \longrightarrow 0
\end{align}
as $n\rightarrow \infty$. Further, suppose that, for $P_{0,X}^{(\infty)}$-almost every fixed realization $\{x_{i}\}_{i \geq 1}$ of $\{X_{i}\}_{i \geq 1}$, there exists a sequence $\{r_{n}\}_{n \geq 1}$, constants $c,\zeta >0$, and an integer $\tilde{N} \geq 1$ such that $r_{n} \geq 0$ for all $n \geq 1$ and $r_{n}\rightarrow 0$ as $n\rightarrow \infty$, and, for each $\gamma \in \Gamma$,
\begin{align}\label{eq:hemi_misspec_minorant}
    \tilde{Q}_{n}(\gamma,p) \geq c \cdot \text{dist}(\gamma,\tilde{\Gamma}_{n,I}(p))^{\zeta}- r_{n}
\end{align}
for all $p \in V_{n,\delta_{n}}$ and $n \geq \tilde{N}$. Then, the following holds for $P_{0,X}^{(\infty)}$-almost every fixed realization $\{x_{i}\}_{i \geq 1}$ of $\{X_{i}\}_{i \geq 1}$: for every $\varepsilon >0$, there exists $N \geq 1$ such that $p \in V_{n,\delta_{n}}$ and $n \geq N$ implies $\tilde{\Gamma}_{n,I}(p) \cap B(\gamma,\varepsilon) \neq \emptyset$ for all $\gamma \in \tilde{\Gamma}_{n,I}(p_{0})$.
\end{proposition}
Proposition \ref{prop:minorantmisspec} states that the analogue of Assumption \ref{as:prior}.3 for $\tilde{\Gamma}_{n,I}(p)$ holds under a uniform convergence and asymptotic \textit{weak sharp minimum-type property} for $\tilde{Q}_{n}(\gamma,p)$. We label the latter this way because if $\zeta = 1$ and $r_{n} = 0$ for all $n \geq 1$, then it collapses to a stochastic version of the idea introduced in \cite{doi:10.1137/0331063}. We also note that, by allowing $\zeta >0$, the $r_{n} = 0$ case closely relates to the polynomial minorant condition from \cite{chernozhukov2007estimation}. Like Proposition \ref{prop:hemi}, we also provide a set of interpretable low-level sufficient conditions for Proposition \ref{prop:minorantmisspec}.
\begin{proposition}\label{prop:misspec_hemi_convex}
Suppose that the following conditions hold:
\begin{enumerate}
    \item Parameter space: $\Gamma \subseteq \mathbb{R}^{d_{\gamma}}$ is a compact and convex subset of $\mathbb{R}^{d_{\gamma}}$.
    \end{enumerate}
Further, suppose that, for $P_{0,X}^{(\infty)}$-almost every fixed realization $\{x_{i}\}_{i \geq 1}$ of $\{X_{i}\}_{i \geq 1}$, the following conditions hold:
\begin{enumerate}
    \setcounter{enumi}{1}
\item Piecewise Maximum of Smooth Functions: for every $n \geq 1$, there exists a finite index set $\mathcal{A}_{n}$ and functions $\{q_{n,a}: a \in \mathcal{A}_{n}\}$ defined on $\Gamma^{o} \times \mathcal{P}$, where $\Gamma^{o}$ is an open set that contains $\Gamma$, such that $$Q_{n}(\gamma,p) = \max_{a \in \mathcal{A}_{n}}q_{n,a}(\gamma,p)$$ for each $n \geq 1$, and $q_{n,a}(\gamma,p)$ is differentiable and convex in $\gamma$ for each $p$, $a$, and $n$.
\item Uniform convergence: the sequences $\{\mathcal{P}_{n}\}_{n \geq 1}$, $\{d_{n}\}_{n \geq 1}$, and $\{\delta_{n}\}_{n \geq 1}$ from Assumption \ref{as:prior} are such that
\begin{align*}
    \Delta_{n}&:=\sup_{(\gamma,p) \in \Gamma \times V_{n,\delta_{n}}}\max_{a \in \mathcal{A}_{n}}|q_{n,a}(\gamma,p)-q_{n,a}(\gamma,p_{0})| \longrightarrow 0 \\
    \Delta_{n}'&:=\sup_{(\gamma,p) \in \Gamma \times V_{n,\delta_{n}}}\max_{a \in \mathcal{A}_{n}}\left|\left|\nabla_{\gamma} q_{n,a}(\gamma,p)-\nabla_{\gamma}q_{n,a}(\gamma,p_{0})\right| \right|_{2} \longrightarrow 0
\end{align*}
as $n\rightarrow \infty$.
\item Gradient: $\nabla_{\gamma}q_{n,a}(\gamma,p_{0})$ is $L_{0}$-Lipschitz on $\Gamma$, where $L_{0} \in (0,\infty)$ is independent of $n$ and $a$, and there exists $\Gamma_{nbd} \subseteq  \Gamma$ and $\breve{N} \geq 1$ such that $\tilde{\Gamma}_{n,I}(p_{0}) \subseteq \Gamma_{nbd}$ for all $n \geq \breve{N}$  and $\sup_{n \geq \breve{N}}\sup_{\gamma \in \Gamma_{nbd}}\max_{a \in \mathcal{A}_{n}}||\nabla_{\gamma}q_{n,a}(\gamma,p_{0})||_{2} < \infty$.
\item Uniform Nondegeneracy: there exists $\eta > 0$ and $\bar{N} \geq 1$ such that, for all $n \geq \bar{N}$ and all $\tilde{\gamma} \in \tilde{\Gamma}_{n,I}(p_{0})$,
\begin{align*}
   \{t \in \mathbb{R}^{d_{\gamma}}: ||t||_{2} \leq \eta\} \subseteq \conv \left\{\nabla_{\gamma}q_{n,a}(\tilde{\gamma},p_{0}): a \in \mathcal{A}_{n}^{*}(\tilde{\gamma},p_{0})\right\},
\end{align*}
where $\conv$ denotes convex hull and $\mathcal{A}_{n}^{*}(\gamma,p) = \{a \in \mathcal{A}_{n}: q_{n,a}(\gamma,p) = Q_{n}(\gamma,p)\}$.
\end{enumerate}
Then Proposition \ref{prop:minorantmisspec} holds.
\end{proposition}
Some brief remarks on the conditions of Proposition \ref{prop:misspec_hemi_convex}. First, the prior only enters in Condition 3, where it is assumed that $\{V_{n,\delta_{n}}\}_{n \geq 1}$ ensure uniform convergence of the component functions and their gradients. Appendix \ref{ap:uconmisspec} argues that, for some $Q_{n}(\gamma,p)$, $\Delta_{n}$ and $\Delta_{n}'$ are bounded by the empirical uniform norm of $f$ and $\nabla_{\gamma}f$, so a similar takeaway to the correctly specified case applies: Condition 3 can be satisfied if $V_{n,\delta_{n}}$ is contained in a $||\cdot||_{n,\infty}$-ball centered at $p_{0}$ with shrinking radius. Second, the key restriction on $p_{0}$ is Condition 5, which is interpreted as the misspecification Slater condition in Proposition \ref{prop:convexhemi}. Appendix \ref{ap:complementarity} argues these properties are complementary: Condition 5 can fail under correct specification when the slater condition holds, while Condition 5 deals with $\Gamma_{n,I}(p_{0}) = \emptyset$, a setting where, mechanically, Slater fails.

Finally, Condition 5 implies that $\tilde{\Gamma}_{n,I}(p_{0})$ is a singleton for large $n$. This is not an artifact of the proof: with a flat bottom, lower hemicontinuity of the exact pseudo-identified set can fail under Conditions 1-4. Indeed, in such case, posterior draws of $p$ act as small tilts of the criterion that collapse the exact set of minimizers onto draw-dependent faces of $\tilde{\Gamma}_{n,I}(p_{0})$, so the posterior for the exact set concentrates on strict subsets of $\tilde{\Gamma}_{n,I}(p_{0})$ and is spuriously precise -- a Bayesian analogue of the spurious precision under misspecification analyzed by \cite{andrews2024misspecified}.\footnote{Appendix \ref{ap:failurehemicontinuity} presents an example of this.} When a non-singleton pseudo-identified set is anticipated, we recommend reporting the posterior of a slightly relaxed pseudo-identified set $\tilde{\Gamma}_{n,I}^{t}:=\{\gamma \in \Gamma: \tilde{Q}_{n}(\gamma,p) \leq t\}$, $t > 0$, rather than the exact set of minimizers. Proposition \ref{prop:relaxed} in Appendix \ref{ap:failurehemicontinuity} shows the relaxation remedies this: if $\tilde{\Gamma}_{n,I}(p_{0})$ is a set of weak sharp minima of $\tilde{Q}_{n}(\cdot,p_{0})$, which, unlike Condition 5, is a condition compatible with flat bottoms, then, for large $n$ and uniformly over $p \in V_{n,\delta_{n}}$, $\tilde{\Gamma}_{n,I}^{t}(p)$ contains $\tilde{\Gamma}_{n,I}(p_{0})$ and lies within Hausdorff distance $2t/\eta$ of it. Furthermore, taking $t:=t_{n} \downarrow 0$ with $u_{n}/t_{n}\rightarrow 0$ as $n\rightarrow \infty$ restores Hausdorff consistency without any singleton requirement, which mirrors the slack sequences of \cite{chernozhukov2007estimation}.
\subsubsection{Uniform Convergence of Choice Probabilities Under Gaussian Priors.}\label{sec:GPverification}
A key takeaway from Section \ref{sec:suff1} is that, for certain classes of $\Gamma_{n,I}(p_{0})$ (and $\tilde{\Gamma}_{n,I}(p_{0})$), empirical uniform convergence of the conditional choice probabilities is a sufficient condition for Assumption \ref{as:prior}. This section derives conditions under which $\Pi_{n}((p(x_{1}),...,p(x_{n})) \in \cdot | Y^{(n)})$ concentrates around $(p_{0}(x_{1}),...,p_{0}(x_{n}))$ in the empirical supremum norm for the logistic stick-breaking GP prior. For notation, let $C([0,1]^{d_{x}},\mathbb{R})$ be the space of real-valued continuous functions defined on $[0,1]^{d_{x}}$, let $||\cdot||_{\infty}$ be the supremum norm (i.e., $||b||_{\infty} = \sup_{x \in [0,1]^{d_{x}}}|b(x)|$), and let $C^{a}([0,1]^{d_{x}},\mathbb{R})$, $a > 0$, be the H\"{o}lder space of functions that are $\lfloor a \rfloor$ times differentiable with $\lfloor a \rfloor$th derivative that is H\"{o}lder continuous with index $(a- \lfloor a \rfloor)$.
\begin{assumption}[Covariates]\label{as:covariates}
The following conditions hold:
1. $\mathcal{X} = [0,1]^{d_{x}}$, 2. $\{X_{i}\}_{i \geq 1}$ is i.i.d (i.e., $P_{0,X}^{(\infty)} = \bigotimes_{i \geq 1}P_{0,X}$), and 3. $P_{0,X}$ admits a density $p_{0,X}$ that is uniformly bounded away from zero and infinity on $[0,1]^{d_{x}}$.
\end{assumption}
\begin{assumption}[Gaussian Process]\label{as:GP}
The prior $\Pi$ for $p $ is $LSBGP(\{0\}_{k=1}^{K-1},\{\kappa_{k}\}_{k=1}^{K-1})$ with the Gaussian processes $B_{1},...,B_{K-1}$ additionally satisfying 1. $B_{k} \in C^{a_{k}}([0,1]^{d_{x}},\mathbb{R})$ for all $a_{k} < \alpha_{k}$, and 2. are truncated to $||B_{k}||_{\infty} \leq \overline{B}_{k}$ for some $\overline{B}_{k} < \infty$ for each $k=1,...,K-1$.
\end{assumption}
Assumptions \ref{as:covariates} and \ref{as:GP} are standard restrictions on $\{X_{i}\}_{i \geq 1}$ and the prior. Assumption \ref{as:covariates}.1 imposes that the covariates take values in $[0,1]^{d_{x}}$. This condition is without loss of generality if $\mathcal{X}$ is a compact subset of $\mathbb{R}^{d_{x}}$, and is routinely imposed in nonparametric econometrics. Assumption \ref{as:covariates}.2 imposes that the covariates are i.i.d, a standard assumption in the partial identification literature. Assumption \ref{as:covariates}.3 requires that the marginal distribution of $X_{i}$ has a density which is uniformly bounded away from zero and infinity on $[0,1]^{d_{x}}$, a condition that is satisfied, for instance, if $\mathcal{X}$ is compact, and $p_{0,X}$ is continous and positive everywhere. Assumption \ref{as:GP}.1 requires that the sample paths of the Gaussian process $B_{k}$ are almost $\alpha_{k}$ regular in the sense that they take values in the H\"{o}lder space $C^{a_{k}}([0,1]^{d_{x}})$. Common Gaussian processes (in particular, the Mat\'{e}rn processes from Section \ref{sec:GPimplementation}) satisfy this condition \citep{vaart2008rates,JMLR:v12:vandervaart11a}. Assumption \ref{as:GP}.2 imposes that the sample paths are uniformly bounded. It is a technical condition because $\overline{B}_{k}$ permitted to be arbitrarily large, but finite, and, for this reason, can be ignored in practice. Uniform boundedness restrictions on GP priors are often imposed when deriving strong norm posterior concentration guarantees \citep{gine2011rates,norets2015bayesian,walker2026semiparametric}.\footnote{Lemma 5.1 of \cite{van2008reproducing} guarantees that the event $\{||B_{k}||_{\infty} \leq \bar{B}_{k}\}$ receives positive measure under the probability law of the Gaussian process.}

We start with an intermediate empirical mean square posterior concentration result. For notation, let $(\mathcal{H}_{k},||\cdot||_{\mathcal{H}_{k}})$ be the reproducing kernel Hilbert space (RKHS) attached to $B_{k}$, let $\overline{\mathcal{H}}_{k}$ be the closure of $\mathcal{H}_{k}$ in $(C([0,1]^{d_{x}},\mathbb{R}),||\cdot||_{\infty})$, and let $$\varphi_{k,b_{0,k}}(\delta) = \inf_{f_{k} \in \mathcal{H}_{k}: ||f_{k}-b_{0,k}||_{\infty}< \delta}\frac{1}{2}||f_{k}||_{\mathcal{H}_{k}} - \log P_{GP(0,\kappa_{k})}(||B||_{k}< \delta)$$ for $\delta > 0$ be the concentration function of $GP(0,\kappa_{k})$ at $b_{0,k}$, where $b_{0,k} = \Lambda^{-1}(q_{0,k})$ for each $k=1,...,K-1$.
\begin{proposition}\label{prop:L2}
Suppose Assumptions \ref{as:dgp}.1, \ref{as:covariates}.1, and \ref{as:GP}.2 hold, $b_{0,k} \in \overline{\mathcal{H}}_{k}$ and $||b_{0,k}||_{\infty} \leq \bar{B}_{k}$ for each $k \in \{1,...,K-1\}$, and, for each $k \in \{1,...,K-1\}$, there exists a sequence $\{\tilde{\delta}_{n,k}\}_{n \geq 1}$ with $\tilde{\delta}_{n,k} \rightarrow 0$ as $n\rightarrow \infty$, $n\tilde{\delta}_{n,k}^{2} \rightarrow \infty$ as $n\rightarrow \infty$, and $\varphi_{k,b_{0,k}}(\tilde{\delta}_{n,k}) \leq n \tilde{\delta}_{n,k}^{2}$. Then for $P_{0,X}^{(\infty)}$-almost every fixed realization $\{x_{i}\}_{i \geq 1}$ of $\{X_{i}\}_{i \geq 1}$,
\begin{align*}
     \Pi_{n}\left(||p-p_{0}||_{n,2}\geq M\tilde{\delta}_{n}\middle | Y^{(n)}\right) \overset{P_{0}^{(n)}}{\longrightarrow} 0
\end{align*}
as $n\rightarrow \infty$, where $M>0$ is a large constant and $\tilde{\delta}_{n} = \max_{1 \leq k \leq K-1}\tilde{\delta}_{n,k}$.
\end{proposition}
Proposition \ref{prop:L2} states that if $b_{0,k}$ can be uniformly approximated by sequences in the RKHS $\mathcal{H}_{k}$, then the posterior concentrates around empirical $L^{2}$ balls centered at $p_{0,k}$ (and the rate of convergence is quantified by $\varphi_{b_{0,k}}$). It can be viewed as an extension of the binary regression results in \cite{vaart2008rates} (i.e., Theorem 3.2) to categorical regression, and, since it is based on the logit stick breaking Gaussian process prior, it may be of independent general interest, as we are not aware of formal posterior contraction rate guarantees for these priors.\footnote{\cite{pati2013posterior} and \cite{norets2014posterior} derive posterior consistency results (without rates) in the related context of covariate dependent stick breaking models for conditional densities.} In some cases, these posteriors contract at the optimal rate. For example, if $\kappa_{k}$ belongs to the Mat\'{e}rn class with regularity parameter $\alpha_{k} > 0$ and $b_{0,k}$ is in the intersection of a H\"{o}lder and a Sobolev space of order $\alpha_{0,k}> 0$, then Theorem 5 of \cite{JMLR:v12:vandervaart11a} implies that $\tilde{\delta}_{n,k} = n^{-\min\{\alpha_{k},\alpha_{0,k}\}/(2\alpha_{k}+d_{x})}$ for each $k=1,...,K-1$. Consequently, when the prior and true smoothness match (i.e., $\alpha_{k} = \alpha_{0,k}$ for each $k=1,...,K-1$), we conclude that $\tilde{\delta}_{n} = n^{-\alpha_{0}/(2\alpha_{0}+d_{x})}$, where $\alpha_{0} = \min\{\alpha_{0,1},...,\alpha_{0,K-1}\}$, thereby achieving the optimal rate from \cite{stone1982optimal}.

Proposition \ref{prop:post_sup_norm} establishes that, when combined with smoothness restrictions on $b_{0,k}$ and $B_{k}$, posterior consistency in the empirical $L^{2}$ distance (Proposition \ref{prop:L2}) implies posterior consistency in the empirical supremum norm. The argument is based on wavelet interpolation and follows a similar logic to Proposition 2 of \cite{walker2026semiparametric}, except that it applies to a different class of priors (i.e., the logistic stick-breaking GP prior rather than an infinite-dimensional exponential family). The only additional condition, Assumption \ref{as:wavelet}, is a condition on a wavelet basis, and, since its exposition is somewhat cumbersome, we defer it to the Appendix. For Mat\'{e}rn GPs with $\alpha_{k} = \alpha_{0,k}$ for each $k=1,...,K-1$, the condition holds if $\min_{1 \leq k \leq K-1}\alpha_{k}>(1+\sqrt{5})d_{x}/4$, which is a slight strengthening of the baseline restriction $\min_{1 \leq k \leq K-1}\alpha_{k} > d_{x}/2$ in Assumption \ref{as:GP} (see Appendix \ref{ap:verifyingwavelet} for more discussion).
\begin{proposition}\label{prop:post_sup_norm}
Suppose that Assumptions \ref{as:dgp}.1, \ref{as:covariates}, \ref{as:GP}, and \ref{as:wavelet} hold, the conditions of Proposition \ref{prop:L2} hold, and for each $k \in \{1,...,K-1\}$, there exists $\alpha_{0,k} > d_{x}/2$ such that $b_{0,k} \in C^{\alpha_{0,k}}([0,1]^{d_{x}},\mathbb{R})$. Then, for $P_{0,X}^{(\infty)}$-almost every fixed realization $\{x_{i}\}_{i \geq 1}$ of $\{X_{i}\}_{i \geq 1}$, the following holds: for every $\varepsilon > 0$,
\begin{align*}
    \Pi_{n}\left(||p-p_{0}||_{n,\infty} \geq \varepsilon \middle | Y^{(n)}\right) \overset{P_{0}^{(n)}}{\longrightarrow} 0
\end{align*}
as $n\rightarrow \infty$.
\end{proposition}
\section{Extensions}\label{sec:extensions}
This section extends our framework to settings with aggregated discrete responses and continuous endogenous variables.
\subsection{Aggregated Discrete Responses.} In some economic settings (e.g., empirical industrial organization), the researcher may only have access to aggregated discrete outcome over some units (e.g., market shares). Our framework readily extends to this setting. For concreteness, suppose that the observed data is $(Y',X')'$, where $X = (\tilde{X}',M)$ is a vector of observed covariates $\tilde{X}$ and an exogenous number of units $M$ over which the aggregation is performed (e.g., market size), $Y = M\tilde{Y} \in \mathbb{N}^{K}$ is the vector of counts, and $\tilde{Y}$ is a $K$-dimensional random vector that takes values in the probability simplex (i.e., shares). Assuming the aggregated units are independent, the reduced-form model is
\begin{align*}
    Y_{i}|p \overset{ind}{\sim} Multinomial(m_{i},p(\tilde{x}_{i})), \quad i=1,...,n, \quad p \sim \Pi,
\end{align*}
where $Multinomial(m,p(\tilde{x}))$ is the multinomial distribution based on $m$ trials and event probabilities $p(\tilde{x})$, and $\Pi$ is the same prior over $p$ as before. Consequently, by defining $\Gamma_{n,I}(p) = \{\gamma \in \Gamma: f(x_{i},p(\tilde{x}_{i}),\gamma) \in \mathbb{R}^{d_{f}}_{+} \ \forall \ i=1,...,n\}$ (i.e., the same as before except acknowledging the partition $x_{i}=(\tilde{x}_{i},m_{i})$), a posterior for $(p(\tilde{x}_{1}),...,p(\tilde{x}_{n}))$ implies a posterior for $\Gamma_{n,I}(p)$. This demonstrates that, conceptually, there is virtually no difference between our main setup and the aggregated discrete choice setting.

Similarly, the stick-breaking technique used for our implementation readily extends to this setting. Using the stick-breaking weights, $p_{k} = q_{k}\prod_{j < k}(1-q_{j})$ for $k=1,...,K-1$ and $p_{K} = \prod_{j=1}^{K-1}(1-q_{j})$, the multinomial likelihood $L_{n}(p)$ satisfies $L_{n}(p) = L_{n}^{*}(q)$ with
\begin{align*}
    L_{n}^{*}(q) = \prod_{k=1}^{K-1}\prod_{i \in \mathcal{I}_{k}}{m_{i,k} \choose Y_{i,k}}q_{k}(\tilde{x}_{i})^{Y_{i,k}}(1-q_{k}(\tilde{x}_{i}))^{m_{i,k}-Y_{i,k}},
\end{align*}
where $m_{i,k} = m_{i} - \sum_{j < k}Y_{i,j}$, $Y_{i,k}$ is the $k$th element of $Y_{i}$, and $\mathcal{I}_{k} = \{i: m_{i,k} > 0\}$. Notice that, for each $k$, the factor is the density of $|\mathcal{I}_{k}|$ independent Binomial distributions with $m_{i,k}$ trials. Consequently, by specifying the law of $(q_{1},...,q_{K-1})$ to be that of $(\Lambda(B_{1}),...,\Lambda(B_{K-1}))$, where $B_{k} \overset{ind}{\sim} GP(\mu_{k},\kappa_{k})$ for each $k=1,...,K-1$, the binomial version of \cite{polson2013bayesian}'s P\'{o}lya-Gamma data augmentation can be used to sample from $\pi_{n,k}(B_{\mathcal{I}_{k},k}|Y^{(n)})$ (see Appendix \ref{ap:posteriordraws}). The other implementation aspects (e.g., post-processing to sample from $\pi_{n,k}(B_{\mathcal{I}_{k}^{c},k}|Y^{(n)},B_{\mathcal{I}_{k},k})$ and computing $\Gamma_{n,I}(p)$) do not change because they do not depend on the likelihood.

Since the reduced-form parameter $p$ is unchanged in the aggregated setting, the only possible concern for the asymptotic theory is whether this extension affects posterior concentration of $p$ around $p_{0}$. That is, whether the results from Section \ref{sec:GPverification} extends. The next result states that if $\{m_{i}\}_{i \geq 1}$ forms a bounded sequence, then Propositions \ref{prop:L2} and \ref{prop:post_sup_norm} are unaffected by the aggregation. Consequently, under bounded $M$, aggregated discrete choice settings are both practically and theoretically accommodated.
\begin{proposition}\label{prop:GPmultinomial}
Suppose that the true data-generating process satisfies $Y_{i}|\{X_{j}\}_{j \geq 1} \overset{ind}{\sim} Multinomial(M_{i},p_{0}(\tilde{X}_{i}))$ for every $i \geq 1$, $\{\tilde{X}_{i}\}_{i \geq 1}$ satisfies Assumption \ref{as:covariates}, and $\{M_{i}\}_{i \geq 1}$ is bounded with $P_{0,X}^{(\infty)}$-probability equal to one. Further, suppose that Assumptions \ref{as:GP} and \ref{as:wavelet} hold,  $b_{0,k} \in \overline{\mathcal{H}}_{k} \cap C^{\alpha_{0,k}}([0,1]^{d_{x}},\mathbb{R})$ for some $\alpha_{0,k} > d_{x}/2$ and $||b_{0,k}||_{\infty} \leq \bar{B}_{k}$ for each $k \in \{1,...,K-1\}$, and, for each $k \in \{1,...,K-1\}$, there exists a sequence $\{\tilde{\delta}_{n,k}\}_{n \geq 1}$ with $\tilde{\delta}_{n,k} \rightarrow 0$ as $n\rightarrow \infty$, $n\tilde{\delta}_{n,k}^{2} \rightarrow \infty$ as $n\rightarrow \infty$, and $\varphi_{k,b_{0,k}}(\tilde{\delta}_{n,k}) \leq n \tilde{\delta}_{n,k}^{2}$. Then, for $P_{0,X}^{(\infty)}$-almost every fixed realization $\{x_{i}\}_{i \geq 1}$ of $\{X_{i}\}_{i \geq 1}$, the following is true:  for every $\varepsilon > 0$,
\begin{align*}
    \Pi_{n}\left(||p-p_{0}||_{n,\infty} \geq \varepsilon \middle | Y^{(n)}\right) \overset{P_{0}^{(n)}}{\longrightarrow} 0
\end{align*}
as $n\rightarrow \infty$, where $P_{0}^{(n)} = \bigotimes_{i=1}^{n}Multinomial(m_{i},p_{0}(\tilde{x}_{i}))$
\end{proposition}
\subsection{Continuous Response Variables.} Some important partially identified models involve continuously distributed $Y$, such as linear regression models with interval censoring \citep{manski2002inference}, incomplete auction models \citep{haile2003inference}, and models of exporter choice \citep{dickstein2018exporters}, to name a few. Our framework extends to this setting. First, we redefine the reduced-form model as
\begin{align*}
Y_{i}|p_{Y|X}\overset{ind}{\sim} p_{Y|X}(\cdot|x_{i}), \quad i=1,...,n, \quad p_{Y|X} \sim \Pi,
\end{align*}
where $p_{Y|X}$ is the conditional probability density function of $Y|X$ and $\Pi$ is a prior for $p_{Y|X}$. Then, we redefine the conditional identified set as 
\begin{align*}
\Gamma_{n,I}(p_{Y|X}) = \left\{\gamma \in \Gamma: f(x_{i},p_{Y|X}(\cdot|x_{i}),\gamma) \in \mathbb{R}_{+}^{d_{f}} \ \forall \ i=1,...,n\right\}.     
\end{align*}
Subject to regularity conditions, our Bayesian inference framework extends: the posterior for $(p_{Y|X}(\cdot|x_{1}),...,p_{Y|X}(\cdot|x_{n}))$ implies a posterior for $\Gamma_{n,I}(p_{Y|X})$ via the mapping $(p_{Y|X}(\cdot|x_{1}),...,p_{Y|X}(\cdot|x_{n})) \mapsto \Gamma_{n,I}(p_{Y|X})$. Importantly, by defining $f(x_{i},p_{Y|X}(\cdot|x_{i}),\gamma) = \int g(y,x_{i},\gamma)p_{Y|X}(y|x_{i})dy$, this nests conditional moment inequalities with continuous $Y$.

Since the consistency theory only uses the finite support of $Y$ in Section \ref{sec:GPverification}, some of our posterior consistency results readily extend to continuous $Y$. Suppose that the posterior for $p_{Y|X}$ concentrates on $V_{n,\delta_{n}}^{dens} :=\mathcal{P}_{n,Y|X} \cap \{p_{Y|X}: e_{n}(p_{Y|X},p_{0,Y|X}) < \delta_{n}\}$, where $e_{n}$ is a semimetric over conditional densities and $\mathcal{P}_{n,Y|X}$ is a sieve, and that these sets are structured enough to ensure that
\begin{align}
\sup_{(\gamma,p) \in \Gamma \times V_{n,\delta_{n}}^{dens}}\max_{1 \leq i \leq n}||f(x_{i},p_{Y|X}(\cdot|x_{i}),\gamma)-f(x_{i},p_{0,Y|X}(\cdot|x_{i}),\gamma)||_{2}\longrightarrow 0  \label{eq:continuous_ucon1}\\
\sup_{(\gamma,p) \in \Gamma \times V_{n,\delta_{n}}^{dens}}\max_{1 \leq i \leq n}||W(x_{i},p_{Y|X}(\cdot|x_{i}),\gamma)-W(x_{i},p_{0,Y|X}(\cdot|x_{i}),\gamma)||_{2}\longrightarrow 0 \label{eq:continuous_ucon2}
\end{align}
as $n\rightarrow \infty$. Then, a continuous $Y$ extension of Assumption \ref{as:prior} for convex $\Gamma_{n,I}(p_{Y|X})$ and $Q_{n,r}(\gamma,p_{Y|X}) = ||\text{dist}_{W(\cdot,p_{Y|X}(\cdot|\cdot),\gamma)}(f(\cdot,p_{Y|X}(\cdot|\cdot),\gamma),\mathbb{R}_{+}^{d_{f}})||_{n,r}$, $r \in [1,\infty]$, holds so long as the true density $p_{0,Y|X}$ satisfies analogous conditions to $p_{0}$ in Propositions \ref{prop:ucon} and \ref{prop:convexhemi}.\footnote{Continuous $Y$ extensions of Propositions \ref{prop:taxicab}, \ref{prop:minorantmisspec}, and \ref{prop:misspec_hemi_convex} can be achieved. For conciseness, we omit them.} Proposition \ref{prop:continuousY} below formalizes this (with Assumption \ref{as:continuousDGP} in Appendix \ref{ap:extensionproofs} stating conditions on $p_{0,Y|X}$). Consequently, if the DGP satisfies $Y_{i}|\{X_{j}\}_{j \geq 1} \overset{ind}{\sim}p_{0,Y|X}$ and $Q_{n,r}(\gamma,p_{0,Y|X})$ is well-separated at $\Gamma_{n,I}(p_{0,Y|X}) \neq \emptyset$ (i.e., a modification of Assumption \ref{as:dgp} holds), then a virtually identical argument to Theorem \ref{thm:consistency} establishes posterior consistency for $\Gamma_{n,I}(p_{Y|X})$ at $\Gamma_{n,I}(p_{0,Y|X})$ in the Hausdorff distance. For conditional moment inequalities with bounded $g(Y,X,,\gamma)$ and $W(x,p_{Y|X}(\cdot|x),\gamma)$ based on $Var(g(Y,X,\gamma)|X=x)$, Proposition 2 of \cite{walker2026semiparametric} provides sufficient conditions for (\ref{eq:continuous_ucon1}) and (\ref{eq:continuous_ucon2}) under GP priors for $p_{Y|X}$.
\begin{proposition}\label{prop:continuousY}
Suppose that, for $P_{0,X}^{(\infty)}$-almost every fixed realization $\{x_{i}\}_{i \geq 1}$ of $\{X_{i}\}_{i \geq1}$, there exists sets $\{\mathcal{P}_{n,Y|X}\}_{n \geq 1}$, metrics $\{e_{n}\}_{n \geq 1}$, and sequences $\{\delta_{n}\}_{n \geq 1}$ such that $\delta_{n} \rightarrow 0$ as $n\rightarrow \infty$, and
\begin{align*}
\Pi_{n}(p_{Y|X} \in \mathcal{P}_{n,Y|X}|Y^{(n)}) &= 1+o_{P_{0,Y|X}^{(n)}}(1) \\
\Pi_{n}(e_{n}(p_{Y|X},p_{0,Y|X})<\delta_{n}|Y^{(n)}) &= 1+o_{P_{0,Y|X}^{(n)}}(1)
\end{align*}
Moreover, suppose that the criterion is $Q_{n,r}(\gamma,p) =  ||\text{dist}_{W(\cdot,p_{Y|X}(\cdot|\cdot),\gamma)}(f(\cdot,p_{Y|X}(\cdot|\cdot),\gamma),\mathbb{R}_{+}^{d_{f}})||_{n,r}$ for some $r \in [1,\infty]$, and, for $P_{0,X}^{(\infty)}$-almost every fixed realization $\{x_{i}\}_{i \geq 1}$ of $\{X_{i}\}_{i \geq1}$,
\begin{align*}
\sup_{(\gamma,p) \in \Gamma \times V_{n,\delta_{n}}^{dens}}\max_{1 \leq i \leq n}||f(x_{i},p_{Y|X}(\cdot|x_{i}),\gamma)-f(x_{i},p_{0,Y|X}(\cdot|x_{i}),\gamma)||_{2}\longrightarrow 0  \\
\sup_{(\gamma,p) \in \Gamma \times V_{n,\delta_{n}}^{dens}}\max_{1 \leq i \leq n}||W(x_{i},p_{Y|X}(\cdot|x_{i}),\gamma)-W(x_{i},p_{0,Y|X}(\cdot|x_{i}),\gamma)||_{2}\longrightarrow 0
\end{align*}
as $n\rightarrow \infty$. If Assumption \ref{as:continuousDGP} holds, then Assumption \ref{as:prior} holds (with $p_{Y|X}$ in place of $p$).
\end{proposition}

\section{Conclusion}\label{sec:conclusion}
This paper develops a nonparametric Bayesian approach to inference in partially identified discrete response models. The central observation is simple: in a large class of models, the identified set is a functional of an unrestricted reduced-form conditional choice probability. Rather than placing prior structure on unidentified structural parameters, we place a flexible prior on this identifiable reduced form parameter and map its posterior into a posterior for the identified set. This delivers a fully Bayesian procedure for set-valued parameters while keeping the source of statistical uncertainty transparent. The extensions demonstrate that our insights paper also apply to aggregated discrete choice and continuous outcomes.

A main advantage of this approach relative to existing methods is that it works directly with the conditional restrictions that define the model. In conditional moment inequality models, there is no need to replace conditional moments with a user-chosen collection of unconditional moments or instruments. More generally, the procedure does not require discretizing continuously distributed covariates simply to make the problem finite dimensional. These steps are common in empirical implementations, but they can discard identifying information and make inference sensitive to choices that are external to the economic model. By learning the conditional probability mass function directly, our framework accommodates rich covariate variation while preserving the original identifying restrictions. In this sense, it combines the flexibility typically associated with nonparametric frequentist methods with a genuinely Bayesian treatment of uncertainty about the identified set.

Implementation is also straightforward. With the logistic stick-breaking Gaussian process priors considered here, P\'{o}lya--Gamma data augmentation yields conditionally Gaussian posterior updates, and the stick-breaking components can be sampled independently and in parallel. Given a conditional choice probability, draws of the identified set are obtained by solving the same optimization problem that would be used if those probabilities were known, making identified set computation a parallelizable post-processing step.

Our general asymptotic theory shows that our proposal has a clear large-sample interpretation: under correct specification, the posterior concentrates on the true identified set; under misspecification, the posterior probability of an empty identified set provides a diagnostic, while inference can instead be conducted on a pseudo-identified set. The low-level analysis in Section \ref{sec:verifyingconditions} highlights that our general asymptotic theory applies to the logistc stick-breaking Gaussian process priors. Taken together, these results provide a practical route to fully Bayesian inference in partially identified models without requiring researchers to coarsen covariates or alter the conditional identifying content of the model.

\bibliographystyle{ecca}
\bibliography{multinomial}

\appendix
\section{Reduced-Form Posterior Sampling Details}\label{ap:posteriordraws}
This section describes posterior sampling for the reduced-form parameter. The overall algorithm is as follows:
\begin{enumerate}
    \item For $k=1,...,K-1$,
\begin{enumerate}[label=\roman*.]
    \item Perform the steps in Section \ref{ap:data} to obtain a sample $\{B_{\mathcal{I}_{k},k}^{[s]}\}_{s=1}^{S}$ from the marginal posterior $\pi_{n,k}(B_{\mathcal{I}_{k},k}|Y^{(n)})$.
    \begin{itemize}
        \item With fixed hyperparameters and discrete outcomes, see \ref{ap:fixedsample}.
        \item With fixed hyperparameters and aggregate discrete outcomes, see \ref{ap:aggregatesample}.
        \item To incorporate hyperparameter selection, see \ref{ap:hyperparameters}.
    \end{itemize}
    \item Using the sample $\{B_{\mathcal{I}_{k},k}^{[s]}\}_{s=1}^{S}$, perform the steps in Section \ref{ap:postprocess} to generate a sample $\{B_{\mathcal{I}_{k},k}^{[s]}\}_{s=1}^{S}$ from $\pi_{n,k}(B_{\mathcal{I}_{k}^{c},k}|Y^{(n)},B_{\mathcal{I}_{k},k})$.
\end{enumerate}
    \item Using the sample $\{B_{n,k}^{[s]}\}_{s=1}^{S}$, where $B_{n,k}^{[s]} = (B_{\mathcal{I}_{k},k}^{[s]},B_{\mathcal{I}_{k}^{c},k}^{[s]})$, compute $\{(p^{[s]}(x_{i}))_{i=1}^{n}\}_{s=1}^{S}$ using the stick-breaking formula.
\end{enumerate}
\subsection{Sampling from $\pi_{n,k}(B_{\mathcal{I}_{k},k}|Y^{(n)})$.}\label{ap:data} This section details sampling from $\pi_{n,k}(B_{\mathcal{I}_{k}}|Y^{(n)})$. It has three subsections, each corresponding to Step 1.i above.
\subsubsection{Sampling from $\pi_{n,k}(B_{\mathcal{I}_{k},k}|Y^{(n)})$}\label{ap:fixedsample} Recall that sampling from $\pi_{n,k}(B_{\mathcal{I}_{k},k}|Y^{(n)})$ corresponds to Bayesian inference for a nonparametric logistic regression model, and, as a result, the P\'{o}lya-Gamma sampler of \cite{polson2013bayesian} is applicable. Specifically, the steps for generating a draw $B_{\mathcal{I}_{k},k}$ from $\pi_{n,k}(B_{\mathcal{I}_{k},k}|Y^{(n)})$ are
\begin{enumerate}
\item Sample $\omega_{i,k} \overset{ind}{\sim}PG(1,B_{k}(x_{i}))$ for $i \in \mathcal{I}_{k}$, where $PG(b,c)$ denotes a P\'{o}lya-Gamma distribution with parameters $b > 0$ and $c \in \mathbb{R}$.
\item Sample $B_{\mathcal{I}_{k},k} \sim \mathcal{N}(\hat{\mu}_{\mathcal{I}_{k},k},\hat{\kappa}_{\mathcal{I}_{k},\mathcal{I}_{k}k})$, where
\begin{align*}
\hat{\mu}_{\mathcal{I}_{k},k} &= \mu_{\mathcal{I}_{k},k} + \kappa_{\mathcal{I}_{k},\mathcal{I}_{k},k}(\kappa_{\mathcal{I}_{k},\mathcal{I}_{k},k}+\Omega_{\mathcal{I}_{k}}^{-1})^{-1}(Z_{\mathcal{I}_{k},k}-\mu_{\mathcal{I}_{k},k})\\
\hat{\kappa}_{\mathcal{I}_{k},k}  &=     \kappa_{\mathcal{I}_{k},\mathcal{I}_{k},k} -  \kappa_{\mathcal{I}_{k},\mathcal{I}_{k},k}(\kappa_{\mathcal{I}_{k},\mathcal{I}_{k},k}+\Omega_{\mathcal{I}_{k}}^{-1})^{-1}\kappa_{\mathcal{I}_{k},\mathcal{I}_{k},k}
\end{align*}
with $\mu_{\mathcal{I}_{k},k} = (\mu_{k}(x_{i}))_{i \in \mathcal{I}_{k}}$, $\kappa_{\mathcal{I}_{k},\mathcal{I}_{k}}=(\kappa(x_{i},x_{j}))_{i,j \in \mathcal{I}_{k}}$, $\Omega_{\mathcal{I}_{k}} = \text{diag}((\omega_{i,k})_{i \in \mathcal{I}_{k}})$, and $Z_{\mathcal{I}_{k},k} = ((\mathbf{1}\{Y_{i} =y_{k}\}-1/2)/\omega_{i,k})_{i \in \mathcal{I}_{k}}$.
\end{enumerate}
Cycling through these steps $S$ times (and independently across $k$) generates a sample $\{(B_{\mathcal{I}_{1},1}^{[s]},...,B_{\mathcal{I}_{K-1},K-1}^{[s]})\}_{s=1}^{S}$ from the marginal posterior $\pi_{n}\left(B_{\mathcal{I}_{1},1},...,B_{\mathcal{I}_{K-1},K-1} \middle |Y^{(n)}\right)$.
\subsubsection{Sampling from $\pi_{n,k}(B_{\mathcal{I}_{k},k}|Y^{(n)})$ in the Aggregate Case.}\label{ap:aggregatesample} Recall from Section \ref{sec:extensions} that $\pi_{n,k}(B_{\mathcal{I}_{k}}|Y^{(n)})$ corresponds to Bayesian inference in a model, where $Y_{i,k} \overset{ind}{\sim} Bin(m_{i,k},q_{k}(x_{i}))$ for each $i \in \mathcal{I}_{k}$. Consequently, the binomial version of P\'{o}lya-Gamma sampling from \cite{polson2013bayesian} applies. This means that the sweeps of the Gibbs sampler are
\begin{enumerate}
    \item Sample $\omega_{i,k} \overset{ind}{\sim} PG(m_{i,k},B_{k}(x_{i}))$ for $i \in \mathcal{I}_{k}$.
    \item Sample $B_{n,k} \sim \mathcal{N}(\hat{\mu}_{n,k},\kappa_{n,k})$, where the only difference from Section \ref{ap:data} is that $Z_{\mathcal{I}_{k},k}$ is redefined as $Z_{\mathcal{I}_{k},k} = ((Y_{i,k}-m_{i,k}/2)/\omega_{i,k})_{i \in \mathcal{I}_{k}}$.
\end{enumerate}
Cycling through these steps $S$ times (and independently across $k$) generates a sample $\{(B_{\mathcal{I}_{1},1}^{[s]},...,B_{\mathcal{I}_{K-1},K-1}^{[s]})\}_{s=1}^{S}$ from the marginal posterior $\pi_{n}\left(B_{\mathcal{I}_{1},1},...,B_{\mathcal{I}_{K-1},K-1} \middle |Y^{(n)}\right)$.
\subsubsection{Incorporating Hyperparameters}\label{ap:hyperparameters}
Recall that, by factorizing the posterior into the product $\pi_{n,k}(B_{n,k},\tau_{k}|Y^{(n)}) = \pi_{n,k}(B_{\mathcal{I}_{k}^{c},k}|Y^{(n)},B_{\mathcal{I}_{k},k},\tau_{k})\pi_{n,k}(B_{\mathcal{I}_{k},k},\tau_{k}|Y^{(n)})$, the only way that prior hyperparameters $\tau_{k}$ affect the algorithm is that now we have to sample from $\pi_{n,k}(B_{\mathcal{I}_{k},k},\tau_{k}|Y^{(n)})$. This amounts to augmenting the P\'{o}lya-Gamma Gibbs sampler with sweeps that additionally sample from the conditional posterior $\pi_{n,k}(\tau_{k}|Y^{(n)},B_{\mathcal{I}_{k}})$. Namely,
\begin{enumerate}
\item Sample $\omega_{i,k} \overset{ind}{\sim}PG(1,B_{k}(x_{i}))$ for $i \in \mathcal{I}_{k}$, where $PG(b,c)$ denotes a P\'{o}lya-Gamma distribution with parameters $b > 0$ and $c \in \mathbb{R}$. For the \textit{aggregrate} case, replace $PG(1,B_{k}(x_{i}))$ with $PG(m_{i,k},B_{k}(x_{i}))$
\item Sample $B_{\mathcal{I}_{k},k} \sim \mathcal{N}(\hat{\mu}_{\mathcal{I}_{k},k},\hat{\kappa}_{\mathcal{I}_{k},\mathcal{I}_{k}k})$, where
\begin{align*}
\hat{\mu}_{\mathcal{I}_{k},k} &= \mu_{\mathcal{I}_{k},k} + \kappa_{\mathcal{I}_{k},\mathcal{I}_{k},k}(\kappa_{\mathcal{I}_{k},\mathcal{I}_{k},k}+\Omega_{\mathcal{I}_{k}}^{-1})^{-1}(Z_{\mathcal{I}_{k},k}-\mu_{\mathcal{I}_{k},k})\\
\hat{\kappa}_{\mathcal{I}_{k},k}  &=     \kappa_{\mathcal{I}_{k},\mathcal{I}_{k},k} -  \kappa_{\mathcal{I}_{k},\mathcal{I}_{k},k}(\kappa_{\mathcal{I}_{k},\mathcal{I}_{k},k}+\Omega_{\mathcal{I}_{k}}^{-1})^{-1}\kappa_{\mathcal{I}_{k},\mathcal{I}_{k},k}
\end{align*}
with $\mu_{\mathcal{I}_{k},k} = (\mu_{k}(x_{i}))_{i \in \mathcal{I}_{k}}$, $\kappa_{\mathcal{I}_{k},\mathcal{I}_{k}}=(\kappa(x_{i},x_{j}))_{i,j \in \mathcal{I}_{k}}$, $\Omega_{\mathcal{I}_{k}} = \text{diag}((\omega_{i,k})_{i \in \mathcal{I}_{k}})$, and $Z_{\mathcal{I}_{k},k} = ((\mathbf{1}\{Y_{i} =y_{k}\}-1/2)/\omega_{i,k})_{i \in \mathcal{I}_{k}}$ all computed at the current value of $\tau_{k}$. For the \textit{aggregate} case, replace $Z_{\mathcal{I}_{k},k}$ with $Z_{\mathcal{I}_{k},k} = ((Y_{i,k}-m_{i,k}/2)/\omega_{i,k})_{i \in \mathcal{I}_{k}}$.
\item Sample $\tau_{k} \sim \pi_{n,k}(\tau_{k}| Y^{(n)},\mathcal{B}_{\mathcal{I}_{k},k})$.
\end{enumerate}
Cycling through these steps $S$ times (and independently across $k$) generates a sample $\{(B_{\mathcal{I}_{1},1}^{[s]},...,B_{\mathcal{I}_{K-1},K-1}^{[s]}, \tau_{1}^{[s]},...,\tau_{K-1}^{[s]})\}_{s=1}^{S}$ from $\pi_{n}(B_{\mathcal{I}_{1},1},...,B_{\mathcal{I}_{K-1},K-1},\tau_{1},...,\tau_{K-1}  |Y^{(n)})$.
The conditional for $\tau_{k}$ can be sampled via a Metropolis-Hastings step (or perhaps by exploiting some conjugacy). In our implementation, $\tau_{k} = (\alpha,l_{k})$ and we perform these sweeps for the first $10,000$ iterations (i.e., the exploration phase), using Random Walk Metropolis-Hastings to sample the length-scale $l_{k}$ (targeting an acceptance rate of $0.44$). For the remaining $10,000$ iterations, we fix $\tau_{k}$ at $(\alpha,\hat{l}_{k})$, where $\hat{l}_{k}$ is the median of the draws from second half of the exploration phase, and then apply Steps 1 and 2 only in the sweeps.
\subsection{Sampling from $\pi_{n,k}(B_{\mathcal{I}_{k}^{c},k}|Y^{(n)},B_{\mathcal{I}_{k},k})$.}\label{ap:postprocess} Recall that we showed $\pi_{n,k}(B_{\mathcal{I}_{k}^{c},k}|Y^{(n)},B_{\mathcal{I}_{k},k})=\pi_{n,k}(B_{\mathcal{I}_{k}^{c},k}|B_{\mathcal{I}_{k},k})$. Since $B_{k} \sim GP(\mu_{k},\kappa_{k})$, $B_{\mathcal{I}_{k}^{c},k}|B_{\mathcal{I}_{k},k} \sim \mathcal{N}(\hat{\mu}_{\mathcal{I}_{k}^{c},k}, \hat{\kappa}_{\mathcal{I}_{k}^{c},\mathcal{I}_{k}^{c},k})$,
where
\begin{align*}
   &\hat{\mu}_{\mathcal{I}_{k}^{c},k} = \mu_{\mathcal{I}_{k}^{c},k}+ \kappa_{\mathcal{I}_{k}^{c},\mathcal{I}_{k},k}\kappa_{\mathcal{I}_{k},\mathcal{I}_{k},k}^{-1}(B_{\mathcal{I}_{k},k}-\mu_{\mathcal{I}_{k},k}) \\
   &\hat{\kappa}_{\mathcal{I}_{k}^{c},k}= \kappa_{\mathcal{I}_{k}^{c},\mathcal{I}_{k}^{c},k} -\kappa_{\mathcal{I}_{k}^{c},\mathcal{I}_{k},k}\kappa_{\mathcal{I}_{k},\mathcal{I}_{k},k}^{-1}\kappa_{\mathcal{I}_{k},\mathcal{I}_{k}^{c},k},
\end{align*}
with $\mu_{\mathcal{I}_{k}^{c},k} = (\mu_{k}(x_{i}))_{i \notin \mathcal{I}_{k}}$, $\kappa_{\mathcal{I}_{k},\mathcal{I}_{k}^{c},k} = (\kappa(x_{i},x_{j}))_{i \in \mathcal{I}_{k},j \notin \mathcal{I}_{k}}$, $\kappa_{\mathcal{I}_{k}^{c},\mathcal{I}_{k},k}= (\kappa(x_{i},x_{j}))_{i \notin \mathcal{I}_{k},j \in \mathcal{I}_{k}}$, and $\kappa_{\mathcal{I}_{k}^{c},\mathcal{I}_{k}^{c},k}=(\kappa(x_{i},x_{j}))_{i,j\notin \mathcal{I}_{k}}$. Consequently, $\{(B_{\mathcal{I}_{2}^{c},2}^{[s]},...,B_{\mathcal{I}_{K-1}^{c},K-1}^{[s]})\}_{s=1}^{S}$ is generated by sampling from independently sampling $\mathcal{N}(\hat{\mu}_{\mathcal{I}_{k}^{c},k}, \hat{\kappa}_{\mathcal{I}_{k}^{c},\mathcal{I}_{k}^{c},k})$ for each $k=2,...,K-1$ and each $s=1,...,S$ (substituting $B_{\mathcal{I}_{k},k}^{[s]}$ into the expression for $\hat{\mu}_{\mathcal{I}_{k}^{c},k}$). If there is hyperparameter selection, then the same applies here except that $\hat{\mu}_{\mathcal{I}_{k}^{c},k}$ and $\hat{\kappa}_{\mathcal{I}_{k}^{c},\mathcal{I}_{k}^{c},k}$ must be adjusted to reflect the value of $\tau_{k}$.
\section{More on the Dynamic Binary Response Model}\label{ap:simulation}
\subsection{Inequality Restrictions in the Dynamic Binary Response Model}
Under conditional stationarity $U_{1}|X,\alpha \sim U_{2}|X,\alpha$ (and some other mild regularity conditions), Theorem 1 of \cite{khan2023identification} states that $\gamma \in \Gamma_{n,I}(p)$ if and only if the following holds for all $t,s=1,2$ and $i=1,...,n$,
\begin{align*}
    &\mathbf{1}\{p(Y_{t}=1|X=x_{i}) \geq P(Y_{s}=1|X=x_{i})\}\left(\beta (t-s) + (x_{i,t}-x_{i,s}) + |\theta| \right) \geq 0 \\
    &\mathbf{1}\{p(Y_{t}=1|X=x_{i}) \geq 1-p(Y_{s-1}=1,Y_{s}=0|X=x_{i})\}\left(\beta (t-s) + (x_{i,t}-x_{i,s}) -\min\{0,\theta\} \right) \geq 0 \\
    &\mathbf{1}\{p(Y_{t}=1|X=x_{i}) \geq 1-P(Y_{s-1}=0,Y_{s}=0|X=x_{i})\}\left(\beta (t-s)+ (x_{i,t}-x_{i,s})+\max\{0,\theta\}\right) \geq 0 \\
    &\mathbf{1}\{p(Y_{t-1}=1,Y_{t}=1|X=x_{i})\geq P(Y_{s}=1|X=x_{i})\}\left(\beta (t-s) + (x_{i,t}-x_{i,s}) + \max\{0,\theta\}\right) \geq 0 \\
    &\mathbf{1}\{p(Y_{t-1}=1,Y_{t}=1|X=x_{i})\geq 1 - P(Y_{s-1}=1,Y_{s}=0|X=x_{i})\}(\beta (t-s) + (x_{i,t}-x_{i,s})) \geq 0 \\
    &\mathbf{1}\{p(Y_{t-1}=1,Y_{}=1|X=x_{i}) \geq 1-p(Y_{s-1}=0,Y_{s}=0|X=x_{i})\}\left(\beta (t-s) + (x_{i,t}-x_{i,s}) + \theta\right) \\
    &\mathbf{1}\{p(Y_{t-1}=0,Y_{t}=1|X=x_{i}) \geq p(Y_{s}=1|X=x_{i})\}\left((\beta (t-s) + (x_{i,t}-x_{i,s}) -\min\{0,\theta\}\right) \geq 0 \\
    &\mathbf{1}\{p(Y_{t-1}=0,Y_{t}=1|X=x_{i}) \geq 1-p(Y_{s-1}=1,Y_{s}=0|X=x_{i})\}\left(\beta (t-s) + (x_{i,t}-x_{i,s})-\theta\right) \geq 0 \\
    &\mathbf{1}\{p(Y_{t-1}=0,Y_{t}=1|X=x_{i}) \geq 1-p(Y_{s-1}=0,Y_{s}=0|X=x_{i})\}(\beta (t-s) + (x_{i,t}-x_{i,s})) \geq 0.
\end{align*}
Consequently, to determine whether a parameter value is in the identified set, there are $n \times 18$ inequality restrictions to check. This can be characterized as a union of restrictions compatible within Example \ref{ex:linsyst} by partitioning the parameter space based on the sign of $\theta$.
\subsection{Comparison with Discretized Identified Set and Full Support Identified Set}\label{ap:comparisonwithothersets}
Even though the binary choice model is stated in terms of continuous covariates, \cite{khan2023identification} binarize the covariates according to $\tilde{X}_{t} = \mathbf{1}\{X_{t} \geq Median(X_{t})\}$ for each $t=1,2$. Figure \ref{fig:conditionalidsetdisc} plots the conditional identified set based on the discretized covariates $\{\tilde{x}_{i}\}_{i=1}^{n}$. Comparing with Figure \ref{fig:conditionalidset}, these sets are considerably less informative; for instance, it is impossible to rule out no state dependence (i.e. $\theta = 0$).
\begin{figure}
    \centering
    \includegraphics[width=\linewidth]{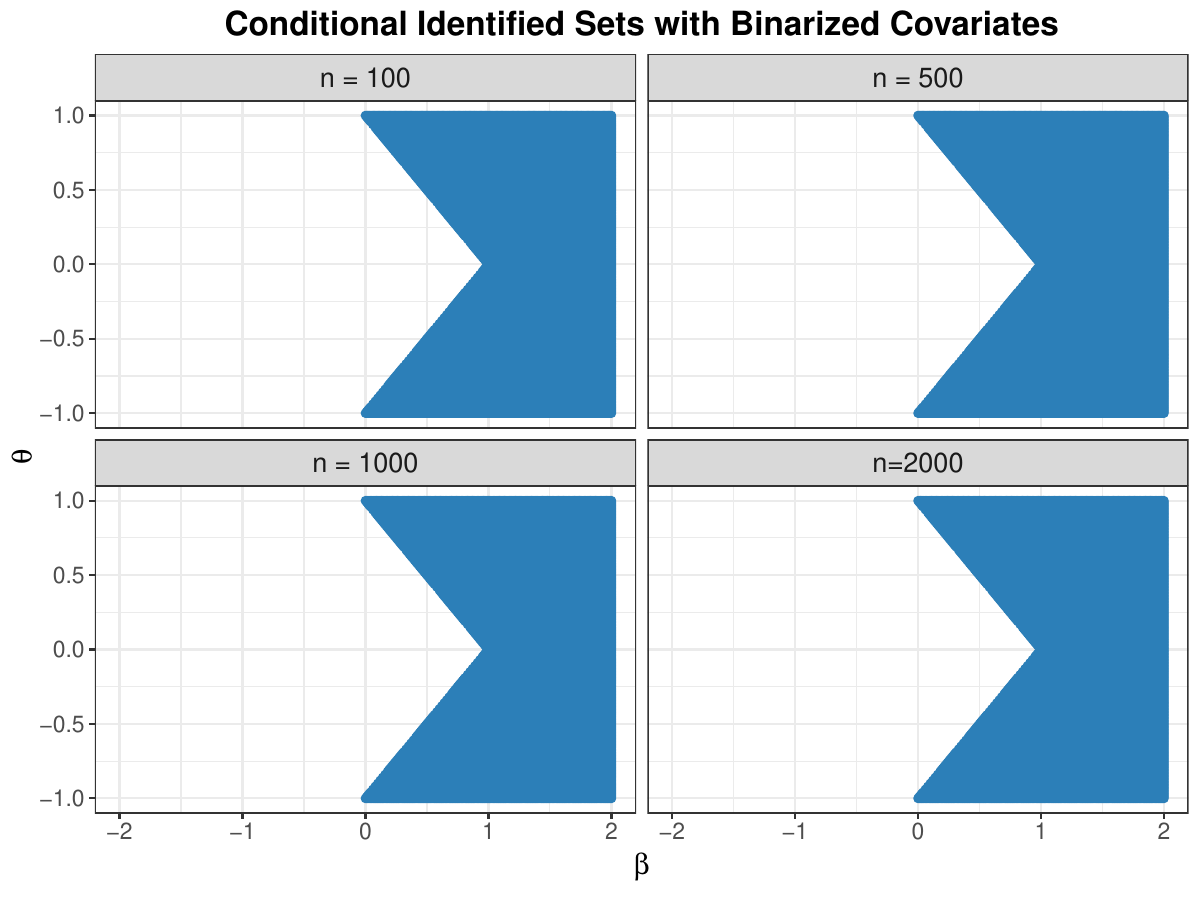}
    \caption{Conditional Identified Sets Based on Binarized Covariates}
    \label{fig:conditionalidsetdisc}
\end{figure}

We also compare the conditional identified set and the full support identified set. Although Theorem \ref{thm:fullsupport} states that these sets are asymptotically equivalent, the finite $n$ difference between $\Gamma_{n,I}(p_{0})$ and $\Gamma_{I}(p_{0})$ could be large. Figure \ref{fig:fullsupport} compares the conditional identified set based on $n=2000$ with the full support identified set.\footnote{Truly computing the full support identified set is infeasible, and, for that reason, we refer to the full support set as that which arises when we evaluate the conditional identified set with $n=20000$.} There is hardly any difference.
\begin{figure}
    \centering
    \includegraphics[width=\linewidth]{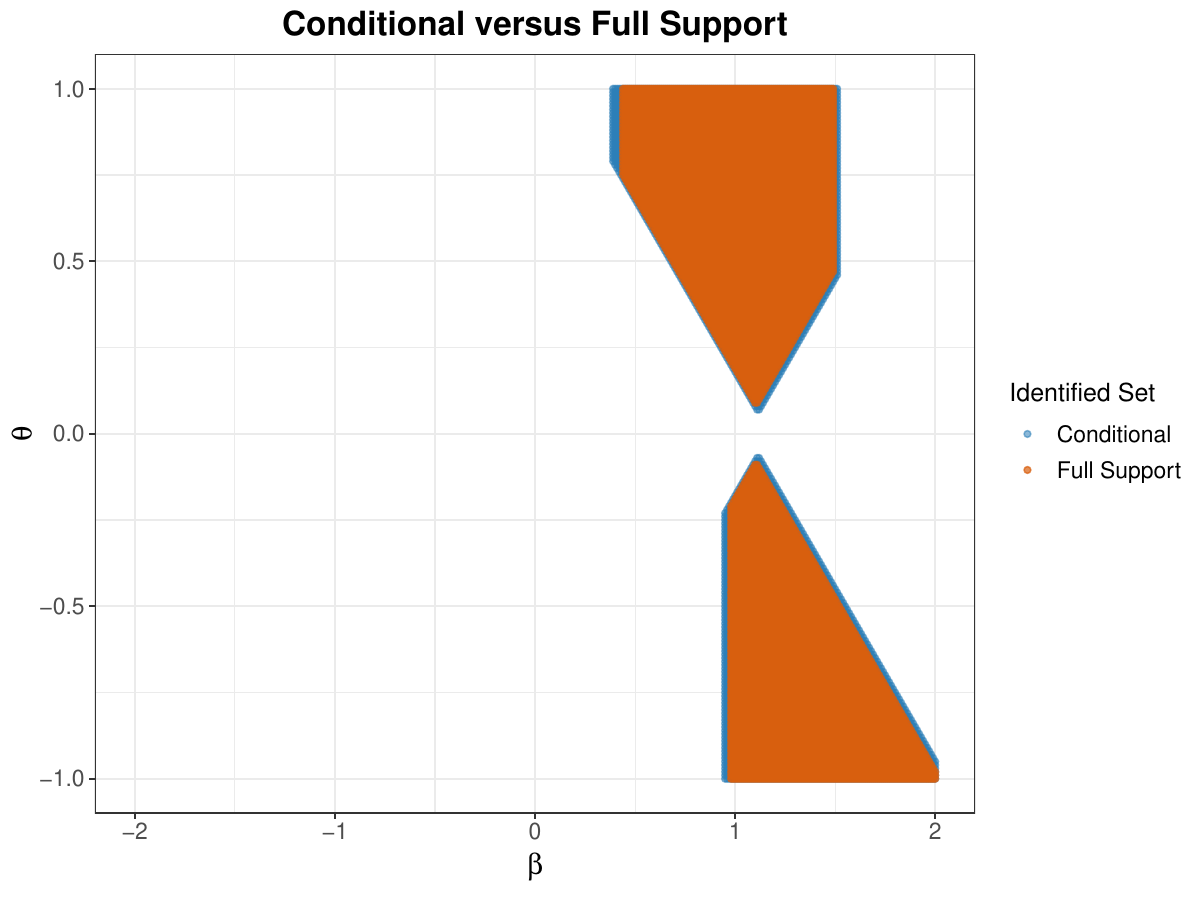}
    \caption{Conditional ($n=2000$) and Full Support Identified Sets}
    \label{fig:fullsupport}
\end{figure}
\section{Elaboration on the Discussion Following Proposition \ref{prop:misspec_hemi_convex}.}\label{ap:hemi_misspec}
\subsection{Uniform Convergence.}\label{ap:uconmisspec} Posterior concentration of $p$ around $p_{0}$ in $||\cdot||_{n,\infty}$ can satisfy Condition 3. Focusing on $Q_{n}^{LP}(\gamma,p)$ for demonstration (a conceptually similar argument applies to any $Q_{n,r}$, $r \in [1,\infty]$, with identity weighting), one can show 
\begin{align*}
Q_{n}^{LP}(\gamma,p) = \max_{a \in \{0,1\}^{nd_{f}}}q_{n,a}(\gamma,p),    
\end{align*}
where the component functions are 
\begin{align*}
q_{n,a}(\gamma,p) = -\frac{1}{n}\sum_{i=1}^{n}\sum_{j=1}^{d_f}a_{ij}\,f_j(x_i,p(x_i),\gamma).    
\end{align*} 
This implies that
\begin{align}\label{eq:ucon_misspec}
    \Delta_{n} \leq \sqrt{d_{f}}\sup_{(\gamma,p) \in \Gamma \times V_{n,\delta_{n}}}||f(x_{i},p(x_{i}),\gamma)-f(x_{i},p_{0}(x_{i}),\gamma)||_{n,\infty},
\end{align}
and
\begin{align}\label{eq:ucon_gradient_misspec}
    \Delta_{n}' \leq \sqrt{d_{f}}\sup_{(\gamma,p) \in \Gamma \times V_{n,\delta_{n}}}||\nabla_{\gamma} f(x_{i},p(x_{i}),\gamma)-\nabla_{\gamma}f(x_{i},p_{0}(x_{i}),\gamma)||_{n,\infty}
\end{align}
and similarly for $\Delta_{n}'$ (except replace with $\nabla_{\gamma} f$). Consequently, if the RHS of (\ref{eq:ucon_misspec}) and (\ref{eq:ucon_gradient_misspec}) are enveloped by a continuous function of $||p-p_{0}||_{n,\infty}$ that vanishes at $0$ (and similarly for $\Delta_{n}'$), then Condition 2 is satisfied when $V_{n,\delta_{n}}$ is contained in a shrinking $||\cdot||_{n,\infty}$-ball centered at $p_{0}$. This is satisfied, for instance, for conditional moment inequalities (i.e., Example \ref{ex:momentsineq}) when $g(y,x,\gamma)$ bounded and continuously differentiable in $\gamma$ (with bounded derivatives) because, in this case, the RHS of (\ref{eq:ucon_misspec}) and (\ref{eq:ucon_gradient_misspec}) admits an upper bound
\begin{align*}
\sqrt{d_{f}}\sup_{x \in  \mathcal{X}}\sup_{\gamma\in\Gamma}\max\left\{\sum_{y \in \mathcal{Y}}\|g(y,x,\gamma)\|_2,\sum_{y \in \mathcal{Y}}\|\nabla_\gamma g(y,x,\gamma)\|_{op}\right\}||p - p_{0}||_{n,\infty}.
\end{align*}
It is also satisfied by linear systems (i.e., Example \ref{ex:linsyst}) provided that $A(\cdot,p(\cdot))$ and $b(\cdot,p(\cdot))$ are continuous at $p_{0}$ with respect to the $||\cdot||_{n,\infty}$-topology.
\subsection{Complementarity of Condition 5 and Slater's Condition.}\label{ap:complementarity} When $\Gamma_{n,I}(p_{0}) = \emptyset$, the Slater condition from Proposition \ref{prop:convexhemi} fails because the existence of such a point $\gamma^{\circ}$ necessarily requires $\Gamma_{n,I}(p_{0}) \neq \emptyset$. To see why Condition 5 can fail under correct specification, consider $Q_{n}^{LP}$ again. For $\tilde{\gamma} \in \tilde{\Gamma}_{n,I}(p_{0})$, let $\mathcal{V}_n:=\left\{(i,j):f_j(x_i,p_0(x_i),\tilde{\gamma})<0\right\}$ and $\mathcal{B}_n:=\left\{(i,j):f_j(x_i,p_0(x_i),\tilde{\gamma})=0\right\}$ denote the violated and binding restrictions, respectively. One can show that the convex hull $\conv\{\nabla_\gamma q_{n,a}(\tilde \gamma,p_0):a\in \mathcal{A}_{n}^{*}(\tilde{\gamma},p_{0})\} $ is given by 
\begin{align}\label{eq:zonotope}
\Big\{-\frac{1}{n}\Big[\sum_{(i,j)\in\mathcal{V}_n}\nabla_\gamma f_j(x_i,p_0(x_i),\tilde{\gamma})+\sum_{(i,j)\in\mathcal{B}_n}w_{ij}\,\nabla_\gamma f_j(x_i,p_0(x_i),\tilde{\gamma})\Big]:\,w_{ij}\in[0,1]\Big\},
\end{align}
Condition 5 requires (\ref{eq:zonotope}) contain $\{t \in \mathbb{R}^{d_{\gamma}}: ||t||_{2} \leq \eta\}$, which means that the violated and binding gradients at the pseudo-true point must positively span $\mathbb{R}^{d_\gamma}$ with a uniform margin $\eta$. When $\Gamma_{n,I}(p_{0}) \neq \emptyset$ and $\Gamma_{n,I}(p_0)$ admits a Slater point, $\mathcal{V}_n=\emptyset$ and Condition 5 typically fails because (\ref{eq:zonotope}) contains the origin as a boundary point. 
\subsection{Failure of Lower Hemicontinuity with Flat Bottoms}\label{ap:failurehemicontinuity}
Condition 5 of Proposition \ref{prop:misspec_hemi_convex} implies that $\tilde{\Gamma}_{n,I}(p_0)$ is a singleton for large $n$. This property is formally proved in Lemma \ref{lem:singletonidset}, and is not an artifact of the proof: with a flat bottom, lower hemicontinuity of the exact pseudo-identified set can fail under Conditions 1--4. To see why, let $d_\gamma=1$, $\Gamma=[-1,2]$, and consider pieces $q_{n,1}(\gamma,p)=-\gamma$ and $q_{n,3}(\gamma,p)=\gamma-1$ for all $p$, together with $q_{n,2}(\gamma,p_0)=0$ and $q_{n,2}(\gamma,p)=\delta_n\gamma$ for $p\in V_{n,\delta_n}$, where $\delta_n\downarrow0$. Then $\tilde{\Gamma}_{n,I}(p_0)=[0,1]$ while $\tilde{\Gamma}_{n,I}(p)=\{0\}$, so $\dist(1,\tilde{\Gamma}_{n,I}(p))=1$ for every $n$, even though all pieces are linear and $\Delta_n\leq2\delta_n\to0$ and $\Delta'_n=\delta_n\to0$. Since the perturbation is compatible with any information about the criterion at $p_0$, no condition imposed at $p_0$ alone can restore lower hemicontinuity of the exact pseudo-identified set in this example; Condition 5 fails precisely at the interior points of $[0,1]$, where the only active piece has zero gradient. The next result says that flat bottoms can be accommodated if $\tilde{\Gamma}_{n,I}(p)$ is enlarged, the criterion satisfies uniform convergence, and a weak sharp minimum property holds at $p_{0}$.
\begin{proposition}
\label{prop:relaxed}
Suppose Assumptions 3 and 4.1 hold, and define, for each $t \geq 0$, the relaxed pseudo-identified set $\tilde{\Gamma}^{\,t}_{n,I}(p)\;:=\;\big\{\gamma\in\Gamma:\,\tilde{Q}_n(\gamma,p)\leq t\big\}$. Suppose that, for $P^{(\infty)}_{0,X}$-almost every fixed realization $\{x_i\}_{i\geq1}$ of $\{X_i\}_{i\geq1}$:
\begin{enumerate}
\item[1.] There exist $\eta>0$ and $\bar{N}\geq1$ such that
\[
\tilde{Q}_n(\gamma,p_0)\geq\eta\,\dist\big(\gamma,\tilde{\Gamma}_{n,I}(p_0)\big)\qquad\text{for all }\gamma\in\Gamma\text{ and all }n\geq\bar{N}.
\]
\item[2.] $u_n:=\sup_{(\gamma,p)\in\Gamma\times V_{n,\delta_n}}\big|\tilde{Q}_n(\gamma,p)-\tilde{Q}_n(\gamma,p_0)\big|\longrightarrow0$ as $n\to\infty$.
\end{enumerate}
Then, for $P^{(\infty)}_{0,X}$-almost every fixed realization $\{x_i\}_{i\geq1}$ of $\{X_i\}_{i\geq1}$, there exists $\hat{N} \geq 1$ such that,
\[
\tilde{\Gamma}_{n,I}(p_0)\subseteq\tilde{\Gamma}^{t}_{n,I}(p)
\qquad\text{and}\qquad
\sup_{p\in V_{n,\delta_n}}d_\mathcal{H}\big(\tilde{\Gamma}^{\,t}_{n,I}(p),\tilde{\Gamma}_{n,I}(p_0)\big)\leq\frac{2t}{\eta}
\]
for all $n \geq \hat{N}$.
\end{proposition}
\begin{remark}
Inspecting the proof, if $t:= t_{n}$ with $t_{n}\downarrow 0$ as $n\rightarrow \infty$ and $u_{n}/t_{n} \rightarrow 0$ as $n\rightarrow \infty$, then Proposition \ref{prop:relaxed} leads to consistency for $\tilde{\Gamma}_{n,I}(p_{0})$. This mirrors the slack sequences from \cite{chernozhukov2007estimation}. For DGPs in which $\eta >0$ is sufficiently large, a fixed, small value of $t$ likely results in an immaterial difference between $\tilde{\Gamma}_{n,I}^{t}$ and $\tilde{\Gamma}_{n,I}$.   
\end{remark}

\section{Proofs of Main Results}
\subsection{Proof of Theorems \ref{thm:randomset}, \ref{thm:pseudoset}, and \ref{thm:fullsupport}, and Related Results}
\begin{proof}[Proof of Theorem \ref{thm:randomset}]
The proof has two steps. Step 1 shows that $\Gamma_{n,I}(p)$ is a closed set for each $p \in \mathcal{P}$. Step 2 shows that $\{p:\Gamma_{n,I}(p) \cap A \neq \emptyset \} \in \mathscr{P}$ for each compact $A$.
\paragraph{Step 1.} Since the criterion $\gamma \mapsto Q_{n}(\gamma,p)$ is lower semicontinuous for each $p \in \mathcal{P}$ (Assumption \ref{as:criterion}.1), the preimage $\{\gamma \in \Gamma: Q_{n}(\gamma,p) > 0\}$ is an open subset of $\Gamma$ for each $p \in \mathcal{P}$. Moreover, since a closed set is the complement of an open set, $\{\gamma \in \Gamma: Q_{n}(\gamma,p) \leq 0\}$ is a closed subset of $\Gamma$ for each $p \in \mathcal{P}$, and, since $Q_{n}(\cdot,p) \geq 0$, we can further conclude that the level set $\{\gamma \in \Gamma: Q_{n}(\gamma,p) = 0\}$ is a closed subset of $\Gamma$ for each $p \in \mathcal{P}$. Hence, the range of $\Gamma_{n,I}: \mathcal{P} \rightarrow 2^{\Gamma}$ is contained in the closed subsets of $\Gamma$, and, as a result, we conclude that $\Gamma_{n,I}(p)$ is a closed set for each $p \in \mathcal{P}$.
\paragraph{Step 2.} Let $A$ be an arbitrary compact subset of $\Gamma$ and define $f_{n,A}: \mathcal{P} \rightarrow \mathbb{R}_{+}$ with $f_{n,A}(p) = \inf_{\gamma \in A}Q_{n}(\gamma,p)$ for each $p \in \mathcal{P}$. The infimum is achieved for each $p \in \mathcal{P}$, as $Q_{n}(\cdot,p):\Gamma \rightarrow \mathbb{R}_{+}$ is lower semicontinuous for each $p \in \mathcal{P}$ (Assumption \ref{as:criterion}.1) and $A$ is compact. Since $\Gamma_{n,I}(p) \cap A \neq \emptyset$ iff $f_{n,A}(p) = 0$, $\{p: \Gamma_{n,I}(p) \cap A \neq \emptyset\} = f_{n,A}^{-1}\{0\}$. As a closed subset of $\mathbb{R}_{+}$, the singleton $\{0\} \in \mathcal{B}(\mathbb{R}_{+})$. Consequently, if $p \mapsto \inf_{\gamma \in A}Q_{n}(\gamma,p)$ is measurable, then $f_{n,A}^{-1}\{0\} \in \mathscr{P}$, which is equivalent to $\{p: \Gamma_{n,I}(p) \cap A \neq \emptyset\} \in \mathscr{P}$. Measurability of $p\mapsto f_{n,A}(p)$ follows from Assumption \ref{as:criterion}.2. Since $A$ was arbitrary, we conclude that $\{p:\Gamma_{n,I}(p) \cap A \neq \emptyset \} \in \mathscr{P}$ for each compact $A$.
\end{proof}
\begin{proof}[Proof of Corollary \ref{cor:functionrandomset}]
Let $A$ be an arbitrary compact subset of $\mathbb{R}^{d_{g}}$. Since $G_{n,I}(p) \cap A \neq \emptyset$ if and only if $\Gamma_{n,I}(p) \cap g^{-1}(A) \neq \emptyset$, we just need to show that $g^{-1}(A)$ is a compact subset of $\Gamma$ because, in such case, Theorem \ref{thm:randomset} implies that $\{p: \Gamma_{n,I}(p) \cap g^{-1}(A) \neq \emptyset\} \in \mathscr{P}$. Compactness of $g^{-1}(A)$ follows because $A$ is compact, $\Gamma$ is compact, and $g: \Gamma \rightarrow \mathbb{R}^{d_{g}}$ is continuous. Since $A$ was fixed arbitrarily, we can repeat the argument to conclude that $\{p: G_{n,I}(p) \cap A \neq \emptyset\} \in \mathscr{P}$ for any compact $A \subseteq \mathbb{R}^{d_{g}}$.
\end{proof}
\begin{proof}[Proof of Theorem \ref{thm:pseudoset}]
We need to show that $\gamma \mapsto \tilde{Q}_{n}(\gamma,p)$ is lower semicontinuous for each $p \in \mathcal{P}$ and each $n \geq 1$, and $p \mapsto \inf_{\gamma \in A}\tilde{Q}_{n}(\gamma,p)$ is $\mathscr{P}$-measurable for each $n \geq 1$ and each compact $A \subseteq \Gamma$. If these conditions hold, then follow the same argument as Theorem \ref{thm:randomset} to conclude that $\tilde{\Gamma}_{n,I}$ defines a random set over $(\mathcal{P},\mathscr{P})$. Since, for each $p \in \mathcal{P}$, $\tilde{Q}_{n}(\cdot,p)$ is $Q_{n}(\cdot,p)$ shifted by a constant (that depends on $p$), the shifted criterion $\tilde{Q}_{n}: \Gamma \rightarrow \mathbb{R}_{+}$ is lower semicontinuous (the infimum $\inf_{\gamma \in \Gamma}Q_{n}(\gamma,p)$ exists by Assumptions \ref{as:parameterspace} and \ref{as:criterion}.1). The mapping $p \mapsto \inf_{\gamma \in \Gamma}Q_{n}(\gamma,p)$ is $\mathscr{P}$-measurable by Assumption \ref{as:parameterspace} and \ref{as:criterion}.2 (setting $A = \Gamma$), and, as a result, $p \mapsto \tilde{Q}_{n}(\gamma,p)$ is $\mathscr{P}$-measurable for each $\gamma \in \Gamma$ (as the sum of measurable functions).
\end{proof}
Before proving Theorem \ref{thm:fullsupport}, we introduce a notion of set convergence. A sequence of sets $\{A_{n}\}_{n \geq 1}$ are said to converge to $A$ in the \textit{Painlev\'{e}-Kuratowski} (PK) sense if $A \subseteq \liminf_{n\rightarrow \infty}A_{n}$ and $\limsup_{n\rightarrow \infty}A_{n} \subseteq A$.
\begin{proof}[Proof of Theorem \ref{thm:fullsupport}]
The proof verifies Proposition 1.7.23 in \cite{molchanov2017theory} and then uses an equivalence between PK convergence and convergence in the Hausdorff distance in compact metric spaces. Specifically, to employ Proposition 1.7.23, we show that $\{x_{i}\}_{i \geq 1} \mapsto \Gamma_{n,I}(p)$ defines a random closed set for each $n \geq 1$ and the following conditions hold:
\begin{enumerate}
    \item $P_{0,X}^{(\infty)}(\Gamma_{n,I}(p) \cap A \neq \emptyset \ i.o.) = 0$ for any compact $A\subseteq \Gamma$ such that $A \cap \Gamma_{I}(p) =\emptyset$.
    \item $P_{0,X}^{(\infty)}(\Gamma_{n,I}(p) \cap A^{o} = \emptyset \ i.o.)=0$ for any open $A^{o} \subseteq \Gamma$ such that $A^{o} \cap \Gamma_{I}(p) \neq \emptyset$.
\end{enumerate}
Step 1 and Step 2 verify these properties, thereby establishing almost-sure PK convergence of $\Gamma_{n,I}(p)$ to $\Gamma_{I}(p)$. Step 3 then establishes the desired Hausdorff convergence.
\paragraph{Step 1.} Since $\gamma \mapsto Q_{n}(\gamma,p)$ is lower semicontinuous for each $\{x_{i}\}_{i \geq 1}$ and each $n \geq 1$, we can follow the same argument as Step 1 of Theorem \ref{thm:randomset} to conclude that $\{x_{i}\}_{i \geq 1} \mapsto \Gamma_{n,I}(p)$ defines a mapping from $\mathcal{X}^{\infty}$ into the space of closed subsets of $\Gamma$ for each $n \geq 1$. To show that $\{\{x_{i}\}_{i \geq 1}: \Gamma_{n,I}(p) \cap A  \neq \emptyset\} \in \mathscr{X}^{\infty}$, we note that $ \Gamma_{n,I}(p) \cap A  \neq \emptyset$ is equivalent to $\inf_{\gamma \in A}Q_{n}(\gamma,p) = 0$, and, as a result, $\{\{x_{i}\}_{i \geq 1}: \Gamma_{n,I}(p) \cap A  \neq \emptyset\} \in \mathscr{X}^{\infty}$ if and only if $\{\{x_{i}\}_{i \geq 1}:\inf_{\gamma \in A}Q_{n}(\gamma,p) = 0\} \in \mathscr{X}^{\infty}$, which holds, like Step 2 of Theorem \ref{thm:randomset}, by Condition 2 of the theorem. Hence, under the first two assumptions of the theorem, we conclude that $\{x_{i}\}_{i \geq 1} \mapsto \Gamma_{n,I}(p)$ defines a random closed set on $(\mathcal{X}^{\infty},\mathscr{X}^{\infty})$ for each $n \geq 1$.
\paragraph{Step 2.} We verify the two conditions stated above.
\paragraph{Step 2A.} Let $A$ be a compact subset of $\Gamma$ such that $A \cap \Gamma_{I}(p) = \emptyset$. We show that
\begin{align*}
\sum_{n=1}^{\infty}P_{0,X}^{(\infty)}(\Gamma_{n,I}(p) \cap A \neq \emptyset) < \infty    
\end{align*}
so that $P_{0,X}^{(\infty)}(\Gamma_{n,I}(p) \cap A \neq \emptyset \ i.o.) = 0$ holds by the Borel-Cantelli lemma. On the event $\Gamma_{n,I}(p) \cap A \neq \emptyset$, there exists $\gamma \in A $ such that $\min_{1 \leq j \leq d_{f}}f_{j}(x_{i},p(x_{i}),\gamma) \geq 0$ for all $i=1,...,n$, where $f_{j}(x,p(x),\gamma)$ is an element of $f(x,p(x),\gamma)$. This implies $F_{A}(x_{i},p(x_{i})) \geq 0$ for all $i=1,...,n$, where $F_{A}(x_{i},p(x_{i})) = \sup_{\gamma \in A}\min_{1 \leq j \leq d_{f}}f_{j}(x_{i},p(x_{i}),\gamma)$. Note that $x \mapsto F_{A}(x,p(x))$ is measurable by assumption. Using $\{X_{i}\}_{i \geq 1}$ is i.i.d, we then conclude that $P_{0,X}^{(\infty)}(\Gamma_{n,I}(p) \cap A \neq \emptyset) \leq P_{0,X}(F_{A}(X,p(X)) \geq 0)^{n}$. Consequently, we want to show that $P_{0,X}(F_{A}(X,p(X))< 0) > 0$ to verify the first condition. Since $A \cap \Gamma_{I}(p)=\emptyset$ is equivalent to $A \subseteq \Gamma \setminus \Gamma_{I}(p)$, we can invoke the well-separatedness assumption
\begin{align*}
 c_{A}:=P_{0,X}\left(\sup_{\gamma \in A}\min_{1 \leq j \leq d_{f}}f_{j}(x,p(x),\gamma) < 0\right) > 0   
\end{align*}
to verify that Condition 1. holds because
\begin{align*}
    \sum_{n=1}^{\infty}P_{0,X}^{(\infty)}(\Gamma_{n,I}(p) \cap A \neq \emptyset) \leq \sum_{n=1}^{\infty}(1-c_{A})^{n} < \infty
\end{align*}
by the convergence of geometric series.
\paragraph{Step 2B.} Let $A^{o}$ be an open set such that $A^{o} \cap \Gamma_{I}(p) \neq \emptyset$. Since $\Gamma_{I}(p) \subseteq \Gamma_{n,I}(p)$ for all $n \geq 1$ and $P_{0,X}^{(\infty)}$-almost every fixed realization $\{x_{i}\}_{i \geq 1}$ of $\{X_{i}\}_{i \geq 1}$, we know that $\Gamma_{n,I}(p) \cap A^{o} \neq \emptyset$ for all $n \geq 1$ and $P_{0,X}^{(\infty)}$-almost every fixed realization $\{x_{i}\}_{i \geq 1}$ of $\{X_{i}\}_{i \geq 1}$. Consequently, $\{\{x_{i}\}_{i \geq 1}: \Gamma_{n,I}(p) \cap A^{o} = \emptyset \} = \emptyset$, and, as a result, $P_{0,X}^{(\infty)}(\Gamma_{n,I}(p) \cap A^{o} \ i.o.) = 0$ for each $p \in \mathcal{P}$. Since $A^{o}$ was an arbitrary open set with $A^{o} \cap \Gamma_{I}(p) = \emptyset$, Condition 2. holds.
\paragraph{Step 3.} Steps 1 and 2 verify that Proposition 1.7.23 in \cite{molchanov2017theory} holds, and, as a result, we conclude that $\Gamma_{n,I}(p)$ PK converges to $\Gamma_{I}(p)$ as $n\rightarrow \infty$ with $P_{0,X}^{(\infty)}$-probability equal to one. Since $\Gamma$ is compact, PK convergence is equivalent to convergence in the Hausdorff distance, and, as a result, we conclude that
\begin{align*}
    P_{0,X}^{(\infty)}\left(\lim_{n\rightarrow \infty}d_{\mathcal{H}}(\Gamma_{n,I}(p),\Gamma_{I}(p))=0\right) = 1
\end{align*}
for each $p \in \mathcal{P}$.

\end{proof}
\subsection{Proof of Theorems \ref{thm:consistency}, \ref{thm:test}, and \ref{thm:misspec_consistency}, and Related Results}
\begin{proof}[Proof of Theorem \ref{thm:consistency}]
The proof has three steps. Step 1 establishes that it suffices to condition on $p \in V_{n,\delta_{n}}$. Step 2 shows that it also suffices to condition on the event $\{\Gamma_{n,I}(p) \neq \emptyset\}$. Step 3 uses the results from the first two steps to establish the claim.
\paragraph{Step 1.} We first show that $\Pi_{n}(p \in V_{n,\delta_{n}}|Y^{(n)}) = 1+o_{P_{0}^{(n)}}(1)$ as $n\rightarrow \infty$, conditionally given $P_{0,X}^{(\infty)}$-almost every fixed realization $\{x_{i}\}_{i \geq 1}$ of $\{X_{i}\}_{i \geq 1}$, where $\{V_{n,\delta_{n}}\}_{n \geq 1}$ satisfies $V_{n,\delta_{n}} = \mathcal{P}_{n} \cap \{p \in \mathcal{P}: d_{n}(p,p_{0})< \delta_{n}\}$. By the complement rule and the union bound,
\begin{align*}
    \Pi_{n}(p \in V_{n,\delta_{n}}|Y^{(n)}) &= 1 - \Pi_{n}(p \in \mathcal{P} \setminus V_{n,\delta_{n}}|Y^{(n)}) \\
    &\geq  1 - \Pi_{n}(p \in \mathcal{P} \setminus \mathcal{P}_{n}|Y^{(n)}) - \Pi_{n}(p \in \mathcal{P}: d_{n}(p,p_{0}) \geq \delta_{n}|Y^{(n)})
\end{align*}
By Assumption \ref{as:prior}.1, the following holds for $P_{0,X}^{(\infty)}$-almost every fixed realization $\{x_{i}\}_{i \geq 1}$ of $\{X_{i}\}_{i \geq 1}$,
\begin{align*}
  \Pi_{n}(p \in \mathcal{P} \setminus \mathcal{P}_{n}|Y^{(n)}) \overset{P_{0}^{(n)}}{\longrightarrow} 0   
\end{align*}
and
\begin{align*}
\Pi_{n}(p \in \mathcal{P}: d_{n}(p,p_{0}) \geq \delta_{n}|Y^{(n)}) \overset{P_{0}^{(n)}}{\longrightarrow } 0    
\end{align*}
as $n\rightarrow \infty$. Consequently, for $P_{0,X}^{(\infty)}$-almost every fixed realization $\{x_{i}\}_{i \geq 1}$ of $\{X_{i}\}_{i \geq 1}$,
\begin{align}\label{eq:sievedneighborhood}
   \Pi_{n}(p \in V_{n,\delta_{n}}|Y^{(n)}) = 1+o_{P_{0}^{(n)}}(1)   
\end{align}
as $n\rightarrow \infty$. Applying of the law of total probability, for $P_{0,X}^{(\infty)}$-almost every fixed realization $\{x_{i}\}_{i \geq 1}$ of $\{X_{i}\}_{i \geq 1}$,
\begin{align*}
&\Pi_{n}(d_{\mathcal{H}}(\Gamma_{n,I}(p),\Gamma_{n,I}(p_{0})) \geq \varepsilon |Y^{(n)})  \\
&\quad = \Pi_{n}(d_{\mathcal{H}}(\Gamma_{n,I}(p),\Gamma_{n,I}(p_{0})) \geq \varepsilon |Y^{(n)},V_{n,\delta_{n}})\Pi_{n}(p \in V_{n,\delta_{n}}|Y^{(n)}) \\
&\quad \quad + \Pi_{n}(d_{\mathcal{H}}(\Gamma_{n,I}(p),\Gamma_{n,I}(p_{0})) \geq \varepsilon |Y^{(n)},V_{n,\delta_{n}}^{c})\Pi_{n}(p \in V_{n,\delta_{n}}^{c}|Y^{(n)}) \\
&\quad \leq   \Pi_{n}(d_{\mathcal{H}}(\Gamma_{n,I}(p),\Gamma_{n,I}(p_{0})) \geq \varepsilon |Y^{(n)},V_{n,\delta_{n}})(1+o_{P_{0}^{(n)}}(1)) + o_{P_{0}^{(n)}}(1). 
\end{align*}
and
\begin{align*}
&\Pi_{n}(d_{\mathcal{H}}(\Gamma_{n,I}(p),\Gamma_{n,I}(p_{0})) \geq \varepsilon |Y^{(n)}) \\
&\quad \geq  \Pi_{n}(d_{\mathcal{H}}(\Gamma_{n,I}(p),\Gamma_{n,I}(p_{0})) \geq \varepsilon |Y^{(n)},V_{n,\delta_{n}})(1+o_{P_{0}^{(n)}}(1))
\end{align*}
as $n\rightarrow \infty$, where the inequalities use that probabilities are bounded from above and below by one and zero, respectively, and the posterior contraction (\ref{eq:sievedneighborhood}). Hence, for $P_{0,X}^{(\infty)}$-almost every fixed realization $\{x_{i}\}_{i \geq 1}$ of $\{X_{i}\}_{i \geq 1}$,
\begin{align*}
&\Pi_{n}(d_{\mathcal{H}}(\Gamma_{n,I}(p),\Gamma_{n,I}(p_{0})) \geq \varepsilon |Y^{(n)}) \\
&\quad = \Pi_{n}(d_{\mathcal{H}}(\Gamma_{n,I}(p),\Gamma_{n,I}(p_{0})) \geq \varepsilon |Y^{(n)},V_{n,\delta_{n}})+ o_{P_{0}^{(n)}}(1),
\end{align*}
meaning that it suffices to show that, for $P_{0,X}^{(\infty)}$-almost every fixed realization $\{x_{i}\}_{i \geq 1}$ of $\{X_{i}\}_{i \geq 1}$, the following holds: for every $\varepsilon > 0$,
\begin{align*}
    \Pi_{n}(d_{\mathcal{H}}(\Gamma_{n,I}(p),\Gamma_{n,I}(p_{0})) \geq \varepsilon |Y^{(n)},V_{n,\delta_{n}}) = o_{P_{0}^{(n)}}(1)
\end{align*}
as $n\rightarrow \infty$.
\paragraph{Step 2.} We next show that that it suffices to show that, for $P_{0,X}^{(\infty)}$-almost every fixed realization $\{x_{i}\}_{i \geq 1}$ of $\{X_{i}\}_{i \geq 1}$, the following holds: for every $\varepsilon > 0$,
\begin{align*}
    \Pi_{n}(d_{\mathcal{H}}(\Gamma_{n,I}(p),\Gamma_{n,I}(p_{0})) \geq \varepsilon |Y^{(n)},V_{n,\delta_{n}} \cap \{\Gamma_{n,I}(p) \neq \emptyset\})= o_{P_{0}^{(n)}}(1)
\end{align*}
as $n\rightarrow \infty$. By the Law of Total Probability,
\begin{align*}
      &\Pi_{n}(d_{\mathcal{H}}(\Gamma_{n,I}(p),\Gamma_{n,I}(p_{0})) \geq \varepsilon |Y^{(n)},V_{n,\delta_{n}}) \\
      &\quad \leq  \Pi_{n}(d_{\mathcal{H}}(\Gamma_{n,I}(p),\Gamma_{n,I}(p_{0})) \geq \varepsilon |Y^{(n)},V_{n,\delta_{n}} \cap \{\Gamma_{n,I}(p) \neq \emptyset\})\Pi_{n}(\Gamma_{n,I}(p) \neq \emptyset | Y^{(n)},V_{n,\delta_{n}}) \\
      &\quad \quad + \Pi_{n}(\Gamma_{n,I}(p) = \emptyset | Y^{(n)},V_{n,\delta_{n}})
\end{align*}
and
\begin{align*}
 &\Pi_{n}(d_{\mathcal{H}}(\Gamma_{n,I}(p),\Gamma_{n,I}(p_{0})) \geq \varepsilon |Y^{(n)},V_{n,\delta_{n}}) \\
 &\quad \geq    \Pi_{n}(d_{\mathcal{H}}(\Gamma_{n,I}(p),\Gamma_{n,I}(p_{0})) \geq \varepsilon |Y^{(n)},V_{n,\delta_{n}} \cap \{\Gamma_{n,I}(p) \neq \emptyset\})\Pi_{n}(\Gamma_{n,I}(p) \neq \emptyset | Y^{(n)},V_{n,\delta_{n}})
\end{align*}
Under our assumptions, the proof of Theorem \ref{thm:test}.2 shows that, for $P_{0,X}^{(\infty)}$-almost every fixed realization $\{x_{i}\}_{i \geq 1}$ of $\{X_{i}\}_{i \geq 1}$,
\begin{align*}
    \Pi_{n}(\Gamma_{n,I}(p) \neq \emptyset | Y^{(n)},V_{n,\delta_{n}}) = 1 + o_{P_{0}^{(n)}}(1)
\end{align*}
as $n\rightarrow \infty$. Consequently,
\begin{align*}
     &\Pi_{n}(d_{\mathcal{H}}(\Gamma_{n,I}(p),\Gamma_{n,I}(p_{0})) \geq \varepsilon |Y^{(n)},V_{n,\delta_{n}}) \\
     &\quad =   \Pi_{n}(d_{\mathcal{H}}(\Gamma_{n,I}(p),\Gamma_{n,I}(p_{0})) \geq \varepsilon |Y^{(n)},V_{n,\delta_{n}} \cap \{\Gamma_{n,I}(p) \neq \emptyset\})(1+o_{P_{0}^{(n)}}(1)) + o_{P_{0}^{(n)}}(1)
\end{align*}
Consequently, it suffices to show that, for $P_{0,X}^{(\infty)}$-almost every fixed realization $\{x_{i}\}_{i \geq 1}$ of $\{X_{i}\}_{i \geq 1}$, the following holds: for every $\varepsilon > 0$,
\begin{align}\label{eq:restrictnonempty}
    \Pi_{n}(d_{\mathcal{H}}(\Gamma_{n,I}(p),\Gamma_{n,I}(p_{0})) \geq \varepsilon |Y^{(n)},V_{n,\delta_{n}} \cap \{\Gamma_{n,I}(p) \neq \emptyset\}) \overset{P_{0}^{(n)}}{\longrightarrow} 0
\end{align}
as $n\rightarrow \infty$.
\paragraph{Step 3.}
We show (\ref{eq:restrictnonempty}). First, the Hausdorff distance is defined as
\begin{align*}
    d_{\mathcal{H}}(\Gamma_{n,I}(p),\Gamma_{n,I}(p_{0})) = \max\left\{\sup_{\gamma \in \Gamma_{n,I}(p)}\text{dist}(\gamma,\Gamma_{n,I}(p_{0})), \sup_{\gamma \in \Gamma_{n,I}(p_{0})}\text{dist}(\gamma,\Gamma_{n,I}(p))\right\}
\end{align*}
Using this definition and the union bound, one can deduce that
\begin{align}
    &\Pi_{n}(d_{\mathcal{H}}(\Gamma_{n,I}(p),\Gamma_{n,I}(p_{0})) \geq \varepsilon |Y^{(n)},V_{n,\delta_{n}}\cap \{\Gamma_{n,I}(p) \neq \emptyset\}) \nonumber \\
    & \leq T_{\Gamma_{n,I}|Y^{(n)}, V_{n,\delta_{n}}\cap \{\Gamma_{n,I}(p) \neq \emptyset\}}(\Gamma \setminus \Gamma_{n,I}^{\varepsilon}(p_{0})) \nonumber \\
    &\quad + \Pi_{n}\left(\text{dist}(\gamma_{n}(p,p_{0}),\Gamma_{n,I}(p)) \geq \varepsilon \middle |Y^{(n)},V_{n,\delta_{n}}\cap \{\Gamma_{n,I}(p) \neq \emptyset\}\right), \label{eq:upperboundconsistency}
\end{align}
where $\gamma_{n}(p,p_{0}) \in \argmax_{\gamma \in \Gamma_{n,I}(p_{0})}\text{dist}(\gamma,\Gamma_{n,I}(p))$. Consequently, to complete the proof, we show that the terms of the upper bound converge to zero in $P_{0}^{(n)}$-probability as $n\rightarrow \infty$. Before doing this, we note that the probabilities in (\ref{eq:upperboundconsistency}) are well-defined (i.e., the events are $\mathscr{P}$-measurable) for the following reasons:
\begin{enumerate}
    \item \textit{First Term:} Since $\Gamma \setminus \Gamma_{n,I}^{\varepsilon}(p_{0})$ is compact (as closed subset of compact $\Gamma$) and $\Gamma_{n,I}$ is a random closed set over $(\mathcal{P},\mathscr{P})$, we deduce, from the definition of a random closed set, that $\{p:\Gamma_{n,I}(p) \cap (\Gamma \setminus \Gamma_{n,I}^{\varepsilon}(p_{0})) \neq \emptyset\} \in \mathscr{P}$ by setting $A = \Gamma \setminus \Gamma_{n,I}^{\varepsilon}(p_{0})$.
    \item \textit{Second Term:} The mapping $p \mapsto \text{dist}(\gamma,\Gamma_{n,I}(p))$ is random variable on $(\mathcal{P},\mathscr{P})$ for each $\gamma \in \Gamma_{n,I}(p_{0})$ and each $n \geq 1$ because $\{p: \text{dist}(\gamma,\Gamma_{n,I}(p)) \leq t\} = \{\Gamma_{n,I}(p) \cap \{\tilde{\gamma}: ||\gamma-\tilde{\gamma} ||_{2} \leq t\} \neq \emptyset \} \in \mathscr{P}$ for each $t \geq 0$ (since $\Gamma_{n,I}$ is a random closed set). Consequently, $p \mapsto \text{dist}(\gamma_{n}(p,p_{0}),\Gamma_{n,I}(p))$ is also a random variable on $(\mathcal{P},\mathscr{P})$ for each $n \geq 1$, using measurability of the supremum of a continuous function over a compact set, where the continuity follows from the general result that the distance of a point to a set being $1$-Lipschitz and $\Gamma_{n,I}(p_{0})$ is compact (as a closed subset of compact $\Gamma$).
\end{enumerate}
Given these probabilities are well-defined, we now use the following sub-steps to prove that the terms in the upper bound converge to zero in $P_{0}^{(n)}$-probability as $n\rightarrow \infty$.
\paragraph{Step 3A.} We now show that the first term in (\ref{eq:upperboundconsistency}) converges in probability to zero. Under Assumption \ref{as:dgp}.2, 
\begin{align*}
 &T_{\Gamma_{n,I}|Y^{(n)}, V_{n,\delta_{n}}\cap \{\Gamma_{n,I}(p) \neq \emptyset\}}(\Gamma \setminus \Gamma_{n,I}^{\varepsilon}(p_{0})) \\
 &\quad = \Pi_{n}\left(\Gamma_{n,I}(p) \cap \{\gamma \in \Gamma: \text{dist}(\gamma,\Gamma_{n,I}(p_{0})) \geq \varepsilon\} \neq \emptyset \middle | Y^{(n)},V_{n,\delta_{n}}\cap \{\Gamma_{n,I}(p) \neq \emptyset\}\right)   \\
 &\quad \leq \Pi_{n}\left(\Gamma_{n,I}(p) \cap \left\{\gamma \in \Gamma: Q_{n}(\gamma,p_{0}) \geq \xi \right\} \neq \emptyset \middle | Y^{(n)},V_{n,\delta_{n}}\cap \{\Gamma_{n,I}(p) \neq \emptyset\}\right).  
\end{align*}
for large $n$, where the event $\{\Gamma_{n,I}(p) \cap \left\{\gamma \in \Gamma: Q_{n}(\gamma,p_{0}) \geq \xi \right\} \neq \emptyset\} \in \mathscr{P}$ by Assumption \ref{as:dgp}.2 and that $\Gamma_{n,I}(p)$ is a random set. On the event $\Gamma_{n,I}(p) \cap \left\{\gamma \in \Gamma: Q_{n}(\gamma,p_{0}) \geq \xi \right\} \neq \emptyset$, there exists $\tilde{\gamma} \in \Gamma_{n,I}(p)$ such that $|Q_{n}(\tilde{\gamma},p_{0})-Q_{n}(\tilde{\gamma},p)| \geq \xi$, where the centering reflects that $Q_{n}(\tilde{\gamma},p)= 0$ for any $\tilde{\gamma} \in \Gamma_{n,I}(p)$. Then, for $ p \in V_{n,\delta_n}$ taking the supremum over $\Gamma \times  V_{n,\delta_{n}}$ and using that $\Gamma_{n,I}(p) \subseteq \Gamma$, we can further conclude that $\Gamma_{n,I}(p) \cap \left\{\gamma: Q_{n}(\gamma,p_{0}) \geq \xi \right\} \neq \emptyset$ implies that $\xi\leq \sup_{(\gamma,p) \in \Gamma \times V_{n,\delta_{n}}}\left|Q_{n}(\gamma,p)-Q_{n}(\gamma,p_{0})\right|$. Consequently,
\begin{align}\label{eq:bound11}
T_{\Gamma_{n,I}|Y^{(n)}, V_{n,\delta_{n}}\cap \{\Gamma_{n,I}(p) \neq \emptyset\}}(\Gamma \setminus \Gamma_{n,I}^{\varepsilon}(p_{0})) \leq \mathbf{1}\left\{\sup_{(\gamma,p) \in \Gamma \times V_{n,\delta_{n}}}\left|Q_{n}(\gamma,p)-Q_{n}(\gamma,p_{0})\right| \geq \xi\right\}    
\end{align}
Applying Assumption \ref{as:prior}.2, the upper bound of (\ref{eq:bound11}) is identically zero for $n$ large, conditionally given $P_{0,X}^{(\infty)}$-almost every fixed realization $\{x_{i}\}_{i \geq 1}$ of $\{X_{i}\}_{i \geq 1}$. Consequently,
\begin{align*}
 T_{\Gamma_{n,I}|Y^{(n)}, V_{n,\delta_{n}}\cap \{\Gamma_{n,I}(p) \neq \emptyset\}}(\Gamma \setminus \Gamma_{n,I}^{\varepsilon}(p_{0})) \overset{P_{0}^{(n)}}{\longrightarrow} 0    
\end{align*}
as $n\rightarrow \infty$, conditionally given $P_{0,X}^{(\infty)}$-almost every realization $\{x_{i}\}_{i \geq 1}$ of $\{X_{i}\}_{i \geq 1}$.
\paragraph{Step 3B.} We show that, for $P_{0,X}^{(\infty)}$-almost every fixed realization $\{x_{i}\}_{i \geq 1}$ of $\{X_{i}\}_{i \geq 1}$,
\begin{align*}
  \Pi_{n}\left(\text{dist}(\gamma_{n}(p,p_{0}),\Gamma_{n,I}(p)) \geq \varepsilon \middle |Y^{(n)},V_{n,\delta_{n}}\cap \{\Gamma_{n,I}(p) \neq \emptyset\}\right) = o_{P_{0}^{(n)}}(1)
\end{align*}
as $n\rightarrow \infty$. Applying the complement rule,
\begin{align*}
     &\Pi_{n}\left(\text{dist}(\gamma_{n}(p,p_{0}),\Gamma_{n,I}(p)) \geq \varepsilon \middle |Y^{(n)},V_{n,\delta_{n}}\cap \{\Gamma_{n,I}(p) \neq \emptyset\}\right)  \\
     &\quad = 1 -  \Pi_{n}\left(\text{dist}(\gamma_{n}(p,p_{0}),\Gamma_{n,I}(p)) < \varepsilon \middle |Y^{(n)},V_{n,\delta_{n}}\cap \{\Gamma_{n,I}(p) \neq \emptyset\}\right) 
\end{align*}
Since $\Gamma_{n,I}(p_{0})$ is compact and $\gamma\mapsto \text{dist}(\gamma,\Gamma_{n,I}(p))$ is continuous for every $p \in \mathcal{P}$, $\text{dist}(\gamma_{n}(p,p_{0}),\Gamma_{n,I}(p)) < \varepsilon$ holds if and only if $\inf_{\tilde{\gamma} \in \Gamma_{n,I}(p)}||\gamma-\tilde{\gamma}||_{2} < \varepsilon$ for all $\gamma \in \Gamma_{n,I}(p_{0})$. Moreover, $\inf_{\tilde{\gamma} \in \Gamma_{n,I}(p)}||\gamma-\tilde{\gamma}||_{2} < \varepsilon$ for all $\gamma \in \Gamma_{n,I}(p_{0})$ is equivalent to $\Gamma_{n,I}(p) \cap B(\gamma,\varepsilon) \neq \emptyset$ for all $\gamma \in \Gamma_{n,I}(p_{0})$. Hence, by Assumption \ref{as:prior}.3, the probability $\Pi_{n}\left(\text{dist}(\gamma_{n}(p,p_{0}),\Gamma_{n,I}(p)) < \varepsilon \middle |Y^{(n)},V_{n,\delta_{n}}\cap \{\Gamma_{n,I}(p) \neq \emptyset\}\right) $ is identically equal to $1$ for $n$ large, and as a result, we conclude that
\begin{align*}
\Pi_{n}\left(\text{dist}(\gamma_{n}(p,p_{0}),\Gamma_{n,I}(p)) \geq \varepsilon \middle |Y^{(n)},V_{n,\delta_{n}}\cap \{\Gamma_{n,I}(p) \neq \emptyset\}\right) = 1 - (1+o_{P_{0}^{(n)}}(1)) = o_{P_{0}^{(n)}}(1)
\end{align*}
conditionally given $P_{0,X}^{(\infty)}$-almost every fixed realization $\{x_{i}\}_{i \geq 1}$ of $\{X_{i}\}_{i \geq 1}$.
\end{proof}
\begin{proof}[Proof of Corollary \ref{cor:subvectors}]
Let $\varepsilon > 0 $ be fixed arbitrarily. Since $\Gamma$ is compact and $g: \mathbb{R}^{d_{\gamma}} \rightarrow \mathbb{R}^{d_{g}}$ is continuous, then $g$ is uniformly continuous on $\Gamma$. Consequently, there exists $\delta_{\varepsilon} > 0 $ such that $||g(\gamma)-g(\tilde{\gamma})||_{2} < \varepsilon$ for all $\gamma,\tilde{\gamma} \in \Gamma$ with $||\gamma-\tilde{\gamma} ||_{2} < \delta_{\varepsilon}$. Hence, if $d_{\mathcal{H}}(\Gamma_{n,I}(p),\Gamma_{n,I}(p_{0})) < \delta_{\varepsilon}$, then $
\text{dist}(g(\gamma),g(\Gamma_{n,I}(p_{0}))) < \varepsilon$ for all $\gamma \in \Gamma_{n,I}(p)$ and $\text{dist}(g(\gamma),g(\Gamma_{n,I}(p))) < \varepsilon$ for all $\gamma \in \Gamma_{n,I}(p_{0})$. Using that $G_{n,I}(p) = g(\Gamma_{n,I}(p))$ and $G_{n,I}(p_{0}) = g(\Gamma_{n,I}(p_{0}))$, it follows that $d_{\mathcal{H}}(G_{n,I}(p),G_{n,I}(p_{0})) < \varepsilon$. Since Theorem \ref{thm:consistency} guarantees that, for $P_{0,X}^{(\infty)}$-almost every fixed realization $\{x_{i}\}_{i \geq 1}$ of $\{X_{i}\}_{i\geq 1}$, the posterior satisfies $\Pi_{n}(d_{\mathcal{H}}(\Gamma_{n,I}(p),\Gamma_{n,I}(p_{0})) < \delta_{\varepsilon}|Y^{(n)}) = 1+o_{P_{0}^{(n)}}(1)$, it follows that the marginal posterior for $G_{n,I}(p)$ has the property that $\Pi_{n}(d_{\mathcal{H}}(G_{n,I}(p),G_{n,I}(p_{0})) <  \varepsilon |Y^{(n)})$ is identitically zero for large $n$, conditionally given $P_{0,X}^{(\infty)}$-almost every fixed realization $\{x_{i}\}_{i \geq 1}$ of $\{X_{i}\}_{i\geq 1}$.
\end{proof}
\begin{proof}[Proof of Corollary \ref{cor:fullsupportidset}]
The proof has two steps. Step 1 verifies the first claim, while Step 2 verifies the second.
\paragraph{Step 1.}
Let $\varepsilon > 0$ be fixed arbitrarily and note that, under the Assumptions of Theorem \ref{thm:fullsupport}, $P_{0,X}^{(\infty)}(A_{\varepsilon}) = 1$, where 
\begin{align*}
A_{\varepsilon} = \left\{\{x_{i}\}_{i \geq 1}: \exists \ N \geq 1 \ s.t. \ d_{\mathcal{H}}(\Gamma_{n,I}(p_{0}),\Gamma_{I}(p_{0})) < \frac{\varepsilon}{2} \ \forall \ n \geq N\right\}.    
\end{align*}
Consequently, we can condition on realizations $\{x_{i}\}_{i \geq 1}$ of $\{X_{i}\}_{i \geq 1}$ such that $\{x_{i}\}_{i \geq 1} \in A_{\varepsilon}$. Hence, by the triangle inequality, for $P_{0,X}^{(\infty)}$-almost every fixed realization $\{x_{i}\}_{i \geq 1}$ of $\{X_{i}\}_{i \geq 1}$, there is an $N \geq 1$ such that
\begin{align*}
    d_{\mathcal{H}}\left(\Gamma_{n,I}(p),\Gamma_{I}(p_{0})\right) &\leq d_{\mathcal{H}}\left(\Gamma_{n,I}(p),\Gamma_{n,I}(p_{0})\right) + d_{\mathcal{H}}\left(\Gamma_{n,I}(p_{0}),\Gamma_{I}(p_{0})\right) \\
    &\leq d_{\mathcal{H}}\left(\Gamma_{n,I}(p),\Gamma_{n,I}(p_{0})\right) + \frac{\varepsilon}{2}
\end{align*}
holds for each $n \geq N$. Hence, for $P_{0,X}^{(\infty)}$-almost every fixed realization $\{x_{i}\}_{i \geq 1}$ of $\{X_{i}\}_{i \geq 1}$, there is an $N \geq 1$ such that
\begin{align*}
    \Pi_{n}\left(d_{\mathcal{H}}(\Gamma_{n,I}(p),\Gamma_{I}(p_{0})) \geq \varepsilon \middle | Y^{(n)}\right) \leq \Pi_{n}\left(d_{\mathcal{H}}\left(\Gamma_{n,I}(p),\Gamma_{n,I}(p_{0})\right) \geq \frac{\varepsilon}{2} \middle | Y^{(n)}\right)
\end{align*}
for all $n \geq N$. Applying Theorem \ref{thm:consistency}, for $P_{0,X}^{(\infty)}$-almost every fixed realization $\{x_{i}\}_{i \geq 1}$ of $\{X_{i}\}_{i \geq 1}$,
\begin{align*}
    \Pi_{n}\left(d_{\mathcal{H}}\left(\Gamma_{n,I}(p),\Gamma_{n,I}(p_{0})\right) \geq \frac{\varepsilon}{2} \middle | Y^{(n)}\right) \overset{P_{0}^{(n)}}{\longrightarrow} 0
\end{align*}
as $n\rightarrow \infty$. Consequently, for $P_{0,X}^{(\infty)}$-almost every fixed realization $\{x_{i}\}_{i \geq 1}$ of $\{X_{i}\}_{i \geq 1}$,
\begin{align*}
    \Pi_{n}\left(d_{\mathcal{H}}(\Gamma_{n,I}(p),\Gamma_{I}(p_{0})) \geq \varepsilon \middle | Y^{(n)}\right) \overset{P_{0}^{(n)}}{\longrightarrow} 0
\end{align*}
as $n\rightarrow \infty$.
\paragraph{Step 2.} To extend the result to functions of $\gamma$, it suffices to show that
\begin{align}\label{eq:conditionforg}
    P_{0,X}^{(\infty)}\left(\lim_{n\rightarrow \infty}d_{\mathcal{H}}\left(G_{n,I}(p_{0}),G_{I}(p_{0})\right) = 0 \right) = 1
\end{align}
because if this holds, then the same argument as Step 1 can be applied, except replacing $\Gamma_{n,I}$ and $\Gamma_{I}$ with $G_{n,I}$ and $G_{I}$, respectively, and invoking Corollary \ref{cor:subvectors}. Since the closed set $G_{n,I}(p_{0})$ satisfies $\left\{\{x_{i}\}_{i \geq 1}: G_{n,I}(p_{0}) \cap A \neq \emptyset\right\} = \left\{\{x_{i}\}_{i \geq 1}: \Gamma_{n,I}(p_{0}) \cap g^{-1}(A)\right\}$ for any compact $A$ and $g^{-1}(A)$ is compact (as $g$ is a continuous mapping with compact domain $\Gamma$), it follows that  $\left\{\{x_{i}\}_{i \geq 1}: G_{n,I}(p_{0}) \cap A \neq \emptyset\right\}  \in \mathscr{X}^{\infty}$ for each $n \geq 1$. Hence, $\{x_{i}\}_{i \geq 1} \mapsto G_{n,I}(p_{0})$ is a random closed set over $(\mathcal{X}^{\infty},\mathscr{X}^{\infty})$ for each $n \geq 1$, and, as a result, the probability (\ref{eq:conditionforg}) is well-defined. Next, since $g$ is continuous and defined on a compact set $\Gamma$, $g$ is uniformly continuous, and, as a result, we obtain
\begin{align*}
\lim_{n\rightarrow \infty}d_{\mathcal{H}}(\Gamma_{n,I}(p_{0}),\Gamma_{I}(p_{0})) = 0\implies \lim_{n\rightarrow \infty}d_{\mathcal{H}}(G_{n,I}(p_{0}),G_{I}(p_{0})) = 0.    
\end{align*}
Consequently, we invoke Theorem \ref{thm:fullsupport} and follow the same argument as Step 1 to conclude that (\ref{eq:conditionforg}) holds.
\end{proof}
\begin{proof}[Proof of Theorem \ref{thm:test}]
The proof has two steps with each corresponding to the appropriate part of the theorem.
\paragraph{Step 1.} By a similar argument to Step 1 of Theorem \ref{thm:consistency}, it suffices to show that, for $P_{0,X}^{(\infty)}$-almost every fixed realization $\{x_{i}\}_{i \geq 1}$ of $\{X_{i}\}_{i \geq 1}$, 
\begin{align*}
\Pi_{n}(\Gamma_{n,I}(p) = \emptyset | Y^{(n)},V_{n,\delta_{n}}) = 1 + o_{P_{0}^{(n)}}(1).    
\end{align*} 
By assumption, there is a constant $c_{0} > 0$ such that $\inf_{\gamma \in \Gamma}Q_{n}(\gamma,p_{0}) \geq c_{0}$ for $n$ large. Consequently,
\begin{align*}
    &\Pi_{n}(\Gamma_{n,I}(p) = \emptyset | Y^{(n)},V_{n,\delta_{n}}) \\
    &\quad \geq \Pi_{n}\left( \inf_{\gamma \in \Gamma}Q_{n}(\gamma,p) \geq \frac{c_{0}}{2}\middle |Y^{(n)},V_{n,\delta_{n}}\right) \\
    &\quad \geq \Pi_{n}\left( \inf_{\gamma \in \Gamma}\left[Q_{n}(\gamma,p)-Q_{n}(\gamma,p_{0})\right] + \inf_{\gamma \in \Gamma}Q_{n}(\gamma,p_{0}) \geq \frac{c_{0}}{2} \middle |Y^{(n)},V_{n,\delta_{n}}\right) \\
    &\quad =\Pi_{n}\left( \inf_{\gamma \in \Gamma}Q_{n}(\gamma,p_{0}) \geq \frac{c_{0}}{2}-\inf_{\gamma \in \Gamma}\left[Q_{n}(\gamma,p)-Q_{n}(\gamma,p_{0})\right]  \middle |Y^{(n)},V_{n,\delta_{n}}\right) \\
    &\quad \geq \mathbf{1}\left\{\inf_{\gamma \in \Gamma}Q_{n}(\gamma,p_{0})\geq c_{0}\right\}\Pi_{n}\left(\inf_{\gamma \in \Gamma}\left[Q_{n}(\gamma,p)-Q_{n}(\gamma,p_{0})\right] \geq -\frac{c_{0}}{2}\middle | Y^{(n)},V_{n,\delta_{n}}\right) \\
    &\quad = 1+o_{P_{0}^{(n)}}(1),
\end{align*}
where the second inequality uses that the infimum of a sum is greater than the sum of the infima, and the third inequality uses the law of total probability and the fact that probabilities are bounded from below by zero, and the final equality uses that $\inf_{\gamma \in \Gamma}Q_{n}(\gamma,p_{0}) \geq c_{0}$ for $n$ large and Assumption \ref{as:prior}.2.
\paragraph{Step 2.} Following a similar argument to Step 1 of Theorem \ref{thm:consistency}, it suffices to show that, for $P_{0,X}^{(\infty)}$-almost every fixed realization $\{x_{i}\}_{i \geq 1}$ of $\{X_{i}\}_{i \geq 1}$ $\Pi_{n}(\Gamma_{n,I}(p) \neq \emptyset | Y^{(n)},V_{n,\delta_{n}}) = 1 + o_{P_{0}^{(n)}}(1)$. Applying Assumption \ref{as:prior}.3 with $\varepsilon = 1$, there exists $N_{1} \geq 1$ such that $p \in V_{n,\delta_{n}}$ implies that $\Gamma_{n,I}(p) \cap \{\tilde{\gamma} \in \Gamma: ||\tilde{\gamma}-\gamma_{n,0}||_{2} < 1\} \neq \emptyset$ for all $n \geq N_{1}$, where $\gamma_{n,0}$ is an arbitrary element of $\Gamma_{n,I}(p_{0}) \neq \emptyset$. Consequently, $\Gamma_{n,I}(p) \neq \emptyset$ for all $n \geq N_{1}$ on the event $p \in V_{n,\delta_{n}}$. Hence, for large $n$, $\Pi_{n}(\Gamma_{n,I}(p) \neq \emptyset | Y^{(n)},V_{n,\delta_{n}})$ is identically equal to one, establishing the claim.

\end{proof}
\begin{proof}[Proof of Theorem \ref{thm:misspec_consistency}.]
Let $\varepsilon > 0$ be fixed arbitrarily. Using the exact same argument as Step 1 of Theorem \ref{thm:consistency}, it suffices to show that, for $P_{0,X}^{(\infty)}$-almost every fixed realization $\{x_{i}\}_{i \geq 1}$ of $\{X_{i}\}_{i \geq 1}$,
\begin{align*}
    \Pi_{n}(d_{\mathcal{H}}(\tilde{\Gamma}_{n,I}(p),\tilde{\Gamma}_{n,I}(p_{0})) \geq \varepsilon |Y^{(n)},V_{n,\delta_{n}}) = o_{P_{0}^{(n)}}(1)
\end{align*}
as $n\rightarrow \infty$. Using the definition of $d_{\mathcal{H}}$, we can refine this requirement to showing that, for $P_{0,X}^{(\infty)}$-almost every fixed realization $\{x_{i}\}_{i \geq 1}$ of $\{X_{i}\}_{i \geq 1}$,
\begin{align}\label{eq:mod2A}
 T_{\tilde{\Gamma}_{n,I}|Y^{(n)}, V_{n,\delta_{n}}}(\Gamma \setminus \tilde{\Gamma}_{n,I}^{\varepsilon}(p_{0})) \overset{P_{0}^{(n)}}{\longrightarrow} 0 
\end{align}
and
\begin{align}\label{eq:mod2B}
\Pi_{n}\left(\text{dist}(\tilde{\gamma}_{n}(p,p_{0}),\tilde{\Gamma}_{n,I}(p)) \geq \varepsilon \middle |Y^{(n)},V_{n,\delta_{n}}\right) \overset{P_{0}^{(n)}}{\longrightarrow} 0 
\end{align}
as $n\rightarrow \infty$, where $\tilde{\gamma}_{n}(p,p_{0}) = \argmax_{\gamma \in \tilde{\Gamma}_{n,I}(p_{0})}\text{dist}(\gamma,\tilde{\Gamma}_{n,I}(p))$. Applying the modified Assumption \ref{as:dgp}.2, the same arguments from Step 3A of Theorem \ref{thm:consistency} can be applied (except with $\tilde{Q}_{n}(\gamma,p)$ in place of $Q_{n}(\gamma,p)$) to conclude that there exists $\xi > 0$ (depending on $\varepsilon$) such that
\begin{align*}
    T_{\tilde{\Gamma}_{n,I}|Y^{(n)}, V_{n,\delta_{n}}}(\Gamma \setminus \tilde{\Gamma}_{n,I}^{\varepsilon}(p_{0})) \leq \mathbf{1}\left\{2\sup_{(\gamma,p) \in \Gamma \times V_{n,\delta_{n}}}|Q_{n}(\gamma,p) - Q_{n}(\gamma,p_{0})| \geq \xi \right\},
\end{align*}
and, by Assumption \ref{as:prior}.2, we conclude that (\ref{eq:mod2A}) holds. Since $\tilde{\Gamma}_{n,I}(p)$, $p \in \mathcal{P}$, and $\tilde{\Gamma}_{n,I}(p_{0})$ are compact subsets of $\Gamma$ and the modified Assumption \ref{as:prior}.3 holds, the same argument to Step 3B of Theorem \ref{thm:consistency} implies (\ref{eq:mod2B}). Since $\varepsilon >0$ is arbitrary, the proof is complete. 
\end{proof}
\begin{remark}
The proof of Theorem \ref{thm:misspec_consistency} does not require conditioning on $\tilde{\Gamma}_{n,I}(p) \neq \emptyset$ because this event receives probability one by the construction of $\tilde{\Gamma}_{n,I}(p)$.
\end{remark}
\section{Proof of Propositions}
\subsection{Proofs for Sections \ref{sec:suff1} and \ref{sec:suffhemi}, and Appendix \ref{ap:hemi_misspec}}
\subsubsection{Proofs for Section \ref{sec:suff1}}
\begin{proof}[Proof of Proposition \ref{prop:ucon}]
The proof has three steps. The first establishes a sufficient conditions for the claim. The remaining steps (i.e., Steps 2 and 3) verify these conditions.
\paragraph{Step 1.} Since $Q_{n,r}(\gamma,p) = ||\text{dist}_{W(\cdot,p(\cdot),\gamma)}(f(\cdot,p(\cdot),\gamma),\mathbb{R}_{+}^{d_{f}})||_{n,r}$, the triangle inequality and properties of the supremum over sums reveals that
\begin{align}
&\sup_{(\gamma,p) \in \Gamma \times V_{n,\delta_{n}}}\left|Q_{n,r}(\gamma,p)- Q_{n,r}(\gamma,p_{0})\right| \nonumber \\
&\leq \sup_{(\gamma,p) \in \Gamma \times V_{n,\delta_{n}}}\left| ||\text{dist}_{W(\cdot,p(\cdot),\gamma)}(f(\cdot,p(\cdot),\gamma),\mathbb{R}_{+}^{d_{f}})||_{n,r}- ||\text{dist}_{W(\cdot,p(\cdot),\gamma)}(f(\cdot,p_{0}(\cdot),\gamma),\mathbb{R}_{+}^{d_{f}})||_{n,r}\right| \nonumber \\
&+ \sup_{(\gamma,p) \in \Gamma \times V_{n,\delta_{n}}}\left|||\text{dist}_{W(\cdot,p(\cdot),\gamma)}(f(\cdot,p_{0}(\cdot),\gamma),\mathbb{R}_{+}^{d_{f}})||_{n,r}-||\text{dist}_{W(\cdot,p_{0}(\cdot),\gamma)}(f(\cdot,p_{0}(\cdot),\gamma),\mathbb{R}_{+}^{d_{f}})||_{n,r}\right| \nonumber \\
&\leq \sup_{(\gamma,p) \in \Gamma \times V_{n,\delta_{n}}}\left|\left| \text{dist}_{W(\cdot,p(\cdot),\gamma)}(f(\cdot,p(\cdot),\gamma),\mathbb{R}_{+}^{d_{f}})-\text{dist}_{W(\cdot,p(\cdot),\gamma)}(f(\cdot,p_{0}(\cdot),\gamma),\mathbb{R}_{+}^{d_{f}})\right|\right|_{n,r} \label{eq:bound1}\\
&+\sup_{(\gamma,p) \in \Gamma \times V_{n,\delta_{n}}}\left|\left| \text{dist}_{W(\cdot,p(\cdot),\gamma)}(f(\cdot,p_{0}(\cdot),\gamma),\mathbb{R}_{+}^{d_{f}})-\text{dist}_{W(\cdot,p_{0}(\cdot),\gamma)}(f(\cdot,p_{0}(\cdot),\gamma),\mathbb{R}_{+}^{d_{f}})\right|\right|_{n,r} \label{eq:bound2}
\end{align}
The remaining steps show that the terms in the upper bound converges to zero as $n\rightarrow \infty$.
\paragraph{Step 2.} We show that (\ref{eq:bound1}) converges to zero as $n\rightarrow \infty$, uniformly over $\Gamma \times V_{n,\delta_{n}}$. Let 
\begin{align*}
C_{1}(x_{i},p(x_{i}),\gamma))=\left | \left| W^{\frac{1}{2}}(x_{i},p(x_{i}),\gamma)W^{-\frac{1}{2}}(x_{i},p_{0}(x_{i}),\gamma) \right| \right|_{op}.
\end{align*}
Since, for any metric, the distance function is $1$-Lipschitz and $W(x,p(x),\gamma)$ is symmetric and positive-definite,
\begin{align*}
    &\left|\text{dist}_{W(x_{i},p(x_{i}),\gamma)}(f(x_{i},p(x_{i}),\gamma),\mathbb{R}_{+}^{d_{f}}) - \text{dist}_{W(x_{i},p(x_{i}),\gamma)}(f(x_{i},p_{0}(x_{i}),\gamma),\mathbb{R}_{+}^{d_{f}})\right| \\
    &\quad \leq ||f(x_{i},p(x_{i}),\gamma)-f(x_{i},p_{0}(x_{i}),\gamma)||_{W(x_{i},p(x_{i}),\gamma)} \\
    &\quad \leq \sqrt{\lambda_{max}(W(x_{i},p(x_{i}),\gamma))} \cdot ||f(x_{i},p(x_{i}),\gamma)-f(x_{i},p_{0}(x_{i}),\gamma)||_{2} \\
    &\quad \leq \sqrt{\lambda_{max}(W(x_{i},p_{0}(x_{i}),\gamma))} \cdot C_{1}(x_{i},p(x_{i}),\gamma)\cdot ||f(x_{i},p(x_{i}),\gamma)-f(x_{i},p_{0}(x_{i}),\gamma)||_{2} \\
     &\quad \leq \sqrt{\overline{\lambda}_{W}} \cdot C_{1}(x_{i},p(x_{i}),\gamma) \cdot ||f(x_{i},p(x_{i}),\gamma)-f(x_{i},p_{0}(x_{i}),\gamma)||_{2},
\end{align*}
where the last equality uses that $\sup_{\gamma \in \Gamma}\sup_{n \geq 1}\max_{1 \leq i \leq n}\lambda_{max}^{1/2}(W(x_{i},p_{0}(x_{i}),\gamma)) \leq \overline{\lambda}_{W}^{1/2}$ by the assumption of the proposition. Consequently, for any $r \in [1,\infty]$, the following holds
\begin{align*}
    &\sup_{(\gamma,p) \in \Gamma \times V_{n,\delta_{n}}}\left|\left| \text{dist}_{W(\cdot,p(\cdot),\gamma)}(f(\cdot,p(\cdot),\gamma),\mathbb{R}_{+}^{d_{f}})-\text{dist}_{W(\cdot,p(\cdot),\gamma)}(f(\cdot,p_{0}(\cdot),\gamma),\mathbb{R}_{+}^{d_{f}})\right|\right|_{n,r} \\
    &\quad \leq \sqrt{\overline{\lambda}_{W}} \sup_{(\gamma,p) \in \Gamma \times V_{n,\delta_{n}}}\max_{1 \leq i \leq n}C_{1}(x_{i},p(x_{i}),\gamma) \sup_{(\gamma,p) \in \Gamma \times V_{n,\delta_{n}}}||f(\cdot,p(\cdot),\gamma)-f(\cdot,p_{0}(\cdot),\gamma)||_{n,r}
\end{align*}
Since $||\cdot||_{n,r} \leq ||\cdot||_{n,\infty}$ for all $r \in [1,\infty]$, we know, by assumption that,
\begin{align*}
\sup_{(\gamma,p) \in \Gamma \times V_{n,\delta_{n}}}||f(\cdot,p(\cdot),\gamma)-f(\cdot,p_{0}(\cdot),\gamma)||_{n,r} = o(1)
\end{align*}
for all $r \in [1,\infty]$. Moreover, the assumptions $\inf_{\gamma \in \Gamma}\min_{1 \leq i \leq n}\lambda_{min}(W(x_{i},p_{0}(x_{i}),\gamma)) \geq \underline{\lambda}_{W}$ and $\sup_{(\gamma,p) \in \Gamma \times V_{n,\delta_{n}}}||W(\cdot,p(\cdot),\gamma)-W(\cdot,p_{0}(\cdot),\gamma)||_{n,\infty} = o(1)$ implies that
\begin{align*}
    &\sup_{(\gamma,p) \in \Gamma \times V_{n,\delta_{n}}}\max_{1 \leq i \leq n}||W^{\frac{1}{2}}(x_{i},p(x_{i}),\gamma)W^{-\frac{1}{2}}(x_{i},p_{0}(x_{i}),\gamma)-I_{d_{f}}||_{op} \\
    &\quad \leq \sup_{\gamma \in \Gamma}\max_{1 \leq i \leq n}||W^{-\frac{1}{2}}(x_{i},p_{0}(x_{i}),\gamma)||_{op} \\
    &\quad\quad \times \sup_{(\gamma,p) \in \Gamma \times V_{n,\delta_{n}}}\max_{1 \leq i \leq n}||W^{\frac{1}{2}}(x_{i},p(x_{i}),\gamma)-W^{\frac{1}{2}}(x_{i},p(x_{i}),\gamma)||_{op} \\
    & \quad\leq \underline{\lambda}_{W}^{-\frac{1}{2}}\left(\sup_{(\gamma,p) \in \Gamma \times V_{n,\delta_{n}}}\max_{1 \leq i \leq n}||W(x_{i},p(x_{i}),\gamma)-W(x_{i},p_{0}(x_{i}),\gamma)||_{op}\right)^{\frac{1}{2}} \\
    &\quad = o(1),
\end{align*}
which implies
\begin{align*}
    \sup_{(\gamma,p) \in \Gamma \times V_{n,\delta_{n}}}\max_{1 \leq i \leq n} \left| C_{1}(x_{i},p(x_{i}),\gamma) - 1 \right| = o(1)
\end{align*}
by the reverse triangle inequality. Consequently,
\begin{align*}
&\sup_{(\gamma,p) \in \Gamma \times V_{n,\delta_{n}}}\left|\left| \text{dist}_{W(\cdot,p(\cdot),\gamma)}(f(\cdot,p(\cdot),\gamma),\mathbb{R}_{+}^{d_{f}})-\text{dist}_{W(\cdot,p(\cdot),\gamma)}(f(\cdot,p_{0}(\cdot),\gamma),\mathbb{R}_{+}^{d_{f}})\right|\right|_{n,r} \\
&\quad \leq  \sqrt{\overline{\lambda}_{W}}(1+o(1))o(1) \\
&\quad = o(1),  
\end{align*}
which establishes that (\ref{eq:bound1}) vanishes uniformly along $\{\Gamma \times V_{n,\delta_{n}}\}_{n \geq 1}$ as $n\rightarrow \infty$.
\paragraph{Step 3.}
We now show that (\ref{eq:bound2}) converges to zero as $n\rightarrow \infty$, uniformly over $\Gamma \times V_{n,\delta_{n}}$. Let
\begin{align*}
 C_{2}(x_{i},p(x_{i}),\gamma) = \left | \left| W^{\frac{1}{2}}(x_{i},p_{0}(x_{i}),\gamma)W^{-\frac{1}{2}}(x_{i},p(x_{i}),\gamma) \right| \right|_{op}.
\end{align*}
Using that $(W^{\frac{1}{2}}(x_{i},p(x_{i}),\gamma)W^{-\frac{1}{2}}(x_{i},p_{0}(x_{i}),\gamma))^{-1} = W^{\frac{1}{2}}(x_{i},p_{0}(x_{i}),\gamma)W^{-\frac{1}{2}}(x_{i},p(x_{i}),\gamma) $ and the submultiplicativity of operator norms, we know that
\begin{align*}
  C(x_{i},p(x_{i}),\gamma) :=   \max\{C_{1}(x_{i},p(x_{i}),\gamma)),C_{2}(x_{i},p(x_{i}),\gamma)\} \geq 1.
\end{align*}
Writing $||v||_{W(x_{i},p(x_{i}),\gamma)} = ||W^{1/2}(x_{i},p(x_{i}),\gamma) v||_{2}$ for any $v \in \mathbb{R}^{d_{f}}$, we obtain that
\begin{align*}
 \text{dist}_{W(x_{i},p(x_{i}),\gamma)}(f(x_{i},p_{0}(x_{i}),\gamma),\mathbb{R}_{+}^{d_{f}})& \leq C_{1}(x_{i},p(x_{i}),\gamma) \text{dist}_{W(x_{i},p_{0}(x_{i}),\gamma)}(f(x_{i},p_{0}(x_{i}),\gamma),\mathbb{R}_{+}^{d_{f}}) \\
 & \leq C(x_{i},p(x_{i}),\gamma) \text{dist}_{W(x_{i},p_{0}(x_{i}),\gamma)}(f(x_{i},p_{0}(x_{i}),\gamma),\mathbb{R}_{+}^{d_{f}}).
\end{align*}
Consequently,
\begin{align*}
&\text{dist}_{W(x_{i},p(x_{i}),\gamma)}(f(x_{i},p_{0}(x_{i}),\gamma),\mathbb{R}_{+}^{d_{f}})-\text{dist}_{W(x_{i},p_{0}(x_{i}),\gamma)}(f(x_{i},p_{0}(x_{i}),\gamma),\mathbb{R}_{+}^{d_{f}}) \\
&\quad \leq (C(x_{i},p(x_{i}),\gamma)-1)  \text{dist}_{W(x_{i},p_{0}(x_{i}),\gamma)}(f(x_{i},p_{0}(x_{i}),\gamma),\mathbb{R}_{+}^{d_{f}}).
\end{align*}
Similarly,
\begin{align*}
  \text{dist}_{W(x_{i},p_{0}(x_{i}),\gamma)}(f(x_{i},p_{0}(x_{i}),\gamma),\mathbb{R}_{+}^{d_{f}}) &\leq C_{2}(x_{i},p(x_{i}),\gamma) \text{dist}_{W(x_{i},p(x_{i}),\gamma)}(f(x_{i},p_{0}(x_{i}),\gamma),\mathbb{R}_{+}^{d_{f}}) \\
  &\leq C(x_{i},p(x_{i}),\gamma)\text{dist}_{W(x_{i},p(x_{i}),\gamma)}(f(x_{i},p_{0}(x_{i}),\gamma),\mathbb{R}_{+}^{d_{f}}),
\end{align*}
which implies
\begin{align*}
      &\text{dist}_{W(x_{i},p_{0}(x_{i}),\gamma)}(f(x_{i},p_{0}(x_{i}),\gamma),\mathbb{R}_{+}^{d_{f}}) -  \text{dist}_{W(x_{i},p(x_{i}),\gamma)}(f(x_{i},p_{0}(x_{i}),\gamma),\mathbb{R}_{+}^{d_{f}}) \\
      &\quad \leq (C(x_{i},p(x_{i}),\gamma)-1) \text{dist}_{W(x_{i},p(x_{i}),\gamma)}(f(x_{i},p_{0}(x_{i}),\gamma),\mathbb{R}_{+}^{d_{f}}) \\
      &\quad \leq (C(x_{i},p(x_{i}),\gamma)-1)C(x_{i},p(x_{i}),\gamma) \text{dist}_{W(x_{i},p_{0}(x_{i}),\gamma)}(f(x_{i},p_{0}(x_{i}),\gamma),\mathbb{R}_{+}^{d_{f}}).
\end{align*}
Using that $C(x_{i},p(x_{i}),\gamma) \geq 1$, it follows that
\begin{align*}
    &\left|\text{dist}_{W(x_{i},p(x_{i}),\gamma)}(f(x_{i},p_{0}(x_{i}),\gamma),\mathbb{R}_{+}^{d_{f}})-\text{dist}_{W(x_{i},p_{0}(x_{i}),\gamma)}(f(x_{i},p_{0}(x_{i}),\gamma),\mathbb{R}_{+}^{d_{f}})\right| \\
    &\quad \leq (C(x_{i},p(x_{i}),\gamma)-1)C(x_{i},p(x_{i}),\gamma) \text{dist}_{W(x_{i},p_{0}(x_{i}),\gamma)}(f(x_{i},p_{0}(x_{i}),\gamma),\mathbb{R}_{+}^{d_{f}}) \\
    &\quad \leq  \text{dist}_{W(\cdot,p_{0}(\cdot),\gamma)}(f(\cdot,p_{0}(\cdot),\gamma),\mathbb{R}_{+}^{d_{f}}) \\
    &\quad \quad \times \sup_{(\gamma,p) \in \Gamma \times V_{n,\delta_{n}}}\max_{1 \leq i \leq n}|C(x_{i},p(x_{i}),\gamma)-1|\sup_{(\gamma,p) \in \Gamma \times V_{n,\delta_{n}}}\max_{1 \leq i \leq n}C(x_{i},p(x_{i}),\gamma),  
\end{align*}
which implies that
\begin{align*}
&\sup_{(\gamma,p) \in \Gamma \times V_{n,\delta_{n}}}\left|\left| \text{dist}_{W(\cdot,p(\cdot),\gamma)}(f(\cdot,p_{0}(\cdot),\gamma),\mathbb{R}_{+}^{d_{f}})-\text{dist}_{W(\cdot,p_{0}(\cdot),\gamma)}(f(\cdot,p_{0}(\cdot),\gamma),\mathbb{R}_{+}^{d_{f}})\right|\right|_{n,r}   \\
&\quad \leq \sup_{\gamma \in \Gamma}\left | \left |\text{dist}_{W(\cdot,p_{0}(\cdot),\gamma)}(f(\cdot,p_{0}(\cdot),\gamma),\mathbb{R}_{+}^{d_{f}}) \right| \right|_{n,r} \cdot \sup_{(\gamma,p) \in \Gamma \times V_{n,\delta_{n}}}\max_{1 \leq i \leq n}|C(x_{i},p(x_{i}),\gamma)-1| \\
&\quad \quad \times \sup_{(\gamma,p) \in \Gamma \times V_{n,\delta_{n}}}\max_{1 \leq i \leq n}C(x_{i},p(x_{i}),\gamma)
\end{align*}
Provided that
\begin{align}\label{eq:inverse}
    \sup_{(\gamma,p) \in \Gamma \times V_{n,\delta_{n}}}\max_{1 \leq i \leq n}|C_{2}(x_{i},p(x_{i}),\gamma)-1 | = o(1),
\end{align}
we know that
\begin{align*}
    \sup_{(\gamma,p) \in \Gamma \times V_{n,\delta_{n}}}\max_{1 \leq i \leq n}|C(x_{i},p(x_{i}),\gamma)-1| = o(1)
\end{align*}
and
\begin{align*}
\sup_{(\gamma,p) \in \Gamma \times V_{n,\delta_{n}}}\max_{1 \leq i \leq n}C(x_{i},p(x_{i}),\gamma) = 1 + o(1),    
\end{align*}
which implies
\begin{align*}
    \sup_{(\gamma,p) \in \Gamma \times V_{n,\delta_{n}}}\left|\left| \text{dist}_{W(\cdot,p(\cdot),\gamma)}(f(\cdot,p_{0}(\cdot),\gamma),\mathbb{R}_{+}^{d_{f}})-\text{dist}_{W(\cdot,p_{0}(\cdot),\gamma)}(f(\cdot,p_{0}(\cdot),\gamma),\mathbb{R}_{+}^{d_{f}})\right|\right|_{n,r} = o(1)
\end{align*}
as $n\rightarrow \infty$. So, to complete the proof, we show (\ref{eq:inverse}). Observe that, by the triangle inequality,
\begin{align*}
    &\sup_{(\gamma,p) \in \Gamma \times V_{n,\delta_{n}}}\max_{1 \leq i \leq n}|C(x_{i},p(x_{i}),\gamma)-1|  \\
    &\quad \leq \sup_{(\gamma,p) \in \Gamma \times V_{n,\delta_{n}}}\max_{1 \leq i \leq n}||W^{-\frac{1}{2}}(x_{i},p(x_{i}),\gamma)||_{op} \\
    &\quad \quad \times \sup_{(\gamma,p) \in \Gamma \times V_{n,\delta_{n}}}\max_{1 \leq i \leq n}||W^{\frac{1}{2}}(x_{i},p(x_{i}),\gamma)-W^{\frac{1}{2}}(x_{i},p_{0}(x_{i}),\gamma)||_{op}
\end{align*}
In Step 2, we showed that
\begin{align*}
    \sup_{(\gamma,p) \in \Gamma \times V_{n,\delta_{n}}}\max_{1 \leq i \leq n}||W^{\frac{1}{2}}(x_{i},p(x_{i}),\gamma)-W^{\frac{1}{2}}(x_{i},p_{0}(x_{i}),\gamma)||_{op} = o(1),
\end{align*}
and, for $n$ large, this implies that
\begin{align*}
 \sup_{(\gamma,p) \in \Gamma \times V_{n,\delta_{n}}}\max_{1 \leq i \leq n}||W^{-\frac{1}{2}}(x_{i},p(x_{i}),\gamma)||_{op} \leq 2    \sup_{\gamma \in \Gamma}\max_{1 \leq i \leq n}||W^{-\frac{1}{2}}(x_{i},p_{0}(x_{i}),\gamma)||_{op} \leq \frac{2}{\underline{\lambda}_{W}},
\end{align*}
which together yields (\ref{eq:inverse}).
\end{proof}
\begin{proof}[Proof of Proposition \ref{prop:taxicab}]
Since $||(f(x,p(x),\gamma))_{-}||_{1}= \inf_{\tau \in \mathbb{R}_{+}^{d_{f}}}||f(x,p(x),\gamma)-\tau||_{1}$ for all $x$ and all $p$, the function $f(x,p(x),\gamma) \mapsto ||(f(x,p(x),\gamma))_{-}||_{1}$ is $1$-Lipschitz for all $x$ and all $p$ using the Lipschitz continuity of the distance of a point to a set. Combining this with the reverse triangle inequality,
\begin{align*}
    |Q_{n}^{LP}(\gamma,p)-Q_{n}^{LP}(\gamma,p_{0})| \leq ||f(\cdot,p(\cdot),\gamma)-f(\cdot,p_{0}(\cdot),\gamma)||_{n,1},
\end{align*}
and, taking the supremum over $\Gamma \times V_{n,\delta_{n}}$,
\begin{align*}
\sup_{(\gamma,p) \in \Gamma \times V_{n,\delta_{n}}}|Q_{n}^{LP}(\gamma,p)-Q_{n}^{LP}(\gamma,p_{0})| \leq \sup_{(\gamma,p) \in \Gamma \times V_{n,\delta_{n}}}||f(\cdot,p(\cdot),\gamma)-f(\cdot,p_{0}(\cdot),\gamma)||_{n,1}.
\end{align*}
Using that $||\cdot||_{n,1} \leq ||\cdot||_{n,1}$ and invoking the condition of the proposition, we conclude that, for $P_{0,X}^{(\infty)}$-almost every fixed realization $\{x_{i}\}_{i \geq 1}$ of $\{X_{i}\}_{i \geq 1}$,
\begin{align*}
    \sup_{(\gamma,p) \in \Gamma \times V_{n,\delta_{n}}}|Q_{n}^{LP}(\gamma,p)-Q_{n}^{LP}(\gamma,p_{0})| \longrightarrow 0
\end{align*}
as $n\rightarrow \infty$. This completes the proof.
\end{proof}
\subsubsection{Proofs for Section \ref{sec:suffhemi}}
\begin{proof}[Proof of Proposition \ref{prop:hemi}]
It suffices to show that for $P^{(\infty)}_{0,X}$-almost every fixed realization $\{x_{i}\}_{i \geq 1}$ of $\{X_{i}\}_{i \geq 1}$, the following is true: for every $\varepsilon > 0$, there exists $N \geq 1$ such that $p \in V_{n,\delta_{n}}$ and $n \geq N$ implies $\Gamma_{n,I}(p_{0})\subseteq \Gamma_{n,I}^{\varepsilon}$. To that end, let $\varepsilon > 0$ be fixed arbitrarily. If $\gamma \in \Gamma_{n,I}(p_{0})$, then $f(x_{i},p_{0}(x_{i}),\gamma)  \geq 0$ for all $i=1,...n$, which, by uniform convergence, implies that there exists $\bar{N}(\varepsilon) \geq 1$ such that $n \geq \bar{N}(\varepsilon)$ and $p \in V_{n,\delta_{n}}$ implies
\begin{align*}
   f(x_{i},p(x_{i}),\gamma) > - \frac{c}{\sqrt{d_{f}}}\varepsilon \cdot \iota_{d_{f}} \ \forall \ i=1,...,n,
\end{align*}
where $\iota_{d_{f}}$ is a $d_{f} \times 1$ vector of ones. Consequently, there exists $\bar{N}(\varepsilon) \geq 1$ such that $n \geq \bar{N}(\varepsilon)$ and $p \in V_{n,\delta_{n}}$ implies
\begin{align*}
\max_{1 \leq i \leq n}||(f(x_{i},p(x_{i}),\gamma))_{-}||_{2} < c\varepsilon,
\end{align*}
By the assumption of the proposition, if $n \geq \max\{\tilde{N},\bar{N}(\varepsilon)\}$ and $p \in V_{n,\delta_{n}}$, then
\begin{align*}
    \text{dist}(\gamma,\Gamma_{n,I}(p)) < \varepsilon
\end{align*}
Hence, $\gamma \in \Gamma_{n,I}^{\varepsilon}(p)$. Thus, if we set $N = \max\{\bar{N}(\varepsilon),\tilde{N}\}$, then we can repeat this argument over the elements of $\Gamma_{n,I}(p_{0})$ to conclude that $p \in V_{n,\delta_{n}}$ and $n \geq N$ implies $\Gamma_{n,I}(p_{0}) \subseteq \Gamma_{n,I}^{\varepsilon}(p)$. Moreover, since $\varepsilon > 0$ was fixed arbitrarily, we can conclude that this holds for every $\varepsilon$ and we complete the proof.
\end{proof}
\begin{proof}[Proof of Corollary \ref{cor:unions}]
We show that for $P^{(\infty)}_{0,X}$-almost every fixed realization $\{x_{i}\}_{i \geq 1}$ of $\{X_{i}\}_{i \geq 1}$, the following is true: for every $\varepsilon > 0$, there exists $N \geq 1$ such that $p \in V_{n,\delta_{n}}$ and $n \geq N$ implies $\Gamma_{n,I}(p_{0})\subseteq \Gamma_{n,I}^{\varepsilon}$. To this end, let $\{c_{j}\}_{j=1}^{J}$ and $\{\bar{N}_{j}\}_{j=1}^{J}$ be such that (\ref{eq:errorbound}) of Proposition \ref{prop:hemi} hold for each $j \in \{1,...,J\}$. Similarly, for every $\varepsilon >0$, let $\{\tilde{N}_{j}\}_{j=1}^{J}$ be such that
\begin{align*}
    &f^{(j)}_{l}(x_{i},p_{0}(x_{i}),\gamma) \geq 0 \ \forall \ l=1,...,d_{f^{(j)}}, \ i=1,...,n \\
    &\implies ||(f^{(j)}(x_{i},p(x_{i}),\gamma))_{-}||_{2} < c_{j} \varepsilon \ \forall  \ i=1,...,n
\end{align*}
for all $n \geq \tilde{N}_{j}$,  $\gamma \in \Gamma$, and $p \in V_{n,\delta_{n}}$. Such a sequence exists as we assume that the functions that define $\Gamma_{n,I}^{(j)}(p)$ satisfy the uniform convergence in Proposition \ref{prop:hemi}. Given these two properties, if $\gamma \in \Gamma_{n,I}(p_{0})$ and $n \geq N:=\max_{1 \leq j \leq J}\max\{\tilde{N}_{j},\bar{N}_{j}\}$ and $p \in V_{n,\delta_{n}}$, then there exists $j_{n} \in [J]$ such that $\text{dist}(\gamma,\Gamma_{n,I}^{(j_{n})}(p)) < \varepsilon$ (by following the argument of Proposition \ref{prop:hemi}). Since $\text{dist}(\gamma,\Gamma_{n,I}(p)) = \min_{1 \leq j \leq J}\text{dist}(\gamma,\Gamma_{n,I}^{(j)}(p))$, we conclude that $\text{dist}(\gamma,\Gamma_{n,I}(p)) < \varepsilon$, which is equivalent to $\gamma \in \Gamma_{n,I}^{\varepsilon}(p)$. Repeating this argument over all elements of $\Gamma_{n,I}(p_{0})$, it follows that $n \geq N$ and $p \in V_{n,\delta_{n}}$ implies $\Gamma_{n,I}(p_{0})\subseteq \Gamma_{n,I}^{\varepsilon}(p)$.
\end{proof}
\begin{lemma}\label{lem:interiorparspace}
Suppose that Conditions 3 and 5 of Proposition \ref{prop:convexhemi} holds. Then the following is true for $P_{0,X}^{(\infty)}$-almost every fixed realization $\{x_{i}\}_{i \geq 1}$ of $\{X_{i}\}_{i \geq 1}$: there exists $\hat{N} \geq 1$ such that $\Gamma_{n,I}(p) \subseteq \text{int}(\Gamma)$ for all $p \in V_{n,\delta_{n}}$ and all $n \geq \hat{N}$.
\end{lemma}
\begin{proof}
Suppose, by contradiction, there exists a measurable set $A$ with $P_{0,X}^{(\infty)}(A) > 0$ and, for any $\{x_{i}\}_{i \geq 1} \in A$, the result is false. For any fixed $\{x_{i}\}_{i \geq 1} \in A$, there exist a subsequence $\{\omega_n\}_{n\geq 1}$ of $\{n\}_{n \geq 1}$
and a sequence $\{p_{\omega_n}\}_{n\geq 1}$ such that $p_{\omega_n} \in V_{\omega_n,\delta_{\omega_n}}$ and $\Gamma_{\omega_n,I}(p_{\omega_n}) \cap \partial\Gamma \neq \emptyset$ for each
$n \geq 1$. For each $n \geq 1$, choose
$\gamma_{\omega_n} \in \Gamma_{\omega_n,I}(p_{\omega_n}) \cap \partial\Gamma$.
Since $\gamma_{\omega_n} \in \Gamma_{\omega_n,I}(p_{\omega_n})$, we have
\begin{equation*}
0 \leq f_j(x_i, p_{\omega_n}(x_i), \gamma_{\omega_n})
\leq f_j(x_i, p_0(x_i), \gamma_{\omega_n})
+ \sup_{(\gamma,p) \in \Gamma \times V_{\omega_n,\delta_{\omega_n}}}
\left\| f(\cdot, p(\cdot), \gamma) - f(\cdot, p_0(\cdot), \gamma) \right\|_{\omega_n,\infty}
\end{equation*}
for all $i = 1, \dots, \omega_n$ and every $j = 1, \dots, d_f$. Since $\gamma_{\omega_n} \in \partial\Gamma$, the third condition of Proposition \ref{prop:convexhemi}
implies that, for all $n \geq 1$ with $\omega_n \geq \bar{N}_1$,
\begin{equation*}
f_j(x_i, p_0(x_i), \gamma_{\omega_n})
\leq \sup_{\gamma \in \partial\Gamma} \max_{1 \leq i \leq \omega_n}
\max_{1 \leq j \leq d_f} f_j(x_i, p_0(x_i), \gamma) \leq -\xi.
\end{equation*}
By the fifth condition of Proposition \ref{prop:convexhemi}, there exists $N' \geq 1$ such that
\begin{equation*}
\sup_{(\gamma,p) \in \Gamma \times V_{\omega_n,\delta_{\omega_n}}}
\left\| f(\cdot, p(\cdot), \gamma) - f(\cdot, p_0(\cdot), \gamma) \right\|_{\omega_n,\infty}
\leq \frac{\xi}{2}
\end{equation*}
for all $n \geq 1$ with $\omega_n \geq N'$. Combining the three preceding displays,
for all $n \geq 1$ with $\omega_n \geq \max\{\bar{N}_1, N'\},$
\begin{equation*}
0 \leq -\xi + \frac{\xi}{2} = -\frac{\xi}{2} < 0,
\end{equation*}
which is a contradiction. This implies that no such $A$ can exist, and, as a result, for $P_{0,X}^{(\infty)}$-almost every fixed realization $\{x_{i}\}_{i \geq 1}$ of $\{X_{i}\}_{i \geq 1}$, there exists $\hat{N} \geq 1$ such that
$\Gamma_{n,I}(p) \subseteq \mathrm{int}(\Gamma)$ for all $p \in V_{n,\delta_n}$
and all $n \geq \hat{N}$.    
\end{proof}
\begin{lemma}\label{lem:slater}
Suppose that Conditions 3, 4, and 5 of Proposition \ref{prop:convexhemi} hold. Then, for $P_{0,X}^{(\infty)}$-almost every fixed realization $\{x_{i}\}_{i \geq 1}$ of $\{X_{i}\}_{i \geq 1}$, there exists $\check{N} \geq 1$ such that $\gamma^{\circ}$ is a Slater point for the class of constraint systems $\{\Gamma_{n,I}(p): p \in V_{n,\delta_{n}}\}$.    
\end{lemma}
\begin{proof}
We want to show that, conditional on $P_{0,X}^{(\infty)}$-almost every fixed realization $\{x_{i}\}_{i \geq 1}$ on $\{X_{i}\}_{i \geq 1}$, there exists $\check{N} \geq 1$ such that 
\begin{align*}
\gamma^{\circ} \in \text{int}(\Gamma)    
\end{align*} 
and 
\begin{align*}
\inf_{p \in V_{n,\delta_{n}}}\min_{1 \leq i \leq n}\min_{1 \leq j \leq d_{f}}f_{j}(x_{i},p(x_{i}),\gamma^{\circ}) > 0    
\end{align*} for all $n \geq \check{N}$. By the fourth condition of Proposition \ref{prop:convexhemi},
\[
\min_{1\leq i\leq n}\min_{1\leq j\leq d_f}
f_j\bigl(x_i,p_0(x_i),\gamma^\circ\bigr)
\geq \eta
\]
for all \(n\geq \bar N_2\), which implies $\gamma^{\circ} \in \Gamma_{n,I}(p_{0})$ for all $n \geq N_{2}$. Combining this with the first part of the third
condition gives $\gamma^{\circ} \in \operatorname{int}(\Gamma)$ for every $ n \geq \max\{\bar{N}_{1},\bar{N}_{2}\}$. Next, combining the fourth and fifth conditions, there exists $\check{N} \geq \max\{\bar{N}_{1},\bar{N}_{2}\}$ such that
\begin{align*}
\min_{1\leq i\leq n}\min_{1\leq j\leq d_f}
f_j\bigl(x_i,p(x_i),\gamma^\circ\bigr)
\geq \frac{\eta}{2}
\end{align*}
for all \(p\in V_{n,\delta_n}\) and all \(n\geq \check N\).
\end{proof}
\begin{proof}[Proof of Proposition \ref{prop:convexhemi}]
The proof requires checking that the second condition of Proposition \ref{prop:hemi} holds because, in that case, one can apply Proposition \ref{prop:hemi} to conclude that Assumption \ref{as:prior}.3 holds. This requires two steps. The first step derives a sufficient condition for the result and the second step verifies the sufficient condition.
\paragraph{Step 1.} Throughout the proof, set $n \geq \min\{\hat{N},\check{N}\}$, where $\hat{N}$ and $\check{N}$ are from Lemmas \ref{lem:interiorparspace} and \ref{lem:slater}, respectively, so that we can use $\Gamma_{n,I}(p) \subseteq \text{int}(\Gamma)$ for all $p \in V_{n,\delta_{n}}$, $\gamma^{\circ}$ is a Slater point for the constraint system $\Gamma_{n,I}(p)$ for all $p \in V_{n,\delta_{n}}$, and $\Gamma_{n,I}(p) \neq \emptyset$ for all $p \in V_{n,\delta_{n}}$ (an implication of Lemma \ref{lem:slater}). Further, set $\gamma \notin \Gamma_{n,I}(p)$ (otherwise, the result holds trivially if $\gamma \in \Gamma_{n,I}(p)$). We first show that there is a constant $C_{n} \geq 0$ such that
\begin{align*}
    \max_{1 \leq i \leq n}||(f(x_{i},p(x_{i}),\gamma))_{-}||_{2} \geq \text{dist}(\gamma,\Gamma_{n,I}(p))C_{n}.
\end{align*}
Let $\gamma_{n,p}$ solve $\min_{\tilde{\gamma} \in \Gamma_{n,I}(p)}||\gamma-\tilde{\gamma}||_{2}$. This distance-minimizing vector $\gamma_{n,p}$ exists by compactness of $\Gamma_{n,I}(p)$ and continuity of $||\cdot||_{2}$ (hence, the Weierstrass Extreme Value Theorem applies). Next, let $\mathcal{F}_{n}^{*}(\gamma,p) = \left\{(i,j) \in [n] \times [d_{f}]: f_{j}(x_{i},p(x_{i}),\gamma_{n,p}) = 0\right\}$ be the set of active constraints at $\gamma_{n,p}$ and note that $\mathcal{F}_{n}^{*}(\gamma,p) \neq \emptyset$ because $\gamma_{n,p} \in \partial \Gamma_{n,I}(p)$.\footnote{Note that $[n]:=\{1,...,n\}$ and $[d_{f}] = \{1,...,d_{f}\}$.} For any $(i,j) \in \mathcal{F}_{n}^{*}(\gamma,p)$, the concavity and differentiability of $f_{j}(x_{i},p(x_{i}),\gamma)$ in $\gamma$ (i.e., the second assumption of the proposition) reveals that
\begin{align*}
    f_{j}(x_{i},p(x_{i}),\gamma) &\leq f_{j}(x_{i},p(x_{i}),\gamma_{n,p}) + \nabla_{\gamma}f_{j}(x_{i},p(x_{i}),\gamma_{n,p})'(\gamma-\gamma_{n,p}) \\
    &= \nabla_{\gamma} f_{j}(x_{i},p(x_{i}),\gamma_{n,p})'(\gamma-\gamma_{n,p}),
\end{align*}
where the last equality uses that $f_{j}(x_{i},p(x_{i}),\gamma_{n,p}) = 0$ for any $i,j \in \mathcal{F}_{n}^{*}(\gamma,p)$. Consequently, we can use that $z \mapsto |(z)_{-}|$ is nonincreasing and that $||\gamma-\gamma_{n,p}||_{2} >0$ (because $\gamma \notin \Gamma_{n,I}(p)$) to conclude
\begin{align*}
\max_{1 \leq i \leq n}||(f(x_{i},p(x_{i}),\gamma))_{-}||_{2} &\geq \max_{(i,j) \in \mathcal{F}_{n}^{*}(\gamma,p)}\left|\min\{f_{j}(x_{i},p(x_{i}),\gamma),0\}\right| \\
&\geq ||\gamma-\gamma_{n,p}||_{2} \cdot \max_{(i,j) \in \mathcal{F}_{n}^{*}(\gamma,p)}\left|(\nabla_{\gamma}f_{j}(x_{i},p(x_{i}),\gamma_{n,p})'t_{n,p})_{-}\right| \\
&= \text{dist}(\gamma,\Gamma_{n,I}(p)) \cdot \max_{(i,j) \in \mathcal{F}_{n}^{*}(\gamma,p)}\left|(\nabla_{\gamma}f_{j}(x_{i},p(x_{i}),\gamma_{n,p})'t_{n,p})_{-}\right|
\end{align*}
where $t_{n,p} = (\gamma-\gamma_{n,p}) /||\gamma-\gamma_{n,p}||_{2}$. Let 
\begin{align*}
    C_{n}=\inf_{(\gamma,p)\in \Gamma \times V_{n,\delta_{n}}}\max_{(i,j) \in \mathcal{F}_{n}^{*}(\gamma,p)}\left|(\nabla_{\gamma}f_{j}(x_{i},p(x_{i}),\gamma_{n,p})'t_{n,p})_{-}\right| \geq 0.
\end{align*}
If we can show that
\begin{align}\label{eq:condition}
C:=\liminf_{n\rightarrow\infty}C_{n} >0,   
\end{align}
then we complete the proof because we set $c = C/2$ and $\tilde{N}$ to be the minimal $n$ for which $\inf_{(\gamma,p)\in \Gamma \times V_{n,\delta_{n}}}\max_{(i,j) \in \mathcal{F}_{n}^{*}(\gamma,p)}\left|(\nabla_{\gamma}f_{j}(x_{i},p(x_{i}),\gamma_{n,p})'t_{n,p})_{-}\right| > C/2$ (such $n$ exists if $C>0$).
\paragraph{Step 2.} We break the verification of (\ref{eq:condition}) into two substeps.
\paragraph{Step 2A.} Consider the constrained optimization problem $\min_{\tilde{\gamma} \in \Gamma_{n,I}(p)}||\gamma-\tilde{\gamma}||_{2}$. We first note that this is a convex program because the objective is convex, the mapping $\tilde{\gamma} \mapsto f_{j}(x_{i},p(x_{i}),\tilde{\gamma})$ is concave for each $i,j$ (by the second condition of the proposition), and $\Gamma$ is convex and compact by the first condition of the proposition. Next, we invoke Lemma \ref{lem:slater} to conclude that the Karusch-Kuhn-Tucker conditions are necessary at $\gamma_{n,p}$ (see, for example, Theorem 28.3 of \cite{rockafeller1970convex}). Because $\gamma_{n,p} \in \Gamma_{n,I}(p) \subseteq \operatorname{int}(\Gamma)$ (recall this holds because $n \geq \max\{\hat{N},\check{N}\}$ so that Lemma \ref{lem:interiorparspace} holds), 
the constraint \(\gamma\in \Gamma\) is inactive at
\(\gamma_{n,p}\), and stationarity gives
\begin{align*}
t_{n,p}'
=
-\sum_{(i,j)\in \mathcal{F}_n^*(\gamma,p)}
\lambda_{n,ij}(\gamma,p)\nabla_{\gamma}f_{j}(x_{i},p(x_{i}),\gamma_{n,p})',
\end{align*}
where $\lambda_{n,ij}(\gamma,p)\geq 0$ is the multiplier associated with the constraint $f_j\bigl(x_i,p(x_i),\widetilde\gamma\bigr)\geq 0$, and, the restriction to $\mathcal{F}_{n}^{*}(\gamma,p)$ in the sum reflects that, by complementary slackness, $\lambda_{n,ij}(\gamma,p)=0$ for $(i,j) \notin \mathcal{F}_{n}^{*}(\gamma,p)$. Post-multiplying by $t_{n,p}$ and using that $||t_{n,p}||_{2} = 1$, we know that
\begin{align*}
    1 &=-\sum_{(i,j) \in \mathcal{F}_{n}^{*}(\gamma,p)}\lambda_{n,ij}(\gamma,p)\nabla_{\gamma}f_{j}(x_{i},p(x_{i}),\gamma_{n,p})'t_{n,p} \\
    &\leq \max_{(i,j) \in \mathcal{F}_{n}^{*}(\gamma,p)}\{-\nabla_{\gamma}f_{j}(x_{i},p(x_{i}),\gamma_{n,p})'t_{n,p}\}||\lambda_{n}(\gamma,p)||_{1} \\
    &= \max_{(i,j) \in \mathcal{F}_{n}^{*}(\gamma,p)}|(\nabla_{\gamma}f_{j}(x_{i},p(x_{i}),\gamma_{n,p})'t_{n,p})_{-}|\cdot ||\lambda_{n}(\gamma,p)||_{1} \\
    &\leq \max_{(i,j) \in \mathcal{F}_{n}^{*}(\gamma,p)}|(\nabla_{\gamma}f_{j}(x_{i},p(x_{i}),\gamma_{n,p})'t_{n,p})_{-}| \sup_{(\gamma,p) \in \Gamma\times V_{n,\delta_{n}}}||\lambda_{n}(\gamma,p)||_{1},
\end{align*}
where the first inequality uses the definition of the taxicab norm $||\cdot||_{1}$ (i.e. $||z||_{1} = \sum_{j=1}^{d}|z_{j}|$ for $z \in \mathbb{R}^{d}$) and that $\lambda_{n}(\gamma,p) \in \mathbb{R}^{nd_{f}}_{+}$, the second equality uses that the first equality and the property $\lambda_{n}(\gamma,p) \in \mathbb{R}_{+}^{nd_{f}}$ implies there exists $(i,j) \in \mathcal{F}_{n}^{*}(\gamma,p)$ such that $\nabla_{\gamma}f_{j}(x_{i},p(x_{i}),\gamma_{n,p})'t_{n,p} < 0$, and, as a result, $\max_{(i,j) \in \mathcal{F}_{n}^{*}(\gamma,p)}\{-\nabla_{\gamma}f_{j}(x_{i},p(x_{i}),\gamma_{n,p})'t_{n,p}\}=\max_{(i,j) \in \mathcal{F}_{n}^{*}(\gamma,p)}|(\nabla_{\gamma}f_{j}(x_{i},p(x_{i}),\gamma_{n,p})'t_{n,p})_{-}|$, and the last inequality uses the definition of the supremum. Consequently, if there is a constant $\tilde{C} > 0$ and $\tilde{N} \geq 1$ such that
\begin{align}\label{eq:multipliers}
    \sup_{(\gamma,p) \in \Gamma\times V_{n,\delta_{n}}}||\lambda_{n}(\gamma,p)||_{1} \leq \tilde{C}
\end{align}
for all $n \geq \tilde{N}$, then
\begin{align*}
    \max_{(i,j) \in \mathcal{F}_{n}^{*}(\gamma,p)}|(\nabla_{\gamma}f_{j}(x_{i},p(x_{i}),\gamma_{n,p})'t_{n,p})_{-}| \geq \frac{1}{\tilde{C}} > 0
\end{align*}
for all $n \geq \tilde{N}$, which implies $ C > 0$ and we complete the proof.
\paragraph{Step 2B.} To show \eqref{eq:multipliers}, define
\begin{align*}
\mathcal{L}_n(\widetilde\gamma,p)
=
||\gamma-\tilde{\gamma}||_{2}
-
\sum_{(i,j) \in \mathcal{F}_n^*(\gamma,p)}
\lambda_{n,ij}(\gamma,p)
f_j\left(x_i,p(x_i),\widetilde{\gamma}\right).
\end{align*}
The function $\mathcal{L}_{n}(\gamma,p)$ is a differentiable, convex function over $\Gamma$, because it is the sum of a
convex function and nonnegative multiples of the convex functions $\tilde{\gamma} \mapsto -f_j\left(x_i,p(x_i),\widetilde\gamma\right)$ (by the second condition of the proposition), and all of the terms involving $\tilde{\gamma}$ are differentiable. Its gradient at $\gamma_{n,p}$ is
\begin{align*}
-t_{n,p}'
-
\sum_{(i,j)\in F_n^*(\gamma,p)}
\lambda_{n,ij}(\gamma,p)
\nabla_{\gamma}f_{j}(x_{i},p(x_{i}),\gamma_{n,p})'
\end{align*}
and satisfies
\begin{align*}
-t_{n,p}'
-
\sum_{(i,j)\in F_n^*(\gamma,p)}
\lambda_{n,ij}(\gamma,p)
\nabla_{\gamma}f_{j}(x_{i},p(x_{i}),\gamma_{n,p})'=0
\end{align*}
by the stationary condition above. Since $\gamma_{n,p}\in \operatorname{int}(\Gamma)$, it follows that $\gamma_{n,p}$ minimizes $\mathcal{L}_n(\cdot,p)$ over $\Gamma$, and, since $f_{j}(x_{i},p(x_{i}),\gamma_{n,p})=0$ for all $(i,j) \in \mathcal{F}_{n}^{*}(\gamma,p)$, the value of the objective is $\mathcal{L}_{n}(\gamma_{n,p},p) = ||\gamma-\gamma_{n,p}||_{2}$. Consequently, since \(\gamma^\circ\in \Gamma\),
\[
\begin{aligned}
||\gamma-\gamma_{n,p}||_{2}
&=
\inf_{\widetilde\gamma\in \Gamma}\mathcal{L}_n(\widetilde\gamma,p) \\
&\leq
\|\gamma-\gamma^\circ\|_2
-
\sum_{(i,j)\in F_n^*(\gamma,p)}
\lambda_{n,ij}(\gamma,p)
f_j\bigl(x_i,p(x_i),\gamma^\circ\bigr).
\end{aligned}
\]
where the second equality uses the definition of the infimum. Using the nonnegativity of the norm, it follows that
\begin{align}\label{eq:dualitybound}
  \sum_{(i,j) \in \mathcal{F}_{n}^{*}(\gamma,p)}\lambda_{n,ij}(\gamma,p)f_{j}(x_{i},p(x_{i}),\gamma^{\circ})  \leq  ||\gamma-\gamma^{\circ}||_{2} - ||\gamma-\gamma_{n,p}||_{2}.
\end{align}
Since $\lambda_{n}(\gamma,p) \in \mathbb{R}_{+}^{nd_{f}}$ and Lemma \ref{lem:slater} holds for $n \geq \max\{\hat{N},\check{N}\}$,
\begin{align*}
    \sum_{(i,j) \in \mathcal{F}_{n}^{*}(\gamma,p)}\lambda_{n,ij}(\gamma,p)f_{j}(x_{i},p(x_{i}),\gamma^{\circ}) &\geq ||\lambda_{n}(\gamma,p)||_{1} \min_{1 \leq i \leq n}\min_{1 \leq j \leq d_{f}}f_{j}(x_{i},p(x_{i}),\gamma^{\circ}) \\
    &\geq \frac{\eta}{2}||\lambda_{n}(\gamma,p)||_{1} 
\end{align*}
for all $(\gamma,p) \in \Gamma \times V_{n,\delta_{n}}$ and $n \geq \check{N}$, (\ref{eq:dualitybound}) implies
\begin{align*}
 ||\lambda_{n}(\gamma,p)||_{1} \leq \frac{2}{\eta}\left(||\gamma-\gamma^{\circ}||_{2} - ||\gamma-\gamma_{n,p}||_{2}\right)
\end{align*}
By the triangle inequality and $||\gamma-\gamma_{n,p}||_{2} \leq ||\gamma-\gamma^{\circ}||_{2}$, we know
\begin{align*}
||\gamma-\gamma^{\circ}||_{2} - ||\gamma-\gamma_{n,p}||_{2}= \left|||\gamma-\gamma^{\circ}||_{2} - ||\gamma-\gamma_{n,p}||_{2}\right|   \leq ||\gamma^{\circ}-\gamma_{n,p}||_{2} \leq \text{diam}(\Gamma)
\end{align*}
which implies that
\begin{align*}
     ||\lambda_{n}(\gamma,p)||_{1}  \leq \frac{2\text{diam}(\Gamma)}{\eta}
\end{align*}
for all $(\gamma,p) \in \Gamma \times V_{n,\delta_{n}}$ and $n \geq \check{N}$. Hence, we conclude that, by setting $\tilde{C} = 2 \text{diam}(\Gamma)/\eta$ and $\tilde{N} \geq \max\{\hat{N},\check{N}\}$, that (\ref{eq:multipliers}) holds. This completes the proof.
\end{proof}
\begin{proof}[Proof of Proposition \ref{prop:minorantmisspec}]
 It suffices to show that for every $\varepsilon >0$, there exists $N \geq 1$ such that $p \in V_{n,\delta_{n}}$ and $n \geq N$ implies $\tilde{\Gamma}_{n,I}(p_{0}) \subseteq \tilde{\Gamma}_{n,I}^{\varepsilon}(p)$. To that end, let $\gamma \in \tilde{\Gamma}_{n,I}(p_{0}) $ be fixed arbitrarily. Since $\tilde{Q}_{n}(\gamma,p_{0}) = 0$ for all $\gamma \in \tilde{\Gamma}_{n,I}(p_{0})$,
 \begin{align*}
     \tilde{Q}_{n}(\gamma,p) \leq\sup_{(\gamma,p) \in \Gamma \times V_{n,\delta_{n}}}\left|\tilde{Q}_{n}(\gamma,p) - \tilde{Q}_{n}(\gamma,p_{0})\right|=:u_{n}
 \end{align*}
 for all $p \in V_{n,\delta_{n}}$. Then, applying (\ref{eq:hemi_misspec_minorant}),
 \begin{align*}
     \text{dist}(\gamma,\tilde{\Gamma}_{n,I}(p)) \leq \left(\frac{u_{n}+r_{n}}{c}\right)^{\frac{1}{\zeta}}=:\rho_{n}(c,\zeta)
 \end{align*}
 for all $p \in V_{n,\delta_{n}}$ and $n \geq \tilde{N}$. Since $\rho_{n}(c,\zeta) \rightarrow 0$ as $n\rightarrow \infty$, there exists $\check{N} \geq 1$ such that $\rho_{n}(c,\zeta) < \varepsilon$ for all $n \geq \check{N}$. Consequently,
 \begin{align*}
  \text{dist}(\gamma,\tilde{\Gamma}_{n,I}(p)) < \varepsilon
 \end{align*}
 for all $n \geq N:=\max\{\tilde{N},\check{N}\}$, which is equivalent to $\gamma \in \tilde{\Gamma}_{n,I}^{\varepsilon}(p)$. Since $\gamma \in \tilde{\Gamma}_{n,I}(p_{0})$ was chosen arbitrarily, it follows that $p \in V_{n,\delta_{n}}$ and $n \geq N$ implies $\tilde{\Gamma}_{n,I}(p_{0}) \subseteq \tilde{\Gamma}_{n,I}^{\varepsilon}(p)$.
\end{proof}
Before proving Proposition \ref{prop:misspec_hemi_convex}, we establish two technical lemmas.
\begin{lemma}\label{lem:singletonidset}
Suppose that Conditions 1, 2, and 5 of Proposition \ref{prop:misspec_hemi_convex} hold. Then, for $P_{0,X}^{(\infty)}$-almost every fixed realization $\{x_{i}\}_{i \geq 1}$ of $\{X_{i}\}_{i \geq 1}$, $\tilde{\Gamma}_{n,I}(p_{0})$ is a singleton for all $n \geq \bar{N}$.
\end{lemma}
\begin{proof}
Let $n\geq\bar{N}$, $\tilde{\gamma}_{n,p_{0}}\in\tilde{\Gamma}_{n,I}(p_0)$, and $\gamma\in\Gamma$ with $\gamma\neq\tilde{\gamma}_{n,p_{0}}$, and set $t_{n}:=(\gamma-\tilde{\gamma}_{n,p_{0}})/||\gamma-\tilde{\gamma}_{n,p_{0}}||_2$. By Condition 5, $\eta t_{n}\in\conv\{\nabla_\gamma q_{n,a}(\tilde{\gamma}_{n,p_{0}},p_0):a\in\mathcal{A}^{*}_n(\tilde{\gamma}_{n,p_{0}},p_0)\}$, so there exist weights $\{\lambda_{a}: a \in \mathcal{A}_{n}^{*}(\tilde{\gamma}_{n,p_{0}},p_{0})\}$ with $\lambda_{a}\geq0$ for all $a \in \mathcal{A}_{n}^{*}(\tilde{\gamma}_{n,p_{0}},p_{0})$, $\sum_{a \in \mathcal{A}_{n}^{*}(\tilde{\gamma}_{n,p_{0}},p_{0})}\lambda_{a}=1$, and $\eta t_{n}=\sum_{a\in \mathcal{A}_{n}^{*}(\tilde{\gamma}_{n,p_{0}},p_{0})}\lambda_{a}\nabla_{\gamma} q_{n,a}(\tilde{\gamma}_{n,p_{0}},p_{0})$. Multiplying both sides by $t_{n}$ and using $||t_{n}||_2=1$,
\begin{align*}
\eta=\sum_{a \in \mathcal{A}_{n}^{*}(\tilde{\gamma}_{n,p_{0}},p_{0})}\lambda_a\,\nabla_\gamma q_{n,a}(\tilde{\gamma}_{n,p_{0}},p_0)'t_{n}\leq \max_{a\in\mathcal{A}^{*}_n(\tilde{\gamma}_{n,p_{0}},p_{0})}\nabla_\gamma q_{n,a}(\tilde{\gamma}_{n,p_{0}},p_0)'t_{n}
\end{align*}
so there exists $a^{*}\in\mathcal{A}^{*}_n(\tilde{\gamma}_{n,p_{0}},p_0)$ such that
\begin{equation}
\nabla_\gamma q_{n,a^{*}}(\tilde{\gamma}_{n,p_{0}},p_0)'t_{n}\geq\eta.
\label{eq:kinkdir}
\end{equation}
By convexity and differentiability of $q_{n,a^{*}}(\cdot,p_0)$ on $\Gamma^{o}$ (the gradient inequality) and exact activity of $a^{\star}$ at $(\tilde{\gamma}_{n,p_{0}},p_0)$ (i.e., $q_{n,a^{*}}(\tilde{\gamma}_{n,p_{0}},p_0)=Q_n(\tilde{\gamma}_{n,p_{0}},p_0)=m_n(p_0)$),
\begin{align*}
Q_n(\gamma,p_0)&\geq q_{n,a^{*}}(\gamma,p_0)\\
&\geq q_{n,a^{*}}(\tilde{\gamma}_{n,p_{0}},p_0)+\nabla_\gamma q_{n,a^{*}}(\tilde{\gamma}_{n,p_{0}},p_0)'(\gamma-\tilde{\gamma}_{n,p_{0}}) \\
&\geq m_n(p_0)+\eta\,\|\gamma-\tilde{\gamma}_{n,p_{0}}\|_2,
\end{align*}
which implies
\begin{align*}
\tilde{Q}_n(\gamma,p_0) \geq \eta \cdot ||\gamma-\tilde{\gamma}_{n,p_{0}} ||_2
\end{align*}
by substituting $m_{n}(p_{0})$ from each side of the inequality. Since we can repeat this argument for all $\gamma \in \Gamma \setminus \{\tilde{\gamma}_{n,p_{0}}\}$ and the inequality holds trivially at $\tilde{\gamma}_{n,p_{0}}$, we conclude that
\begin{align}\label{eq:growth}
 \tilde{Q}_n(\gamma,p_0) \geq \eta \cdot ||\gamma-\tilde{\gamma}_{n,p_{0}} ||_2 \quad \text{for all $\gamma \in \Gamma$}.   
\end{align}
Applying \eqref{eq:growth} with $\gamma$ equal to any other element of $\tilde{\Gamma}_{n,I}(p_0)$ forces that element to coincide with $\tilde{\gamma}_{n,p_{0}}$, and, as a result, $\tilde{\Gamma}_{n,I}(p_0)=\{\tilde{\gamma}_{n,p_{0}}\}$ (i.e., a singleton) for all $n \geq \bar{N}$.
\end{proof}
\begin{lemma}\label{lem:localization}
Suppose that the conditions of Proposition \ref{prop:misspec_hemi_convex} hold. Then, for $P_{0,X}^{(\infty)}$-almost every fixed realization $\{x_{i}\}_{i \geq 1}$ of $\{X_{i}\}_{i \geq 1}$, there exists a sequence $\{\kappa_{n}\}_{n \geq 1}$ such that
\begin{align*}
       \sup_{p \in V_{n,\delta_{n}}}\sup_{\tilde{\gamma} \in \tilde{\Gamma}_{n,I}(p)}||\tilde\gamma-\tilde\gamma_{n,p_{0}}||_2 \leq \kappa_{n}
\end{align*}
for all $n \geq \bar{N}$, where $\bar{N} \geq 1$ is the index from Lemma \ref{lem:singletonidset} and $\tilde\gamma_{n,p_{0}}$ is the single element of $\tilde{\Gamma}_{n,I}(p_{0})$. Moreover, if $\Delta_{n} \rightarrow 0$ as $n\rightarrow \infty$, then $\kappa_{n} \rightarrow 0$ as $n\rightarrow \infty$.
\end{lemma}
\begin{proof}
Since $\tilde{Q}_n(\tilde\gamma,p)=0$ for all $\tilde{\gamma} \in \tilde{\Gamma}_{n,I}(p)$ and all $p \in V_{n,\delta_{n}}$, the triangle inequality and the inequality $|\max_{a}u_{a} - \max_{a}v_{a}| \leq \max_{a}|u_{a}-v_{a}|$ yields $\tilde{Q}_n(\tilde\gamma,p_0)\leq2\Delta_n$. Then, the inequality $\tilde{Q}_{n}(\gamma,p_{0}) \geq \eta \cdot ||\gamma-\tilde{\gamma}_{n,p_{0}}||_{2}$ for all $\gamma \in \Gamma$ and $n \geq \bar{N} $(derived in Lemma \ref{lem:singletonidset}) implies
\begin{equation}
||\tilde\gamma-\tilde{\gamma}_{n,p_{0}}||_2 \leq \frac{2\Delta_n}{\eta}=: \kappa_n
\label{eq:localize}
\end{equation}
for all $\tilde{\gamma} \in \tilde{\Gamma}_{n,I}(p)$, all $p \in V_{n,\delta_{n}}$, and $n \geq \bar{N}$, where $\eta > 0$ is the constant in Condition 5 of Proposition \ref{prop:misspec_hemi_convex}. Consequently, the first claim of the lemma holds by taking the supremum (as $\kappa_{n}$ is a uniform upper bound). The second claim holds because $\kappa_{n} = 2 \Delta_{n}/\eta$.
\end{proof}
\begin{proof}[Proof of Proposition \ref{prop:misspec_hemi_convex}]
The proof has two steps. Step 1 verifies (\ref{eq:hemi_misspec_ucon}). Step 2 verifies (\ref{eq:hemi_misspec_minorant}). 
\paragraph{Step 1.} Let $m_{n}(p) = \inf_{\gamma \in \Gamma}Q_{n}(\gamma,p)$. The infimum is well-defined because $Q_n(\cdot,p)$ is continuous on $\Gamma^{o}$ (being a finite maximum of convex functions, each of which is continuous on the open set $\Gamma^{o}$) and $\Gamma \subseteq \Gamma^{o}$ is compact, which means that the Weierstrass Extreme Value Theorem applies. Next, we apply the elementary inequality $|\max_{a}u_a-\max_{a}v_a|\leq\max_{a}|u_a-v_a|$ to obtain
\begin{align*}
\sup_{(\gamma,p)\in\Gamma\times V_{n,\delta_n}}\big|Q_n(\gamma,p)-Q_n(\gamma,p_0)\big|\leq\Delta_n
\quad\text{and}\quad
\sup_{p\in V_{n,\delta_n}}\big|m_n(p)-m_n(p_0)\big|\leq\Delta_n,
\end{align*}
and, as a result, the triangle inequality implies that
\begin{align*}
\sup_{(\gamma,p)\in\Gamma\times V_{n,\delta_n}}\big|\tilde{Q}_n(\gamma,p)-\tilde{Q}_n(\gamma,p_0)\big|\leq 2\Delta_n.
\end{align*}
Consequently, (\ref{eq:hemi_misspec_ucon}) holds because $\Delta_{n} \rightarrow 0$ as $n\rightarrow \infty$ by Condition 3.
\paragraph{Step 2.} Let $N^{*} \geq \bar{N}$ chosen so that $\sup_{\gamma \in \Gamma_{nbd}}\max_{a \in \mathcal{A}_{n}}||\nabla_{\gamma}q_{n,a}(\gamma,p_{0})||_{2} \leq \acute{C}$ for all $n \geq N^{*}$, where $\acute{C} >0$ is sufficiently large, and $\Delta_{n}' + L'\kappa_{n} < \eta/2$ for all $n \geq N^{*}$. Such an $N^{*}$ exists by Conditions 3 and 4, and Lemma \ref{lem:localization}. Next, let $\gamma \in \Gamma \setminus \tilde{\Gamma}_{n,I}(p)$ for $p \in V_{n,\delta_{n}}$. By non-emptiness and compactness of $\tilde{\Gamma}_{n,I}(p)$, there exists $\tilde{\gamma}_{n,p} \in \tilde{\Gamma}_{n,I}(p)$ such that $||\gamma-\tilde{\gamma}_{n,p}||_{2} = \text{dist}(\gamma,\tilde{\Gamma}_{n,I}(p)) > 0$. Let $t_{n,p} = (\gamma-\tilde{\gamma}_{n,p})/||\gamma-\tilde{\gamma}_{n,p}||_{2}$ and let $a^{*} \in \mathcal{A}_{n}$ be such that $\nabla_{\gamma}q_{n,a^{*}}(\tilde{\gamma}_{n,p_{0}},p_{0})'t_{n,p} \geq \eta$ for all $n \geq N^{*}$, where $\tilde{\gamma}_{n,p_{0}}$ is the single element of $\tilde{\Gamma}_{n,I}(p_{0})$. Such an $a^{*}$ exists by Condition 5. By convexity and differentiability of $q_{n,a^{*}}(\gamma,p)$ in $\gamma$,
\begin{align*}
    Q_{n}(\gamma,p) &\geq q_{n,a^{*}}(\gamma,p) \\
    &\geq q_{n,a^{*}}(\tilde{\gamma}_{n,p},p) + \nabla_{\gamma}q_{n,a^{*}}(\tilde{\gamma}_{n,p},p)'(\gamma-\tilde{\gamma}_{n,p}) \\
    &=q_{n,a^{*}}(\tilde{\gamma}_{n,p},p) + \nabla_{\gamma}q_{n,a^{*}}(\tilde{\gamma}_{n,p},p)'t_{n,p} \cdot  \text{dist}(\gamma,\tilde{\Gamma}_{n,I}(p)),
\end{align*}
Consequently, we complete the proof by verifying that the following holds:
\begin{enumerate}
    \item There is a sequence $\{r_{n}\}_{n \geq 1}$ such that $r_{n} \geq 0$ for each $n \geq 1$, $r_{n} \rightarrow 0$ as $n\rightarrow \infty$, and $q_{n,a^{*}}(\tilde{\gamma}_{n,p},p) \geq Q_{n}(\tilde{\gamma}_{n,p},p)-r_{n}$ for all $p \in V_{n,\delta_{n}}$ and $n \geq N^{*}$.
    \item There is a constant $c >0$ such that $\nabla_{\gamma}q_{n,a^{*}}(\tilde{\gamma}_{n,p},p)'t_{n,p} \geq c$ for all $n \geq N^{*}$ and all $p \in V_{n,\delta_{n}}$.
\end{enumerate}
If these conditions hold, then we verify (\ref{eq:hemi_misspec_minorant}) because, by rearranging the inequality $Q_{n}(\gamma,p)  \geq q_{n,a^{*}}(\tilde{\gamma}_{n,p},p) + \nabla_{\gamma}q_{n,a^{*}}(\tilde{\gamma}_{n,p},p)'t_{n,p} \cdot  \text{dist}(\gamma,\tilde{\Gamma}_{n,I}(p))$ derived above (and using that $\gamma \in \Gamma \setminus \tilde{\Gamma}_{n,I}(p)$ was arbitrary), we obtain that, for each $\gamma \in \Gamma$,
\begin{align*}
    \tilde{Q}_{n}(\gamma,p) \geq c \cdot \text{dist}(\gamma,\tilde{\Gamma}_{n,I}(p)) - r_{n}
\end{align*}
for all $n \geq N^{*}$ and $p \in V_{n,\delta_{n}}$. The case with $\gamma \in \tilde{\Gamma}_{n,I}(p)$ holds because $\tilde{Q}_{n}(\gamma,p) = 0$ and $\text{dist}(\gamma,\tilde{\Gamma}_{n,I}(p)) = 0$, so the same $c$ and $\{r_{n}\}_{n \geq 1}$ can be used. The remaining substeps verify these claims.
\paragraph{Step 2A.} We claim that $q_{n,a^{*}}(\tilde{\gamma}_{n,p},p) \geq Q_{n}(\tilde{\gamma}_{n,p},p)-r_{n}$ for all $p \in V_{n,\delta_{n}}$ and $n \geq N^{*}$, where $r_{n} =  \acute{C} \kappa_{n} + 2\Delta_{n}$ and $\kappa_{n}$ is the sequence from Lemma \ref{lem:localization}. Using Condition 3 and the convex, differentiability of $q_{n,a^{*}}$ in $\gamma$, 
\begin{align*}
    q_{n,a^{*}}(\tilde{\gamma}_{n,p},p) &\geq q_{n,a^{*}}(\tilde{\gamma}_{n,p},p_{0})-\Delta_{n} \\
    &\geq q_{n,a^{*}}(\tilde{\gamma}_{n,p_{0}},p_{0}) + \nabla_{\gamma} q_{n,a^{*}}(\tilde{\gamma}_{n,p_{0}},p_{0})'(\tilde{\gamma}_{n,p}-\tilde{\gamma}_{n,p_{0}})-\Delta_{n},
\end{align*}
Using the definition of $m_{n}(p)$, $q_{n,a^{*}}(\tilde{\gamma}_{n,p_{0}},p_{0}) = m_{n}(p_{0})$, and, by the result in Step 1, 
\begin{align*}
q_{n,a^{*}}(\tilde{\gamma}_{n,p_{0}},p_{0}) \geq m_{n}(p) - \Delta_{n} = Q_{n}(\tilde{\gamma}_{n,p},p)-\Delta_{n}   
\end{align*} 
for all $n \geq N^{*}$ and all $p \in V_{n,\delta_{n}}$, where the equality uses that $\tilde{\gamma}_{n,p} \in \tilde{\Gamma}_{n,I}(p)$. Applying Conditions 4 and 5 and the Cauchy-Schwarz inequality,
\begin{align*}
    \nabla_{\gamma} q_{n,a^{*}}(\tilde{\gamma}_{n,p_{0}},p_{0})'(\gamma_{n,p}-\gamma_{n,p_{0}}) &\geq  - \acute{C} \kappa_{n}
\end{align*}
for all $n \geq N^{*}$ and $p \in V_{n,\delta_{n}}$. Putting these two results together,
\begin{align*}
    q_{n,a^{*}}(\tilde{\gamma}_{n,p_{0}},p_{0}) \geq Q_{n}(\tilde{\gamma}_{n,p},p)-r_{n}
\end{align*}
for all $n \geq N^{*}$ and all $p \in V_{n,\delta_{n}}$. Since $\Delta_{n} \geq 0$, $\kappa_{n} \geq 0$, and $\Delta_{n}\rightarrow 0$ as $n\rightarrow \infty$, $r_{n} \rightarrow 0$ as $n\rightarrow \infty$. This verifies the first claim.
\paragraph{Step 2B.} We show that there exists a constant $c>0$ such that $\nabla_{\gamma}q_{n,a^{*}}(\tilde{\gamma}_{n,p},p)'t_{n,p} \cdot  \text{dist}(\gamma,\tilde{\Gamma}_{n,I}(p)) \geq c \cdot \text{dist}(\gamma,\tilde{\Gamma}_{n,I}(p)) $ for all $n \geq N^{*}$ and all $p \in V_{n,\delta_{n}}$. By the Cauchy--Schwarz and triangle inequalities,
\begin{align*}
\nabla_\gamma q_{n,a^{*}}(\tilde\gamma_{n,p},p)'t_{n,p}
&\geq\nabla_\gamma q_{n,a^{*}}(\tilde\gamma_{n,p_{0}},p_0)'t_{n,p}
-\big\|\nabla_\gamma q_{n,a^{*}}(\tilde\gamma_{n,p},p)-\nabla_\gamma q_{n,a^{\star}}(\tilde\gamma_{n,p_{0}},p_0)\big\|_2\\ 
&\geq \nabla_\gamma q_{n,a^{*}}(\tilde\gamma_{n,p_{0}},p_0)'t_{n,p}
-\left| \left|\nabla_\gamma q_{n,a^{*}}(\tilde\gamma_{n,p},p)-\nabla_\gamma q_{n,a^{*}}(\tilde\gamma_{n,p},p_{0})\right |\right|_{2} \\
&\quad - \left| \left|\nabla_\gamma q_{n,a^{\star}}(\tilde\gamma_{n,p},p_0) -\nabla_\gamma q_{n,a^{\star}}(\tilde\gamma_{n,p_{0}},p_0)\right| \right|_{2} \\
&\geq\eta-\Delta'_n-L'\kappa_n \\
&=: c_{n},
\end{align*}
for all $p \in V_{n,\delta_{n}}$ and $n \geq N^{*}$, where the last inequality uses Conditions 3, 4, and 5. Since $\Delta_{n}'+ L'\kappa_{n} < \eta/2$ for $n \geq N^{*}$, it follows that $c_{n} > \eta /2$ for all $n \geq N^{*}$, so set $c = \eta/2$ and we verify the claim.
\end{proof}
\begin{proof}[Proof of Proposition \ref{prop:relaxed}]
Fix a realization for which Conditions 1 and 2 hold and let $n\geq\max\{\bar{N},\check{N}\}$, where $\check{N}$ is chosen such that $u_{n} \leq t$ for all $n \geq \check{N}$, and $p\in V_{n,\delta_n}$. Such an $\check{N}$ exists because $u_{n}\rightarrow 0$ as $n\rightarrow \infty$. If $\gamma\in\tilde{\Gamma}_{n,I}(p_0)$, then $\tilde{Q}_n(\gamma,p)\leq\tilde{Q}_n(\gamma,p_0)+u_n=u_n\leq t$ for all $n \geq \max\{\bar{N},\check{N}\}$, so $\gamma\in\tilde{\Gamma}^{\,t}_{n,I}(p)$; this establishes the containment and, in particular, $\sup_{\gamma\in\tilde{\Gamma}_{n,I}(p_0)}\dist(\gamma,\tilde{\Gamma}^{\,t}_{n,I}(p))=0$. Conversely, if $\gamma\in\tilde{\Gamma}^{\,t}_{n,I}(p)$, then $\tilde{Q}_n(\gamma,p_0)\leq\tilde{Q}_n(\gamma,p)+u_n\leq 2t$ for all $n \geq \max\{\bar{N},\check{N}\}$, so Condition 1 gives $\dist(\gamma,\tilde{\Gamma}_{n,I}(p_0))\leq 2t/\eta$ for all $n \geq \max\{\bar{N},\check{N}\}$. Combining these with the definition of the Hausdorff distance completes the proof.
\end{proof}
\subsection{Proofs for Section \ref{sec:GPverification} and Some Additional Discussion.}
Before proving Propositions \ref{prop:L2} and \ref{prop:post_sup_norm} (and the associated lemmas), we introduce Kullback-Leibler `neighborhoods',
\begin{align*}
    \mathcal{U}_{n,2}^{*}(p_{0},\delta) = \left\{p : \frac{1}{n}\sum_{i=1}^{n}KL(p_{0}(x_{i}),p(x_{i})) < \delta^{2}, \quad \frac{1}{n}\sum_{i=1}^{n}KLV(p_{0}(x_{i}),p(x_{i}))<\delta^{2}\right\}
\end{align*}
for each $\delta > 0$, where
\begin{align*}
    KL(p_{0}(x),p(x)) = \sum_{k=1}^{K}p_{0,k}(x)\log\left(\frac{p_{0,k}(x)}{p_{k}(x)}\right)
\end{align*}
and
\begin{align*}
    KLV(p_{0}(x),p(x)) = \sum_{k=1}^{K}p_{0,k}(x)\left(\log\left(\frac{p_{0,k}(x)}{p_{k}(x)}\right)\right)^{2}
\end{align*}
are the Kullback-Leibler divergence and variation of the categorical distribution, respectively.
\begin{lemma}\label{lem:KLdivergence}
Suppose that Assumptions \ref{as:covariates}.1 and \ref{as:GP}.2 hold, $b_{0,k} \in \overline{\mathcal{H}}_{k}$, $||b_{0,k}||_{\infty} \leq \overline{B}_{k}$, and $\{\tilde{\delta}_{n,k}\}_{n \geq 1}$ is such that $n\tilde{\delta}_{n,k}^{2}\rightarrow \infty$ as $n\rightarrow \infty$ and $\varphi_{b_{0,k}}(\tilde{\delta}_{n,k}) \leq n \tilde{\delta}_{n,k}^{2}$ for each $k=1,...,K-1$. Then there are constants $C_{1}, C_{2} \in (0,\infty)$ such that
\begin{align*}
\Pi\left(p \in \mathcal{U}_{n,2}^{*}(p_{0},C_{1}\tilde{\delta}_{n})\right) \geq C_{2} \exp\left(-(K-1)n\tilde{\delta}_{n}^{2}\right) 
\end{align*}
for $P_{0,X}^{(\infty)}$-almost every fixed realization $\{x_{i}\}_{i \geq 1}$ of $\{X_{i}\}_{i \geq 1}$, where $\tilde{\delta}_{n} = \max_{1 \leq k \leq K-1}\tilde{\delta}_{n,k}$.
\end{lemma}
\begin{proof}
    Since $||B_{k}||_{\infty} \leq \overline{B}_{k}$ and $||b_{0,k}||_{\infty} \leq \bar{B}_{k}$ for each $k=1,...,K-1$ and $\max_{1 \leq k \leq K-1}\overline{B}_{k} < \infty$, a Taylor expansion reveals
\begin{align*}
\max\left\{\frac{1}{n}\sum_{i=1}^{n}KL(p_{0}(x_{i}),p(x_{i})),    \frac{1}{n}\sum_{i=1}^{n}V(p_{0}(x_{i}),p(x_{i})) \right\} \leq \breve{C}\sum_{k=1}^{K} ||p_{k}-p_{0,k}||_{\infty}^{2},
\end{align*}
where $\breve{C}  \in (0,\infty)$ is a constant that depends on $\overline{B}_{1},...,\overline{B}_{K-1}$. Consequently, it suffices to show there exists constants $C_{1},C_{2} \in (0,\infty)$ such that
\begin{align*}
    \Pi\left(\sqrt{\sum_{k=1}^{K}||p_{k}-p_{0,k}||_{\infty}^{2}} < C_{1}\tilde{\delta}_{n}\right) \geq C_{2}\exp\left(-(K-1)n\tilde{\delta}_{n}^{2}\right)
\end{align*}
To this end, we relate $||p_{k}-p_{0,k}||_{\infty}$ to $||B_{k}-b_{0,k}||_{\infty}$ by considering three cases:
\begin{enumerate}
    \item When $k= 1$, we know that $p_{1} = q_{1}$ which implies
    \begin{align*}
    ||p_{1}-p_{0,1}||_{\infty} = ||\Lambda(B_{1})-\Lambda(b_{0,1})||_{\infty}\leq \frac{1}{4}||B_{1}-b_{0,1}||_{\infty}    
    \end{align*} 
    where the inequality uses that $z \mapsto \Lambda(z)$ is Lipschitz with constant $1/4$.
    \item When $k \in \{2,...,K-1\}$, we know that $p_{k} = q_{k} \prod_{j=1}^{k-1}(1-q_{j})$, and, as a result, we obtain
    \begin{align*}
    &\left| \left| p_{k}-p_{0,k} \right| \right|_{\infty} \\
    &\quad = \left| \left|q_{k}\prod_{j=1}^{k-1}(1-q_{j})- q_{0,k}\prod_{j=1}^{k-1}(1-q_{0,j})\right| \right|_{\infty}   \\
    &\quad \leq  \left| \left|q_{k}\prod_{j=1}^{k-1}(1-q_{j})-q_{0,k}\prod_{j=1}^{k-1}(1-q_{j})\right| \right| +\left| \left| q_{0,k}\left(\prod_{j=1}^{k-1}(1-q_{j})-\prod_{j=1}^{k-1}(1-q_{0,j})\right)\right| \right|_{\infty} \\
    &\quad \leq \left| \left|q_{k} - q_{0,k} \right| \right|_{\infty} \cdot \left| \left| \prod_{j=1}^{k-1}(1-q_{j})\right | \right|_{\infty}+ \left|\left|q_{0,k}\right|\right|_{\infty}\cdot \left| \left| \prod_{j=1}^{k-1}(1-q_{j})-\prod_{j=1}^{k-1}(1-q_{0,j}) \right| \right|_{\infty} \\
     &\quad \leq \left| \left|q_{k} - q_{0,k} \right| \right|_{\infty} + \left| \left| \prod_{j=1}^{k-1}(1-q_{j})-\prod_{j=1}^{k-1}(1-q_{0,j}) \right| \right|_{\infty},
    \end{align*}
    where the first inequality adds and subtracts $q_{0,k}\prod_{j=1}^{k-1}(1-q_{j})$ and applies the triangle inequality, and the final inequality uses that $q_{j} \in [0,1]$ for each $j=1,...,K-1$. Using that $q_{k} = \Lambda(B_{k})$, $q_{0,k}= \Lambda(b_{0,k})$, and that $z \mapsto \Lambda(z)$ is Lipschitz with constant $1/4$,
    \begin{align*}
     \left| \left|q_{k} - q_{0,k} \right| \right|_{\infty} \leq \frac{1}{4} \cdot \left| \left|B_{k}-b_{0,k} \right| \right|_{\infty}   
    \end{align*}
    Using the inequality
    \begin{align*}
     \left| \prod_{j=1}^{k-1}(1-q_{j}(x))-\prod_{j=1}^{k-1}(1-q_{0,j}(x)) \right|  \leq  \sum_{j=1}^{k-1}\left| q_{j}(x) - q_{0,j}(x) \right|,  
    \end{align*}
    and the Lipschitz property of $z \mapsto \Lambda(z)$, we obtain
    \begin{align*}
\left| \left| \prod_{j=1}^{k-1}(1-q_{j})-\prod_{j=1}^{k-1}(1-q_{0,j}) \right| \right|_{\infty} &\leq \sum_{j=1}^{k-1}||q_{j}-q_{0,j}||_{\infty} \\
        &\leq \frac{1}{4}\sum_{j=1}^{k-1}||B_{k}-b_{0,k}||_{\infty}.
    \end{align*}
    Combining the upper bounds,
    \begin{align*}
        \left| \left| p_{k}-p_{0,k} \right| \right|_{\infty} \leq \frac{1}{4}\sum_{j=1}^{k}||B_{j}-b_{0,j}||_{\infty}
    \end{align*}
    for $k=2,...,K-1$.
    
    \item When $k = K$, we know that $p_{K} = 1 - \sum_{k=1}^{K-1}p_{k}$. This implies that
    \begin{align*}
        ||p_{K}-p_{0,K}||_{\infty} \leq \sum_{k=1}^{K-1}||p_{k}-p_{0,k}||_{\infty} \leq \frac{1}{4}\sum_{k=1}^{K-1}||B_{k}-b_{0,k}||_{\infty}
    \end{align*}
\end{enumerate}
Combining these cases, we obtain
\begin{align*}
\sqrt{\sum_{k=1}^{K-1}||p_{k}-p_{0,K}||_{\infty}^{2}} \leq C(K)\sqrt{\sum_{k=1}^{K-1}||B_{k}-b_{0,k}||_{\infty}^{2}},
\end{align*}
where $C(K)$ is some finite constant that depends on $K$. Consequently,
\begin{align*}
  \Pi\left(\sqrt{\sum_{k=1}^{K}||p_{k}-p_{0,k}||_{\infty}^{2}} < C_{1}\tilde{\delta}_{n}\right)  &\geq \Pi\left(\sum_{k=1}^{K-1}||B_{k}-b_{0,k}||_{\infty} < \frac{C_{1}\tilde{\delta}_{n}}{(\sqrt{K-1})C(K)}\right) \\
  &\geq \Pi\left(\max_{1 \leq k \leq K-1}||B_{k}-b_{0,k}||_{\infty}< \frac{C_{1}\tilde{\delta}_{n}}{(K-1)^{\frac{3}{2}}C(K)}\right),
\end{align*}
where the inequalities use equivalence of finite product norms. Define $C_{1} = 2(K-1)^{3/2}C(K)$. Using that $B_{1},...,B_{K-1}$ are mutually independent and the definition of $C_{1}$,
\begin{align*}
&\Pi\left(\max_{1 \leq k \leq K-1}||B_{k}-b_{0,K}||_{\infty} < \frac{C_{1}\tilde{\delta}_{n}}{(K-1)^{3/2}C(K)}\right) \\ &\quad =\prod_{k=1}^{K-1}P_{GP_{\overline{B}_{k}}(0,\kappa_{k})}\left(||B_{k}-b_{0,K}||_{\infty} < 2\tilde{\delta}_{n} \right) \\
&\quad= \prod_{k=1}^{K-1}\frac{P_{GP(0,\kappa_{k})}\left(||B_{k}-b_{0,K}||_{\infty} < 2\tilde{\delta}_{n} \right)}{P_{GP(0,\kappa_{k})}(||B_{k}||_{\infty} \leq \overline{B}_{k})} \\
     &\quad \geq C_{2}\prod_{k=1}^{K-1}\exp\left(-\varphi_{b_{0,k}}\left(\tilde{\delta}_{n} \right)\right), 
\end{align*}
where $P_{GP(0,\kappa_{k})}$ and $P_{GP_{\bar{B}_{k}}(0,\kappa_{k})}$ are the probability measures associated with $GP(0,\kappa_{k})$ and the truncation $GP_{\overline{B}_{k}}(0,\kappa_{k})$ of $GP(0,\kappa_{k})$ to $||B_{k}||_{\infty} \leq \bar{B}_{k}$, respectively, and the last inequality applies Lemma 5.3 of \cite{van2008reproducing} and defines $C_{2} > 0$ as $C_{2} = \prod_{k=1}^{K-1}P_{GP(0,\kappa_{k})}(||B_{k}||_{\infty} \leq \overline{B}_{k}) > 0$ (positivity holds by Lemma 5.1 of \cite{van2008reproducing}). Since $\delta \mapsto \varphi_{b_{0,k}}(\delta)$ is strictly decreasing (Lemma 3 of \cite{castillo2008lower}), we know that
\begin{align*}
 \varphi_{b_{0,k}}\left(\tilde{\delta}_{n} \right) \leq \varphi_{b_{0,k}}\left(\tilde{\delta}_{n,k} \right) \leq n \tilde{\delta}_{n,k}^{2} \leq n \tilde{\delta}_{n}^{2}.
\end{align*}
Putting everything together, we conclude that 
\begin{align*}
\Pi\left(p \in \mathcal{U}_{n,2}^{*}(p_{0},C_{1}\tilde{\delta}_{n})\right) \geq C_{2} \exp\left(-(K-1)n\tilde{\delta}_{n}^{2}\right)    
\end{align*}
holds for $P_{0,X}^{(\infty)}$-almost every fixed realization $\{x_{i}\}_{i \geq 1}$ of $\{X_{i}\}_{i \geq 1}$.
\end{proof}
\begin{lemma}\label{lem:normalizing}
Let $\bar{\Pi}_{n}$ be the renormalized restriction of $\Pi$ to the event $\mathcal{U}_{n,2}^{*}(p_{0},C_{1}n\tilde{\delta}_{n}^{2})$. If Assumption \ref{as:dgp}.1 holds, then, for any $C> 0$,
\begin{align*}
P_{0}^{(n)}\left(\int \frac{L_{n}(p)}{L_{n}(p_{0})}d\bar{\Pi}_{n}(p) \leq \exp(-(C_{1}^{2}+C)n\tilde{\delta}_{n}^{2})\right) \longrightarrow 0
\end{align*}
as $n\rightarrow \infty$, conditionally given $P_{0,X}^{(\infty)}$-almost every fixed realization $\{x_{i}\}_{i \geq 1}$ of $\{X_{i}\}_{i \geq 1}$.
\end{lemma}
\begin{proof}
Applying Jensen's inequality,
\begin{align*}
  \log  \int \frac{L_{n}(p)}{L_{n}(p_{0})}d\bar{\Pi}_{n}(p) \geq  \int \left\{\log L_{n}(p) - \log L_{n}(p_{0})\right\}d\bar{\Pi}_{n}(p)
\end{align*}
Consequently,
\begin{align*}
    &P_{0}^{(n)}\left(\int \frac{L_{n}(p)}{L_{n}(p_{0})}d\bar{\Pi}_{n}(p) < \exp(-(C_{1}^{2}+C)n\tilde{\delta}_{n}^{2})\right) \\
    &\leq P_{0}^{(n)}\left(\int \left\{\log L_{n}(p) - \log L_{n}(p_{0})\right\}d\bar{\Pi}_{n}(p) < -(C_{1}^{2}+C)n\tilde{\delta}_{n}^{2}\right)
\end{align*}
Adding/subtracting $E_{P_{0}^{(n)}}(\log L_{n}(p) - \log L_{n}(p_{0}))$ and using that 
\begin{align*}
E_{P_{0}^{(n)}}(\log L_{n}(p) - \log L_{n}(p_{0})) = -\sum_{i=1}^{n}KL(p_{0}(x_{i}),p(x_{i})) \geq - C_{1}^{2}n\tilde{\delta}_{n}^{2}    
\end{align*} 
for any $p \in \mathcal{U}_{n,2}^{*}(p_{0},C_{1}\tilde{\delta}_{n})$, it follows that
\begin{align*}
    &P_{0}^{(n)}\left(\int \left\{\log L_{n}(p) - \log L_{n}(p_{0})\right\}d\bar{\Pi}_{n}(p) < -(C_{1}^{2}+C)n\tilde{\delta}_{n}^{2}\right) \\
    &\quad \leq     P_{0}^{(n)}\left(\int R_{n}(p,p_{0})d\bar{\Pi}_{n}(p) < -Cn\tilde{\delta}_{n}^{2}\right), 
\end{align*}
where $R_{n}(p,p_{0}) = (\log L_{n}(p) - \log L_{n}(p_{0})) - E_{P_{0}^{(n)}}(\log L_{n}(p) - \log L_{n}(p_{0}))$. Applying Chebyshev's inequality,
\begin{align*}
    P_{0}^{(n)}\left(\int R_{n}(p,p_{0})d\bar{\Pi}_{n}(p) < -Cn\tilde{\delta}_{n}^{2}\right) &\leq P_{0}^{(n)}\left(\left|\int R_{n}(p,p_{0})d\bar{\Pi}_{n}(p)\right| > Cn\tilde{\delta}_{n}^{2}\right) \\
    &\leq \frac{1}{C_{1}^{4}n^{2}\tilde{\delta}_{n}^{4}}E_{P_{0}^{(n)}}\left(\int R_{n}(p,p_{0})d\bar{\Pi}_{n}(p)\right)^{2}
\end{align*}
Applying Jensen's inequality and Fubini's theorem,
\begin{align*}
    E_{P_{0}^{(n)}}\left(\int R_{n}(p,p_{0})d\bar{\Pi}_{n}(p)\right)^{2} &\leq     E_{P_{0}^{(n)}}\int \left(R_{n}(p,p_{0})\right)^{2}d\bar{\Pi}_{n}(p) \\
    &=\int Var_{P_{0}^{(n)}}\left(R_{n}(p,p_{0})\right)d\bar{\Pi}_{n}(p) \\
    &\leq C_{1}^{2}n\tilde{\delta}_{n}^{2},
\end{align*}
where the last inequality uses the definition of $\mathcal{U}_{n,2}^{*}(p_{0},C_{1}\tilde{\delta}_{n})$ and the conditional independence of elements of $Y^{(n)}$. Combining this with the above and using that $n\tilde{\delta}_{n}^{2}\rightarrow \infty$, we conclude that
\begin{align*}
    P_{0}^{(n)}\left(\int \frac{L_{n}(p)}{L_{n}(p_{0})}d\bar{\Pi}_{n}(p) < \exp(-(C_{1}+C)n\tilde{\delta}_{n}^{2})\right) \leq \frac{1}{C_{1}^{2}n\tilde{\delta}_{n}^{2}} \longrightarrow 0
\end{align*}
as $n\rightarrow \infty$ for $P_{0,X}^{(\infty)}$-almost every fixed realization $\{x_{i}\}_{i \geq 1}$ of $\{X_{i}\}_{i \geq 1}$.
    
\end{proof}
\begin{proof}[Proof of Proposition \ref{prop:L2}]
The proof has two steps. Step 1 derives a sufficient condition. Step 2 verifies the sufficient condition.
\paragraph{Step 1.} Since $0 \leq p_{k} \leq 1$ for each $k=1,...,K$, there exists a constant $C>0$ (that is independent of $\{x_{i}\}_{i \geq 1}$ and $K$) such that $$||p-p_{0}||_{n,2} \leq C h_{n}(p,p_{0}),$$ where $h_{n}^{2}(p,p_{0}) = n^{-1}\sum_{i=1}^{n}h^{2}(p(x_{i}),p_{0}(x_{i}))$ and $h(p(x_{i}),p_{0}(x_{i}))$ is the Hellinger distance between the PMFs $p(x_{i})$ and $p_{0}(x_{i})$ for each $i=1,...,n$. Consequently, we show that
\begin{align*}
\Pi_{n}\left(h_{n}(p,p_{0}) \geq M\tilde{\delta}_{n}|Y^{(n)}\right) \overset{P_{0}^{(n)}}{\longrightarrow} 0     
\end{align*}
as $n\rightarrow \infty$, conditionally given $P_{0,X}^{(\infty)}$-almost every fixed realization $\{x_{i}\}_{i \geq 1}$ of $\{X_{i}\}_{i \geq 1}$. To that end, let $A_{n}$ be the event that
\begin{align*}
    \int \frac{L_{n}(p)}{L_{n}(p_{0})} d\Pi(p) \geq \exp(-(C_{1}^{2}+1)n\tilde{\delta}_{n}^{2})\Pi(p \in \mathcal{U}_{n,2}^{*}(p_{0},C_{1}\tilde{\delta}_{n})).
\end{align*}
By Lemma \ref{lem:normalizing}, $P_{0}^{(n)}(A_{n}) = 1 + o(1)$ as $n\rightarrow \infty$ for $P_{0,X}^{(\infty)}$-almost every fixed realization $\{x_{i}\}_{i \geq 1}$ of $\{X_{i}\}_{i \geq 1}$. Consequently, it suffices to show that
\begin{align}\label{eq:onAn}
 \Pi_{n}\left(h_{n}(p,p_{0}) \geq M\tilde{\delta}_{n}|Y^{(n)}\right)\mathbf{1}\{A_{n}\} \overset{P_{0}^{(n)}}{\longrightarrow} 0   
\end{align}
as $n\rightarrow \infty$, conditionally given $P_{0,X}^{(\infty)}$-almost every fixed realization $\{x_{i}\}_{i \geq 1}$ of $\{X_{i}\}_{i \geq 1}$. Let $\{\psi_{n}(Y^{(n)})\}_{n \geq 1}$ be a sequence of tests (i.e., $\psi_{n}(Y^{(n)}) \in \{0,1\}$ for each $n \geq 1$) for which there exists universal constants $D >0$ and $\xi >0$ such that, for any $\delta > 0$ and $p_{1}$ with $h_{n}(p_{0},p_{1}) > \delta$,
\begin{align}\label{eq:null}
    E_{P_{0}^{(n)}}[\psi_{n}(Y^{(n)})] \leq \exp(-Dn\delta^{2})
\end{align}
and
\begin{align}\label{eq:alternative}
\sup_{p: h_{n}(p,p_{1})<\xi \delta}E_{P^{(n)}}[1-\psi_{n}(Y^{(n)})] \leq \exp(-Dn\delta^{2}).
\end{align}
for $P_{0,X}^{(\infty)}$-almost every fixed realization $\{x_{i}\}_{i \geq 1}$ of $\{X_{i}\}_{i \geq 1}$, where $P^{(n)}= \bigotimes_{i=1}^{n}Categorical(p(x_{i}))$. Lemma 2 of \cite{ghosal2007convergence} establishes such tests exist. Recognizing that
\begin{align*}
     &\Pi_{n}\left(h_{n}(p,p_{0}) \geq M\tilde{\delta}_{n}|Y^{(n)}\right)\mathbf{1}\{A_{n}\} \\
     &\quad =  \Pi_{n}\left(h_{n}(p,p_{0}) \geq M\tilde{\delta}_{n}|Y^{(n)}\right)\mathbf{1}\{A_{n}\} \psi_{n}(Y^{(n)}) \\
     &\quad \quad + \Pi_{n}\left(h_{n}(p,p_{0}) \geq M\tilde{\delta}_{n}|Y^{(n)}\right)\mathbf{1}\{A_{n}\}(1- \psi_{n}(Y^{(n)})) \\
     &\quad \leq \psi_{n}(Y^{(n)}) + \Pi_{n}\left(h_{n}(p,p_{0}) \geq M\tilde{\delta}_{n}|Y^{(n)}\right)\mathbf{1}\{A_{n}\}(1- \psi_{n}(Y^{(n)})) \\
     &\quad =\Pi_{n}\left(h_{n}(p,p_{0}) \geq M\tilde{\delta}_{n}|Y^{(n)}\right)\mathbf{1}\{A_{n}\}(1- \psi_{n}(Y^{(n)})) + o_{P_{0}^{(n)}}(1)
\end{align*}
with the last equality using (\ref{eq:null}), we conclude that (\ref{eq:onAn}) reduces to showing
\begin{align}\label{eq:contractionalternative}
    \Pi_{n}\left(h_{n}(p,p_{0}) \geq M\tilde{\delta}_{n}|Y^{(n)}\right)\mathbf{1}\{A_{n}\}(1- \psi_{n}(Y^{(n)}))\overset{P_{0}^{(n)}}{\longrightarrow} 0   
\end{align}
as $n\rightarrow \infty$, conditionally given $P_{0,X}^{(\infty)}$-almost every fixed realization $\{x_{i}\}_{i \geq 1}$ of $\{X_{i}\}_{i \geq 1}$. The next step shows (\ref{eq:contractionalternative}). 
\paragraph{Step 2.} Let 
\begin{align*}
\mathcal{P}_{n}^{*} = \left\{(p_{1},...,p_{K}): B_{k} \in \mathcal{B}_{n,k} \ \forall \ k=1,...,K-1\right\},    
\end{align*} 
where 
\begin{align*}
\mathcal{B}_{n,k} = \tilde{\delta}_{n}C^{1}([0,1]^{d_{x}},\mathbb{R})+ M_{n,k}\mathcal{H}_{k}^{1}     
\end{align*}
for each $k=1,...,K-1$, $C^{1}([0,1]^{d_{x}},\mathbb{R})$ is the unit ball in $(C([0,1]^{d_{x}},\mathbb{R}),||\cdot||_{\infty})$, $M_{n,k} = -2 \Phi^{-1}(\exp(-\tilde{C}_{k}n\tilde{\delta}_{n}^{2}))$ with $\tilde{C}_{k} > 1$ is a sufficiently large constant (to be chosen later), and $\mathcal{H}_{k}^{1}$ is the unit ball in the RKHS $(\mathcal{H}_{k},||\cdot||_{\mathcal{H}_{k}})$. Recognizing that
\begin{align*}
    \left\{h_{n}(p,p_{0}) \geq M\tilde{\delta}_{n}\right\} &= \left(\left\{h_{n}(p,p_{0}) \geq M\tilde{\delta}_{n}\right\} \cap \mathcal{P}_{n}^{*}\right) \cup \left( \left\{h_{n}(p,p_{0}) \geq M\tilde{\delta}_{n}\right\} \cap (\mathcal{P}\setminus \mathcal{P}_{n}^{*})\right) \\
    &\subseteq\left(\left\{h_{n}(p,p_{0}) \geq M\tilde{\delta}_{n}\right\} \cap \mathcal{P}_{n}^{*}\right) \cup (\mathcal{P}\setminus \mathcal{P}_{n}^{*})
\end{align*}
and using that $(1-\psi_{n}(Y^{(n)})) \leq 1$,
\begin{align*}
    &\Pi_{n}\left(h_{n}(p,p_{0}) \geq M\tilde{\delta}_{n}|Y^{(n)}\right)\mathbf{1}\{A_{n}\}(1- \psi_{n}(Y^{(n)})) \\
    &\quad\leq \Pi_{n}\left(p \in \mathcal{P}_{n}^{*}: h_{n}(p,p_{0}) \geq M\tilde{\delta}_{n}|Y^{(n)}\right)\mathbf{1}\{A_{n}\}(1- \psi_{n}(Y^{(n)}))\\
    &\quad \quad + \Pi_{n}\left(\mathcal{P} \setminus \mathcal{P}_{n}^{*}|Y^{(n)}\right)\mathbf{1}\{A_{n}\}
\end{align*}
The following substeps show that the upper bound converges to zero in $P_{0}^{(n)}$-probability as $n\rightarrow \infty$, conditionally given $P_{0,X}^{(\infty)}$-almost every fixed realization $\{x_{i}\}_{i \geq 1}$ of $\{X_{i}\}_{i \geq 1}$.
\paragraph{Step 2A.}
For the first term, we apply the definition of $A_{n}$, Lemma \ref{lem:KLdivergence}, and Tonelli's theorem to conclude that
\begin{align*}
    &E_{P_{0}^{(n)}}\left(\Pi_{n}\left(p \in \mathcal{P}_{n}^{*}: h_{n}(p,p_{0}) \geq M\tilde{\delta}_{n}|Y^{(n)}\right)\mathbf{1}\{A_{n}\}(1- \psi_{n}(Y^{(n)}))\right) \\
    &\quad \leq \exp\left(\left(K+C_{1}^{2}\right)n\tilde{\delta}_{n}^{2}\right)E_{P_{0}^{(n)}}\left(\int_{\mathcal{P}_{n}^{*} \cap \{h_{n}(p,p_{0}) \geq M\tilde{\delta}_{n}\}}\frac{L_{n}(p)}{L_{n}(p_{0})}\left(1-\psi_{n}(Y^{(n)})\right)d\Pi(p)\right) \\
    &\quad \leq \exp\left(\left(K+C_{1}^{2}\right)n\tilde{\delta}_{n}^{2}\right)\sup_{p \in \mathcal{P}_{n}^{*}: h_{n}(p,p_{0}) \geq M\tilde{\delta}_{n}}E_{P^{(n)}}\left(1-\psi_{n}(Y^{(n)})\right),
\end{align*}
where $P^{(n)}= \bigotimes_{i=1}^{n}Categorical(p(x_{i}))$. Consequently, we want to show that there exists $\check{C} > K+C_{1}^{2}$ such that
\begin{align*}
    \sup_{p \in \mathcal{P}_{n}^{*}: h_{n}(p,p_{0}) \geq M\tilde{\delta}_{n}}E_{P^{(n)}}\left(1-\psi_{n}(Y^{(n)})\right) \leq \exp(-\check{C}n\tilde{\delta}_{n}^{2}).
\end{align*}
To this end, we show that the metric entropy $\log N(\tau,\mathcal{P}_{n}^{*},h_{n})$ of $\mathcal{P}_{n}^{*}$ with respect to $h_{n}$ satisfies
\begin{align}\label{eq:entropy}
\log N(\xi \tilde{\delta}_{n},\mathcal{P}_{n}^{*},h_{n}) \lesssim n \tilde{\delta}_{n}^{2}    
\end{align}
because, by Lemma 9 of \cite{ghosal2007convergence}, this implies that
\begin{align*}
        \sup_{p \in \mathcal{P}_{n}^{*}: h_{n}(p,p_{0}) \geq M\tilde{\delta}_{n}}E_{P^{(n)}}\left(1-\psi_{n}(Y^{(n)})\right) \leq \exp(-DM^{2}n\tilde{\delta}_{n}^{2}),
\end{align*}
and, as a result, by choosing $M$ sufficiently large, we find $\check{C}$ such that $\check{C} > K + C_{1}^{2}$ and conclude that
\begin{align*}
     \sup_{p \in \mathcal{P}_{n}^{*}: h_{n}(p,p_{0}) \geq M\tilde{\delta}_{n}}E_{P^{(n)}}\left(1-\psi_{n}(Y^{(n)})\right) \longrightarrow 0
\end{align*}
as $n\rightarrow \infty$, conditionally given $P_{0,X}^{(\infty)}$-almost every fixed realization $\{x_{i}\}_{i \geq 1}$ of $\{X_{i}\}_{i \geq 1}$. To show (\ref{eq:entropy}), we first note that, since $p_{k}$ is bounded away from zero and one (by $||B_{k}||_{\infty} \leq \bar{B}_{k}$) and $(B_{1},...,B_{K-1}) \mapsto p_{k}$ is Lipschitz for all $k=1,...,K$, there exists a constant $\acute{C}:=\acute{C}(K,\bar{B}_{1},...,\bar{B}_{K-1}) > 0$ (independent of $\{x_{i}\}_{i \geq 1}$) such that
\begin{align*}
    h_{n}(p,\tilde{p}) \leq \acute{C}\max_{1 \leq k \leq K-1}||B_{k}-\tilde{B}_{k}||_{\infty}
\end{align*}
Consequently, one can show that
\begin{align*}
\log N(\xi \tilde{\delta}_{n},\mathcal{P}_{n}^{*},h_{n}) &\leq \log N\left(\frac{\xi \tilde{\delta}_{n}}{\acute{C}},\mathcal{B}_{n,1}\times \cdots \times \mathcal{B}_{n,K-1},\max_{1 \leq k \leq K-1}||\cdot||_{\infty}\right)    \\
&\leq \sum_{k=1}^{K-1}\log N\left(\frac{\xi \tilde{\delta}_{n}}{\acute{C}},\mathcal{B}_{n,k},||\cdot||_{\infty}\right) 
\end{align*}
where the second inequality uses metric entropy properties with respect to product norms.
Theorem 2.1 of \cite{vaart2008rates} shows that 
\begin{align*}
    \log N\left(\frac{\xi \tilde{\delta}_{n}}{\acute{C}},\mathcal{B}_{n,k},||\cdot||_{\infty}\right) \lesssim n \tilde{\delta}_{n}^{2}
\end{align*}
for our choice of $\mathcal{B}_{n,k}$. Consequently, by combining all of the above, we conclude (\ref{eq:entropy}).
\paragraph{Step 2B.} For the second term, we apply the definition of $A_{n}$, Lemma \ref{lem:KLdivergence}, and Tonelli's theorem to conclude that 
\begin{align*}
E_{P_{0}^{(n)}}\left(\Pi_{n}\left(\mathcal{P} \setminus \mathcal{P}_{n}^{*}|Y^{(n)}\right)\mathbf{1}\{A_{n}\}\right) &\leq \exp\left(\left(K+C_{1}\right)n\tilde{\delta}_{n}^{2}\right)E_{P_{0}^{(n)}}\left(\int_{\mathcal{P} \setminus \mathcal{P}_{n}^{*}}\frac{L_{n}(p)}{L_{n}(p_{0})}d\Pi(p)\right) \\
&=\exp\left(\left(K+C_{1}\right)n\tilde{\delta}_{n}^{2}\right)\Pi(p \notin \mathcal{P}_{n}^{*})
\end{align*}
Letting $C^{\dagger} = \min_{1 \leq k \leq K-1}P_{GP(0,\kappa_{k})}(||B_{k}||_{\infty} \leq \overline{B}_{k}) > 0$ (by Lemma 5.1 in \cite{van2008reproducing}) and applying the union bound,
\begin{align*}
    \Pi(p \notin \mathcal{P}_{n}^{*}) \leq \sum_{j=1}^{K-1}P_{GP_{\overline{B}_{k}}(0,\kappa_{k})}(B_{k} \notin \mathcal{B}_{n,k}) \leq \frac{1}{C^{\dagger}}\sum_{k=1}^{K-1}P_{GP(0,\kappa_{k})}\left(B_{k} \notin \mathcal{B}_{n,k}\right)
\end{align*}
Following the argument of Theorem 2.1 in \cite{vaart2008rates}, we obtain
\begin{align*}
    P_{GP(0,\kappa_{k})}\left(B_{k} \notin \mathcal{B}_{n,k}\right) \leq \exp\left(-\tilde{C}_{k}n\tilde{\delta}_{n}^{2}\right),
    \end{align*}
    which implies
    \begin{align*}
    \Pi(p \notin \mathcal{P}_{n}^{*}) \leq \frac{1}{C^{\dagger}}\sum_{k=1}^{K-1}\exp\left(-\tilde{C}_{k}n\tilde{\delta}_{n}^{2}\right)  \leq \frac{K-1}{C^{\dagger}}\exp\left(-n\tilde{\delta}_{n}^{2}\min_{1 \leq k \leq K-1}\tilde{C}_{k}\right).
    \end{align*}
    Consequently, by setting $\min_{1 \leq k \leq K-1}\tilde{C}_{k}$ to be large, we conclude that
    \begin{align*}
        E_{P_{0}^{(n)}}\left(\Pi_{n}\left(\mathcal{P} \setminus \mathcal{P}_{n}^{*}|Y^{(n)}\right)\mathbf{1}\{A_{n}\}\right) \longrightarrow 0
    \end{align*}
    as $n\rightarrow \infty$ for $P_{0,X}^{(\infty)}$-almost every fixed realization $\{x_{i}\}_{i \geq 1}$ of $\{X_{i}\}_{i \geq 1}$.
\end{proof}
The proof of Proposition \ref{prop:post_sup_norm} relies on the theory of wavelets. Let $\{\psi_{j,l}: j \geq 0, 0 \leq l \leq 2^{jd_{x}-1}\}$ be the $d_{x}$-dimensional tensor product of the boundary-adapted wavelet basis for $L^{2}([0,1])$ from \cite{cohen1993wavelets} (`CDV wavelets' in shorthand), where $L^{2}([0,1]^{d})$ is the space of functions $f: [0,1]^{d} \rightarrow \mathbb{R}$ such that $\int_{[0,1]^{d}} f^{2}(x)dx < \infty$ and is equipped with the inner product $\langle f,g \rangle_{L^{2}([0,1]^{d})} = \int_{[0,1]^{d}} f(x)g(x)dx$. We will rely on the following properties: 1. $\{\psi_{j,l}\}$ forms an orthonormal basis for $L^{2}([0,1]^{d_{x}})$, 2. $||\psi_{j,l}||_{\infty} \lesssim 2^{jd_{x}}$ for all $j,l$, 3. $\sum_{l=0}^{2^{jd_{x}-1}}||\psi_{j,l}||_{\infty} \lesssim 2^{jd_{x}}$ for each $j \geq 0$, 4. $\langle \psi_{j,l},1 \rangle_{L^{2}([0,1]^{d_{x}})} = 0$ for $j$ large, and 5. $\{\psi_{j,l}\}$ characterizes the Besov space $\mathfrak{B}_{\infty,\infty,a}([0,1]^{d_{x}},\mathbb{R})$ for $a \leq \alpha$ with $\alpha > 0$ fixed in that $f \in \mathfrak{B}_{\infty,\infty,a}([0,1]^{d_{x}})$ iff $||f||_{\infty,\infty,a}:=\sup_{j \geq 0, 0 \leq l \leq 2^{jd_{x}-1}}|\langle f, \psi_{j,l} \rangle_{L^{2}([0,1]^{d_{x}})}| < \infty$. This last property is useful because H\"{o}lder spaces $C^{a}([0,1]^{d_{x}},\mathbb{R})$ can be continuously embedded into $\mathfrak{B}_{\infty,\infty,a}([0,1]^{d_{x}},\mathbb{R})$ (see, for example, Chapter 4 of \cite{Gine_Nickl_2021}).
\begin{assumption}\label{as:wavelet}
Let $\mathcal{CLS}^{J_{k}} = \text{span}\{\psi_{j,l}: 0 \leq j \leq J_{k}, 0 \leq l \leq 2^{jd_{x}}-1\}$ for each $k=1,...,K-1$. The collection $\{J_{k}\}_{k=1}^{K-1}$ is chosen so that 
\begin{align*}
2^{J_{k}} = o\left(\min\left\{\tilde{\delta}_{n}^{-\frac{2}{d_{x}}}, \left(\frac{n}{\log n}\right)^{\frac{1}{d_{x}}}\right\}\right)
\end{align*}
and
\begin{align*}
\left(\sqrt{n}\tilde{\delta}_{n}\right)^{\frac{1}{a_{k}-\frac{d_{x}}{2}}} = o(2^{J_{k}})
\end{align*}
where $\{\tilde{\delta}_{n}\}_{n \geq 1}$ is the empirical $L^{2}$ rate from Proposition \ref{prop:L2}, $\{a_{k}\}_{k = 1}^{K-1}$ are constants such that $d_{x}/2 < a_{k} < \alpha_{k}$ for each $k=1,...,K-1$.
\end{assumption}
\begin{proof}[Proof of Proposition \ref{prop:post_sup_norm}] The proof has three steps. Step 1 constructs a collection of sieves $\{\mathcal{P}_{n}\}_{n \geq 1}$ for which the posterior concentrates, so that, in conjunction with Proposition \ref{prop:L2}, we can restrict attention to sieved empirical $L^{2}$ neighborhoods of $p_{0}$. Step 2 derives conditions for the desired uniform convergence via projections onto CDV wavelets. Step 3 uses the results from Step 2 to conclude the theorem.
\paragraph{Step 1.} We start by constructing sieves $\{\mathcal{P}_{n}\}_{n \geq 1}$ for which the posterior concentrates as $n\rightarrow \infty$ so that, in conjunction with Proposition \ref{prop:L2}, we can restrict attention to $\mathcal{P}_{n} \cap B_{n,2}(p_{0},M\tilde{\delta}_{n})$, where $B_{n,2}(p_{0},r) = \{p : \max_{1 \leq k \leq K-1}||p_{k}-p_{0,k}||_{n,2} < r\}$. For each $k=1,...,K-1$, let $a_{k} \in (d_{x}/2,\alpha_{k})$ be a constant compatible with Assumption \ref{as:wavelet}. We show that, for $P_{0,X}^{(\infty)}$-almost every fixed realization $\{x_{i}\}_{i \geq 1}$ of $\{X_{i}\}_{i \geq 1}$, $\Pi_{n}(\mathcal{P}_{n}|Y^{(n)}) = 1 + o_{P_{0}^{(n)}}(1)$ as $n\rightarrow \infty$, where 
\begin{align*}
\mathcal{P}_{n} = \left\{(p_{1},...,p_{K}):  \max_{1 \leq k \leq K-1}||B_{k}||_{a_{k}} \leq M\sqrt{n}\tilde{\delta}_{n}\right\}.    
\end{align*} Since \begin{align*}
P_{GP_{\overline{B}_{k}}(0,\kappa_{k})}\left(||B_{k}||_{a_{k}} \geq M \sqrt{n}\tilde{\delta}_{n}\right) \leq \frac{\exp(-\check{C}M^{2}n\tilde{\delta}_{n}^{2})}{\min_{1 \leq k \leq K-1}P_{GP(0,\kappa_{k})}(||B_{k}||\leq \overline{B}_{k})}    
\end{align*}
for some $\check{C}> 0$ (by Proposition A.2.1 of \cite{vaart2023weak}), we can apply the union bound to conclude that
\begin{align*}
\Pi\left(p \notin \mathcal{P}_{n}\right) \leq \sum_{k=1}^{K-1}P_{GP_{\overline{B}_{k}}(0,\kappa_{k})}\left(||B_{k}||_{a_{k}} \geq M \sqrt{n}\tilde{\delta}_{n}\right) \leq \frac{(K-1)\exp(-\check{C}M^{2}n\tilde{\delta}_{n}^{2})}{\min_{1 \leq k \leq K-1}P_{GP(0,\kappa_{k})}(||B_{k}||\leq \overline{B}_{k})} 
\end{align*}
Following a similar argument as Step 2B of Proposition \ref{prop:L2}, it follows that, by selecting $M>0$ large enough, 
\begin{align*}
\Pi_{n}(\mathcal{P}\setminus \mathcal{P}_{n}|Y^{(n)}) \overset{P_{0}^{(n)}}{\longrightarrow} 0    
\end{align*} 
as $n\rightarrow \infty$, conditionally given $P_{0,X}^{(\infty)}$-almost every fixed realization $\{x_{i}\}_{i \geq 1}$ of $\{X_{i}\}_{i \geq 1}$. Hence, by the union bound and the complement rule,
\begin{align*}
    \Pi_{n}\left(\mathcal{P}_{n} \cap B_{n,2}(p_{0},M\tilde{\delta}_{n})\middle | Y^{(n)}\right) &\geq 1 - \Pi_{n}\left(\mathcal{P} \setminus \mathcal{P}_{n} \middle |Y^{(n)}\right) - \Pi_{n}\left(\mathcal{P} \setminus B_{n,2}(p_{0},M\tilde{\delta}_{n}) \middle |Y^{(n)}\right) \\
    &=1 + o_{P_{0}^{(n)}}\left(1\right),
\end{align*}
and, then, by the law of total probability, we conclude that it suffices to show that, for every $\varepsilon > 0$,
\begin{align*}
    \Pi_{n}\left(\max_{1 \leq i \leq n}\max_{1 \leq k \leq K-1}|p_{k}(x_{i})-p_{0,k}(x_{i})| \geq \varepsilon \middle | Y^{(n)},\mathcal{P}_{n} \cap B_{n,2}(p_{0},M\tilde{\delta}_{n})\right) \overset{P_{0}^{(n)}}{\longrightarrow} 0
\end{align*}
as $n\rightarrow \infty$, conditionally given $P_{0,X}^{(\infty)}$-almost every fixed realization $\{x_{i}\}_{i \geq 1}$ of $\{X_{i}\}_{i \geq 1}$.
\paragraph{Step 2.} Using the same argument to Lemma \ref{lem:KLdivergence}, there is a constant $C:=C(K) \in (0,\infty)$ such that
\begin{align*}
    \max_{1 \leq i \leq n}\max_{1 \leq k \leq K}|p_{k}(x_{i}) - p_{0,k}(x_{i})| \leq C\max_{1 \leq i \leq n}\max_{1 \leq k \leq K-1}|B_{k}(x_{i})-b_{0,k}(x_{i})|
\end{align*}
Let $\widehat{Proj}_{J_{k}}$ be the sample least squares projection onto $\mathcal{CLS}^{J_{k}}$. Applying the triangle inequality,
\begin{align*}
 \max_{1 \leq i \leq n}\max_{1 \leq k \leq K}|B_{k}(x_{i})-b_{0,k}(x_{i})| &\leq    \max_{1 \leq i \leq n}\max_{1 \leq k \leq K}|B_{k}(x_{i})-(\widehat{Proj}_{J_{k}}B_{k})(x_{i})| \\
&\quad + \max_{1 \leq i \leq n}\max_{1 \leq k \leq K}|(\widehat{Proj}_{J_{k}}B_{k})(x_{i})-(\widehat{Proj}_{J_{k}}b_{0,k})(x_{i})| \\
&\quad +  \max_{1 \leq i \leq n}\max_{1 \leq k \leq K}|b_{0,k}(x_{i})-(\widehat{Proj}_{J_{k}}b_{0,k})(x_{i})|
\end{align*}
Consequently, by the union bound, for any $\varepsilon >0$,
\begin{align*}
    & \Pi_{n}\left(\max_{1 \leq i \leq n}\max_{1 \leq k \leq K-1}|p_{k}(x_{i})-p_{0,k}(x_{i})| \geq \varepsilon \middle | Y^{(n)},\mathcal{P}_{n} \cap B_{n,2}(p_{0},M\tilde{\delta}_{n})\right)  \\
    &\leq \Pi_{n}\left(\max_{1 \leq i \leq n}\max_{1 \leq k \leq K-1}\left|(\widehat{Proj}_{J_{k}}B_{k})(x_{i})-(\widehat{Proj}_{J_{k}}b_{0,k})(x_{i})\right| \geq \frac{\varepsilon}{3C}\middle | Y^{(n)},\mathcal{P}_{n} \cap B_{n,2}(p_{0},D\tilde{\delta}_{n})\right) \\
    &\quad + \mathbf{1}\left\{\max_{1 \leq i \leq n}\max_{1 \leq k \leq K-1}\left|b_{0,k}(x_{i})-(\widehat{Proj}_{J_{k}}b_{0,k})(x_{i})\right| \geq \frac{\varepsilon}{3C} \right\} \\
    &\quad + \Pi_{n}\left(\max_{1 \leq i \leq n}\max_{1 \leq k \leq K-1}\left|B_{k}(x_{i})-(\widehat{Proj}_{J_{k}}B_{k})(x_{i})\right| \geq \frac{\varepsilon}{3C} \middle | Y^{(n)}, \mathcal{P}_{n} \cap B_{n,2}(p_{0},D\tilde{\delta}_{n})\right).
\end{align*}
The following substeps derive sufficient conditions for each of these terms vanish in probability as $n\rightarrow \infty$, conditional on $P_{0,X}^{(\infty)}$-almost every fixed realization $\{x_{i}\}_{i \geq 1}$ of $\{X_{i}\}_{i \geq 1}$. Step 3 shows that all of these sufficient conditions are implied by Assumption \ref{as:wavelet}.
\paragraph{Step 2A.} We start with the first term in the upper bound. Let $\hat{\beta}_{k}^{J_{k}}$ and $\hat{\beta}_{0,k}^{J_{k}}$ be the coefficients in the least squares projections $\widehat{Proj}_{J_{k}}B_{k}$ and $\widehat{Proj}_{J_{k}}b_{0,k}$, and let $\Delta_{n,k}^{J_{k}} = \hat{\beta}_{k}^{J_{k}}-\hat{\beta}_{0,k}^{J_{k}}$. The triangle inequality reveals and the definition of the maximum reveals that
\begin{align*}
|(\widehat{Proj}_{J_{k}}B_{k})(x_{i})-(\widehat{Proj}_{J_{k}}b_{0,k})(x_{i})|\leq \sum_{j=0}^{J_{k}} \max_{0 \leq l \leq 2^{jd_{x}}-1}|\Delta_{n,k,j,l}^{J_{k}}|\sum_{l=0}^{2^{jd_{x}}-1}||\psi_{j,l}||_{\infty},
\end{align*}
where  $\Delta_{n,k,j,l}^{J_{k}}$ are the individual elements of $\Delta_{n,k}^{J_{k}}$
Using localization properties of CDV wavelets, $\sum_{l=0}^{2^{jd_{x}}-1}||\psi_{j,l}||_{\infty} \lesssim 2^{\frac{jd_{x}}{2}}$. Using the residual regression formula (see Chapter 17.3 of \cite{goldberger1991course}),
\begin{align*}
    |\Delta_{n,k,j,l}^{J_{k}}| \leq \underbrace{\min_{0 \leq j \leq J_{k}}\min_{0 \leq l \leq 2^{jd_{x}}-1}\left(\frac{1}{n}\sum_{i=1}^{n}\tilde{\psi}_{j,l}^{2}(x_{i}) \right)^{-\frac{1}{2}}}_{R_{n,J_{k}}}||B_{k}-b_{0,k}||_{n,2} \leq R_{n,J_{k}} \cdot ||B-b_{0}||_{n,2}
\end{align*}
for all $0 \leq j \leq J_{k}$ and $0 \leq l \leq 2^{jd_{x}}-1$, where $\tilde{\psi}_{j,l}(x_{i}) = \psi_{j,l}(x_{i}) - \hat{\tau}_{n,jl}'\psi^{J_{k}}_{-jl}(x_{i})$, $\psi^{J_{k}}_{-jl}(x_{i}) = (\psi_{j,l}(x_{i}): 0 \leq \tilde{j} \leq J_{k}, 0 \leq \tilde{l} \leq 2^{\tilde{j}d_{x}-1}, \tilde{j} \neq j, \tilde{l} \neq l)$, and 
\begin{align*}
\hat{\tau}_{n,jl} = \left(\frac{1}{n}\sum_{i=1}^{n}\psi_{-jl}^{J_{k}}(x_{i})\psi_{-jl}^{J_{k}}(x_{i})'\right)^{-1}\frac{1}{n}\sum_{i=1}^{n}\psi_{-jl}^{J_{k}}(x_{i})\psi_{j,l}(x_{i}).    
\end{align*}
for each $0 \leq j \leq J_{k}$ and $0 \leq l \leq 2^{jd_{x}-1}$, and
\begin{align*}
||B-b_{0}||_{n,2} = \sqrt{\frac{1}{n}\sum_{i=1}^{n}\sum_{k=1}^{K-1}(B_{k}(x_{i})-b_{0,k}(x_{i}))^{2}}    
\end{align*}
Consequently,
\begin{align}\label{eq:Bbound}
|(\widehat{Proj}_{J_{k}}B_{k})(x_{i})-(\widehat{Proj}_{J_{k}}b_{0,k})(x_{i})| \lesssim R_{n,J_{k}} 2^{\frac{J_{k}d_{x}}{2}}||B-b_{0}||_{n,2} 
\end{align}
Using that $||B_{k}||_{\infty} \leq \overline{B}_{k}$ and $||b_{0,k}||_{\infty} \leq \overline{B}_{k}$ for each $k=1,...,K-1$, we can conclude that there exists a constant $\acute{C}:=\acute{C}(\overline{B}_{1},...,\overline{B}_{K-1}) \in (0,\infty)$ such that
\begin{align*}
||B-b_{0}||_{n,2} \leq \acute{C} ||p-p_{0}||_{n,2}.
\end{align*}
Consequently, there is a constant $\check{C} > 0$ such that
\begin{align*}
|(\widehat{Proj}_{J_{k}}B_{k})(x_{i})-(\widehat{Proj}_{J_{k}}b_{0,k})(x_{i})| \leq \check{C}M\max_{1 \leq k \leq K-1}(R_{n,J_{k}}2^{\frac{J_{k}d_{x}}{2}}\tilde{\delta}_{n})     
\end{align*}
on the event $\mathcal{P}_{n} \cap B_{n,2}(p_{0},M\tilde{\delta}_{n})$. Consequently, if $\left\{J_{k}\right\}_{k = 1}^{K-1}$ is such that $R_{n,J_{k}}2^{J_{k}d_{x}/2}\tilde{\delta}_{n}\rightarrow 0$ for each $k=1,...,K-1$ as $n\rightarrow \infty$, conditional on the realization $\{x_{i}\}_{i \geq 1}$ of $\{X_{i}\}_{i \geq 1}$, we obtain that, for every $\varepsilon > 0$,
\begin{align*}
\Pi_{n}\left(\max_{1 \leq i \leq n}\max_{1 \leq k \leq K-1}\left|(\widehat{Proj}_{J_{k}}B_{k})(x_{i})-(\widehat{Proj}_{J_{k}}b_{0,k})(x_{i})\right| \geq \frac{\varepsilon}{3C}\middle | Y^{(n)},\mathcal{P}_{n} \cap B_{n,2}(p_{0},D\tilde{\delta}_{n})\right)
\end{align*}
is identically equal to zero for $n$ large, conditionally given $P_{0,X}^{(\infty)}$-almost every realization $\{x_{i}\}_{i \geq 1}$ of $\{X_{i}\}_{i \geq 1}$.
\paragraph{Step 2B.} Let $Proj_{J_{k}}$ be the population least squares projection of $b_{0,k}$ onto $\mathcal{CLS}^{J_{k}}$. Applying the triangle inequality,
\begin{align}
|b_{0,k}(x_{i})-(\widehat{Proj}_{J_{k}}b_{0,k})(x_{i})| &\leq |b_{0,k}(x_{i})-(Proj_{J_{k}}b_{0,k})(x_{i})| \nonumber\\
    &+ |(Proj_{J_{k}}b_{0,k})(x_{i})-(\widehat{Proj}_{J_{k}}b_{0,k})(x_{i})| \label{eq:popprojbound}
\end{align}
for each $k=1,...,K-1$ and $i=1,...,n$. We bound the first term in the upper bound (\ref{eq:popprojbound}). Define $b_{0,k}^{J_{k}}$ to be the element of $\mathcal{CLS}^{J_{k}}$ that satisfies $||b_{0,k}-b_{0,k}^{J_{k}}||_{\infty} = \inf_{b \in \mathcal{CLS}^{J_{k}}}||b_{0,k}-b||_{\infty}$. Applying the triangle inequality,
\begin{align*}
|b_{0,k}(x_{i})-(Proj_{J_{k}}b_{0,k})(x_{i})| &\leq (1+||Proj_{J_{k}}||_{\infty})||b_{0,k}-b_{0,k}^{J_{k}}||_{\infty},   
\end{align*}
where
\begin{align*}
    ||Proj_{J_{k}}||_{\infty} := \sup_{f: f \neq 0, \ ||f||_{\infty} < \infty}\frac{||Proj_{J_{k}}f||_{\infty}}{||f||_{\infty}}
\end{align*}
is the projection operator norm. Under Assumption \ref{as:covariates}, we can apply Theorem 5.1 of \cite{chen2015optimal} to conclude that $||Proj_{J_{k}}||_{\infty} \lesssim 1$ for all $J_{k}$ and $k=1,...,K-1$. Moreover, since $b_{0,k} \in C^{\alpha_{0,k}}([0,1]^{d_{x}},\mathbb{R})$, we know from uniform approximation theory for CDV wavelets (e.g., \cite{chen2007large}) that $||b_{0,k}-b_{0,k}^{J_{k}}||_{\infty} \leq ||b_{0,k}||_{\alpha_{0,k}}2^{-\alpha_{0,k}J_{k}}$. Consequently, there exists a constant $\{C_{k}(\alpha_{0,k},b_{0,k})\}_{k=1}^{K-1} \subseteq (0,\infty)$ such that
\begin{align*}
\max_{1 \leq i \leq n}\max_{1 \leq k \leq K-1}|b_{0,k}(x_{i})-(Proj_{J_{k}}b_{0,k})(x_{i})| \leq  \max_{1 \leq k \leq K-1}2^{-\alpha_{0,k}J_{k}} \cdot\max_{1 \leq k \leq K-1}C_{k}(\alpha_{0,k},b_{0,k}).   
\end{align*}
We now bound the second term in (\ref{eq:popprojbound}). Applying the Cauchy-Schwarz inequality,
\begin{align*}
|(Proj_{J_{k}}b_{0,k})(x_{i})-(\widehat{Proj}_{J_{k}}b_{0,k})(x_{i})| \leq ||\psi^{J_{k}}(x_{i})||_{2} \cdot ||\hat{\beta}^{J_{k}}_{0,k}-\beta^{J_{k}}_{0,k} ||_{2},
\end{align*}
where $\hat{\beta}^{J}_{0,k}$ and $\beta^{J}_{0,k}$ are the sample and population least squares coefficients, respectively. Using localization properties of CDV wavelets,
\begin{align*}
\sup_{x \in [0,1]^{d_{x}}}||\psi^{J_{k}}(x)||_{2} \lesssim 2^{\frac{J_{k}d_{x}}{2}}
\end{align*}
Following the same arguments to Step 3 of the proof of Proposition 2 in \cite{walker2026semiparametric},
\begin{align*}
    ||\hat{\beta}^{J_{k}}_{0,k}-\beta^{J_{k}}_{0,k} ||_{2} \leq \lambda_{min}\left(\frac{1}{n}\sum_{i=1}^{n}\psi^{J_{k}}(x_{i})\psi^{J_{k}}(x_{i})'\right)^{-\frac{1}{2}}\left|\left|b_{0,k}-Proj_{J_{k}}b_{0,k}\right|\right|_{\infty}
\end{align*}
Consequently, there exists constants $\{\tilde{C}_{k}(\alpha_{0,k},b_{0,k})\}_{k=1}^{K-1} \subseteq (0,\infty)$ such that
\begin{align*}
||\psi^{J_{k}}(x_{i})||_{2} \cdot ||\hat{\beta}^{J_{k}}_{0,k}-\beta^{J_{k}}_{0,k} ||_{2} &\lesssim \max_{1 \leq k \leq K-1}\lambda_{min}\left(\frac{1}{n}\sum_{i=1}^{n}\psi^{J_{k}}(x_{i})\psi^{J_{k}}(x_{i})'\right)^{-\frac{1}{2}}\\
&\quad \times \max_{1 \leq k \leq K-1}\tilde{C}_{k}(\alpha_{0,k},b_{0,k}) \cdot \max_{1 \leq k \leq K-1}2^{J_{k}\left(\frac{d_{x}}{2}-\alpha_{0,k}\right)}   
\end{align*}
Together, these results imply that there is a constant $\acute{C}(\alpha_{0},b_{0}) \in (0,\infty)$ such that, for $n$ large,
\begin{align*}
    &\max_{1 \leq i \leq n}\max_{1 \leq k \leq K-1}|b_{0,k}(x_{i})-(\widehat{Proj}_{J_{k}}b_{0,k})(x_{i})| \\
    &\quad \lesssim \acute{C}(\alpha_{0},b_{0})\left(\max_{1 \leq k \leq K-1}\lambda_{min}\left(\frac{1}{n}\sum_{i=1}^{n}\psi^{J_{k}}(x_{i})\psi^{J_{k}}(x_{i})'\right)^{-\frac{1}{2}}\max_{1 \leq k \leq K-1}2^{J_{k}(\frac{d_{x}}{2}-\alpha_{0,k})}
    \right. \\
    &\quad \quad \quad \quad \quad \quad \quad \quad \quad \left.+ \max_{1 \leq k \leq K-1}2^{-J_{k}\alpha_{0,k}}\right)
\end{align*}
Since $\alpha_{0,k} > d_{x}/2$ for all $k=1,...,K-1$ (by assumption), it follows that, for any $\{J_{k}\}_{k =1}^{K-1}$ such that $\min_{1 \leq k \leq K-1}J_{k} \rightarrow \infty$ as $n\rightarrow \infty$ and, for $n$ large, 
\begin{align*}
    \lambda_{min}\left(\frac{1}{n}\sum_{i=1}^{n}\psi^{J_{k}}(x_{i})\psi^{J_{k}}(x_{i})'\right) \geq \tilde{c}
\end{align*}
for some $\tilde{c}>0$, the following is true for every $\varepsilon > 0$,
\begin{align*}
\mathbf{1}\left\{\max_{1 \leq i \leq n}\max_{1 \leq k \leq K-1}\left|b_{0,k}(x_{i})-(\widehat{Proj}_{J_{k}}b_{0,k})(x_{i})\right| \geq \frac{\varepsilon}{3C} \right\} = o(1)
\end{align*}
as $n\rightarrow \infty$. 
\paragraph{Step 2C.} Adding/subtracting $Proj_{J_{k}}B_{k}$ and applying the triangle inequality,
\begin{align*}
|B_{k}(x_{i})-(\widehat{Proj}_{J_{k}}B_{k})(x_{i})| &\leq \left|B_{k}(x_{i}) - (Proj_{J_{k}}B_{k})(x_{i})\right| \\
    &\quad +\left|(Proj_{J_{k}}B_{k})(x_{i})-(\widehat{Proj}_{J_{k}}B_{k})(x_{i})\right|
\end{align*}
for all $k=1,...,K-1$ and all $i=1,...,n$. Since $\max_{1 \leq k \leq K-1}||B_{k}||_{a_{k}} \leq M\sqrt{n}\tilde{\delta}_{n}$ on the event $\mathcal{P}_{n} \cap B_{n,2}(p_{0},M\tilde{\delta}_{n})$, a similar argument to Step 2B reveals that there is a constant $C_{k}$ independent of $B_{k}$ such that
\begin{align*}
\left|B_{k}(x_{i}) - (Proj_{J_{k}}B_{k})(x_{i})\right| \leq MC_{k} \sqrt{n}\tilde{\delta}_{n}2^{-a_{k}J_{k}}   
\end{align*}
on the event $\mathcal{P}_{n} \cap B_{n,2}(p_{0},M\tilde{\delta}_{n})$ for all $k=1,...,K-1$ and all $i=1,...,n$. Likewise, following similar arguments to Step 2C, we can also show that, on the event $\mathcal{P}_{n} \cap B_{n,2}(p_{0},M\tilde{\delta}_{n})$,
\begin{align*}
&\left| (Proj_{J_{k}}B_{k})(x_{i})-(\widehat{Proj}_{J_{k}}B_{k})(x_{i})\right| \\
&\quad \lesssim \lambda_{min}\left(\frac{1}{n}\sum_{i=1}^{n}\psi^{J_{k}}(x_{i})\psi^{J_{k}}(x_{i})'\right)^{-\frac{1}{2}}MC_{k}\sqrt{n}\tilde{\delta}_{n}2^{J_{k}\left(\frac{d_{x}}{2}-a_{k}\right)} 
\end{align*}
for all $k=1,...,K-1$ and all $i=1,...,n$. Consequently,
\begin{align*}
 &\max_{1 \leq i \leq n}\max_{1 \leq k \leq K-1}\left|B_{k}(x_{i})-(\widehat{Proj}_{J_{k}}B_{k})(x_{i})\right| \\
 &\quad \leq \max_{1 \leq k \leq K-1}C_{k} M\sqrt{n}\tilde{\delta}_{n}\left(\max_{1 \leq k \leq K-1}\lambda_{min}\left(\frac{1}{n}\sum_{i=1}^{n}\psi^{J_{k}}(x_{i})\psi^{J_{k}}(x_{i})'\right)^{-\frac{1}{2}}\max_{1 \leq k \leq K-1}2^{J_{k}(\frac{d_{x}}{2}-a_{k})} \right.\\
&\quad \quad \quad \quad \quad \quad \quad \quad \quad \quad \quad \quad \quad \left. +\max_{1 \leq k \leq K-1}2^{-a_{k}J_{k}}\right)    
\end{align*}
for all $p \in \mathcal{P}_{n} \cap B_{n,2}(p_{0},M\tilde{\delta}_{n})$. Since $d_{x}/2 < a_{k}$ for each $k=1,...,K-1$, it follows that if $\{J_{k}\}_{k=1}^{K-1}$ is such that $\sqrt{n}\tilde{\delta}_{n}2^{J_{k}(d_{x}/2-a_{k})}\rightarrow 0$ for each $k=1,...,K-1$ as $n\rightarrow \infty$ and, for large $n$,
\begin{align*}
    \lambda_{min}\left(\frac{1}{n}\sum_{i=1}^{n}\psi^{J_{k}}(x_{i})\psi^{J_{k}}(x_{i})'\right) \geq c
\end{align*}
for each $k=1,...,K-1$, then, for any $\varepsilon > 0$ and $n$ large,
\begin{align*}
\Pi_{n}\left(\max_{1 \leq i \leq n}\max_{1 \leq k \leq K-1}\left|B_{k}(x_{i})-(\widehat{Proj}_{J_{k}}B_{k})(x_{i})\right| \geq \frac{\varepsilon}{3C} \middle | Y^{(n)}, \mathcal{P}_{n} \cap B_{n,2}(p_{0},D\tilde{\delta}_{n})\right)
\end{align*}
is identically zero, conditionally given $P_{0,X}^{(\infty)}$-almost every fixed realization $\{x_{i}\}_{i \geq 1}$ of $\{X_{i}\}_{i \geq 1}$.
\paragraph{Step 3.} From Step 2, we conclude the proof if $\{J_{k}\}_{k = 1}^{K-1}$ is defined such that 1. $J_{k} \rightarrow \infty$ as $n\rightarrow \infty$ for each $k=1,...,K-1$, 2. $\sqrt{n}\tilde{\delta}_{n}2^{J_{k}\left(\frac{d_{x}}{2}-a_{k}\right)} \rightarrow 0$ as $n\rightarrow \infty$ for each $k=1,...,K-1$, 3. $2^{\frac{J_{k}d_{x}}{2}}\tilde{\delta}_{n}\rightarrow 0$ as $n\rightarrow \infty$ for each $k=1,...,K-1$, 4. there exists $c > 0$ such that $\min_{1 \leq k \leq K-1}\lambda_{min}(n^{-1}\sum_{i=1}^{n}\psi^{J_{k}}(x_{i})\psi^{J_{k}}(x_{i})') \geq c $ for $n$ large, and 5. there exists $\tilde{c}>0$ such that $\min_{1 \leq k \leq K-1}R_{n,J_{k}} \geq \tilde{c}$ for $n$ large, conditionally given $P_{0,X}^{(\infty)}$-almost every fixed realization $\{x_{i}\}_{i \geq 1}$ of $\{X_{i}\}_{i\geq 1}$. Condition 1. is satisfied for any $J_{k}:=J_{n,k} \rightarrow \infty$ as $n\rightarrow \infty$. Conditions 2. and 3. require that
\begin{align*}
     \left(\sqrt{n}\tilde{\delta}_{n}\right)^{\frac{1}{a_{k}-\frac{d_{x}}{2}}} = o(2^{J_{k}})
\end{align*}
and
\begin{align*}
    2^{J_{k}} = o\left(\tilde{\delta}_{n}^{-\frac{2}{d_{x}}}\right),
\end{align*}
respectively, both of which are implied by Assumption \ref{as:wavelet}. Under Assumption \ref{as:covariates}, we can follow the argument of Lemma 12 in \cite{walker2026semiparametric} to conclude that Condition 4. holds if
\begin{align*}
    2^{J_{k}} = o\left(\left(\frac{n}{\log n}\right)^{\frac{1}{d_{x}}}\right),
\end{align*}
which is also implied by Assumption \ref{as:wavelet}. Moreover,  Condition 4. implies that $\{R_{n,J_{k}}\}_{k=1}^{K-1}$ from Step 2A satisfies $R_{n,J_{k}} \geq c$ for $n$ large because, letting $\hat{\nu}_{n,l,j}$ be defined such that $\hat{\nu}_{n,l,j}'\psi^{J_{k}}(x_{i}) = \psi_{j,l}(x_{i})-\hat{\tau}_{n,-jl}\psi_{-jl}^{J_{k}}(x_{i})$, we know that  
\begin{align*}
    \frac{1}{n}\sum_{i=1}^{n}\tilde{\psi}_{j,l}^{2}(x_{i}) &= \frac{1}{n}\sum_{i=1}^{n}(\psi_{j,l}(x_{i})-\hat{\tau}_{n,-jl}'\psi_{-jl}^{J_{k}}(x_{i}))^{2} \\
    &= \hat{\nu}_{n,l,j}'\left(\frac{1}{n}\sum_{i=1}^{n}\psi^{J_{k}}(x_{i})\psi^{J_{k}}(x_{i})'\right)\hat{\nu}_{n,l,j} \\
    &\geq ||\hat{\nu}_{n,l,j}||_{2}^{2} \lambda_{min}\left(\frac{1}{n}\sum_{i=1}^{n}\psi^{J_{k}}(x_{i})\psi^{J_{k}}(x_{i})'\right) \\
    &\geq \lambda_{min}\left(\frac{1}{n}\sum_{i=1}^{n}\psi^{J_{k}}(x_{i})\psi^{J_{k}}(x_{i})'\right) \\
    &\geq c
\end{align*}
for all $j,l$, where the first inequality uses the variational characterization of eigenvalues, and the second inequality uses that $||\hat{\nu}_{n,l,j}||_{2}^{2} = ||\hat{\tau}_{n,-jl}||_{2}^{2} + 1 \geq 1$ for all $j,l$. Hence Condition 5. is satisfied for $\tilde{c} = c$. This completes the proof.
\end{proof}
\subsubsection{Verifying Assumption \ref{as:wavelet} for Gaussian Priors}\label{ap:verifyingwavelet}
This section verifies that there are data-generating processes and priors for which Assumption \ref{as:wavelet} holds. Suppose that $\kappa_{k}$ belongs to the Mat\'{e}rn class with smoothness parameter $\alpha_{k} > 0$ for each $k=1,...,K-1$. That is, for any $x,\tilde{x} \in [0,1]^{d_{x}}$, the covariance function $\kappa_{k}$ satisfies
\begin{align*}
    \kappa_{k}(x,\tilde{x}) = \int_{\mathbb{R}^{d}}\exp(i \lambda'(x-\tilde{x}))(1+||\lambda||_{2}^{2})^{-\alpha_{k}-\frac{d_{x}}{2}} d \lambda. 
\end{align*}
Further, suppose that $b_{0,k} \in C^{\alpha_{0,k}}([0,1]^{d_{x}},\mathbb{R}) \cap S^{\alpha_{0,k}}([0,1]^{d_{x}},\mathbb{R})$ for some $\alpha_{0,k} > d_{x}/2$. Then, Theorem 5 of \cite{JMLR:v12:vandervaart11a} implies that $\tilde{\delta}_{n,k} = n^{-r_{k}}$, where $r_{k}=\min\{\alpha_{k},\alpha_{0,k}\}/(2\alpha_{k}+d_{x})$ for each $k=1,...,K-1$, and, as a result, $\tilde{\delta}_{n} = n^{-r}$ with $r=\min_{1 \leq k \leq K-1}r_{k}$. This means that satisfying the first part of Assumption \ref{as:wavelet} requires $2^{J_{k}}/n^{2r/d_{x}} \rightarrow 0$ as $n\rightarrow \infty$, whereas the second part of Assumption \ref{as:wavelet} requires $n^{(1/2-r)/(a_{k}-d_{x}/2)}/2^{J_{k}} \rightarrow 0$ as $n\rightarrow \infty$.\footnote{The condition $2^{J_{k}}/((\log n)^{-1}n)^{1/d_{x}}$ is satisfied too because $r < 1/2$.} Consequently, if $2^{J_{k}} = n^{\zeta_{k}}$ with
\begin{align*}
    \frac{\frac{1}{2}-r}{\alpha_{k}-\frac{d_{x}}{2}}<\zeta_{k} < \frac{2r}{d_{x}}
\end{align*}
for each $k=1,...,K-1$, then Assumption \ref{as:wavelet} holds.\footnote{Since we have the flexibility to select $a_{k} \in (d_{x}/2,\alpha_{k})$ in Assumption \ref{as:wavelet}, we can choose $a_{k}$ to minimize $(1/2-r)/(\alpha_{k}-d_{x}/2)$ which yields $(1/2-r)/(\alpha_{k}-d_{x}/2)$.} This is equivalent to $\alpha_{k} > d_{x}/4r$ for each $k=1,...,K-1$. In the case where $\alpha_{k} = \alpha_{0,k}$ for all $k=1,...,K-1$, the condition reduces to $\min_{1 \leq k \leq K-1}\alpha_{k} > (1+\sqrt{5})d_{x}/4$.
\subsection{Proofs of Propositions \ref{prop:GPmultinomial} and \ref{prop:continuousY}}\label{ap:extensionproofs}
\begin{proof}[Proof of Proposition \ref{prop:GPmultinomial}]
The proof has two steps. Step 1 extends Proposition \ref{prop:L2} to the multinomial setting, while Step 2 extends Proposition \ref{prop:post_sup_norm} to conclude the result.
\paragraph{Step 1.} We show that, for $P_{0,X}^{(\infty)}$-almost every fixed realization $\{x_{i}\}_{i \geq 1}$ of $\{X_{i}\}_{i \geq 1}$, 
\begin{align}\label{eq:L2multinomial}
\Pi_{n}\left(||p-p_{0}||_{n,2} < M\tilde{\delta}_{n}|Y^{(n)}\right) \overset{P_{0}^{(n)}}{\longrightarrow} 1    
\end{align} 
as $n\rightarrow \infty$, where $\tilde{\delta}_{n} = \max_{1 \leq k \leq K-1}\tilde{\delta}_{n,k}$. To that end, we first verify that Lemmas \ref{lem:KLdivergence} and \ref{lem:normalizing} extend. For Lemma \ref{lem:KLdivergence}, since the multinomial density is $m_{i}$-fold product of Categorical densities, the \textit{multinomial} KL divergence and variance satisfies
\begin{align*}
    KL(p_{0}(x_{i}),p(x_{i})) = m_{i} \sum_{k=1}^{K}p_{0,k}(\tilde{x}_{i})\log\left(\frac{p_{0,k}(\tilde{x}_{i})}{p_{k}(\tilde{x}_{i})}\right)
\end{align*}
and
\begin{align*}
 KLV(p_{0}(\tilde{x}_{i}),p(\tilde{x}_{i})) = m_{i} \sum_{k=1}^{K}p_{0,k}(\tilde{x}_{i})\left(\log \frac{p_{0,k}(\tilde{x}_{i})}{p_{k}(\tilde{x}_{i})}\right)^{2}.   
\end{align*}
Moreover, since there exists $\bar{M} >0$ such that $P_{0,X}^{(\infty)}(\sup_{n \geq 1}\max_{1 \leq i \leq n}|M_{i}| \leq \bar{M}) = 1$, and $||B_{k}||_{\infty}, ||b_{0,k}||_{\infty} \leq \overline{B}_{k}$ and $||b_{0,k}||_{\infty} \leq \overline{B}_{k}$ for each $k=1,...,K-1$, a Taylor expansion reveals that
\begin{align*}
    \max\left\{\frac{1}{n}\sum_{i=1}KL(p_{0}(\tilde{x}_{i}),p(\tilde{x}_{i})), \frac{1}{n}\sum_{i=1}^{n}KLV(p_{0}(\tilde{x}_{i}),p(\tilde{x}_{i}))\right\} \leq \breve{C}\bar{M}\sum_{k=1}^{K}||p_{k}-p_{0,k}||_{\infty}^{2}.
\end{align*}
for $P_{0,X}^{(\infty)}$-almost every fixed realization $\{x_{i}\}_{i \geq 1}$ of $\{X_{i}\}_{i \geq 1}$. Consequently, we can follow the argument of Lemma \ref{lem:KLdivergence} to conclude that there are constants $\tilde{C}_{1},\tilde{C}_{2} \in (0,\infty)$ such that
\begin{align*}
    \Pi(p \in \mathcal{U}_{n,2}^{*}(p_{0},\tilde{C}_{1}\tilde{\delta}_{n})) \geq C_{2}\exp(-(K-1)n\tilde{\delta}_{n}^{2})
\end{align*}
for $P_{0,X}^{(\infty)}$-almost every realization $\{x_{i}\}_{i \geq 1}$ of $\{X_{i}\}_{i \geq 1}$, where $\mathcal{U}_{n,2}^{*}(p_{0},\delta)$ is defined as before except replacing the categorical KL divergence and variation with the multinomial counterpart. This verifies that Lemma \ref{lem:KLdivergence} extends to the multinomial setting. An inspection of the proof reveals that Lemma \ref{lem:normalizing} also extends, except replacing $C_{1}$ with $\tilde{C}_{1}$. Given these results, we argue that there exists a constant $C \in (0,\infty)$ independent of $\{x_{i}\}_{i \geq 1}$ and $K$ for which
\begin{align*}
||p-p_{0}||_{n,2} \leq C h_{n}(p,p_{0}),
\end{align*}
where $h_{n}$ is the root mean squared Hellinger distance between $Multinomial(m_{i},p(\tilde{x}_{i}))$ and $Multinomial(m_{i},p_{0}(\tilde{x}_{i}))$. Indeed, if this holds, then we can repeat the same argument as Proposition \ref{prop:L2} to conclude (\ref{eq:L2multinomial}) (in particular, the tests from Lemma 2 of \citep{ghosal2007convergence} also exist for the multinomial distribution). Let $h(p(\tilde{x}_{i}),p_{0}(\tilde{x}_{i}))$ be the Hellinger distance between $Multinomial(m_{i},p(\tilde{x}_{i}))$ and $ Multinomial(m_{i},p_{0}(\tilde{x}_{i}))$, Lemma B.5 of \cite{ghosal2017fundamentals} and properties of multinomial distributions reveals that
\begin{align*}
h^{2}(p(\tilde{x}_{i}),p_{0}(\tilde{x}_{i})) &= 2\left(1-\left(\sum_{k=1}^{K}(p_{k}(\tilde{x}_{i})p_{0,k}(\tilde{x}_{i}))^{\frac{1}{2}}\right)^{m_{i}}\right) \\
&=2\left(1- \left(1-\tilde{h}^{2}(p(\tilde{x}_{i}),p_{0}(\tilde{x}_{i}))\right)^{m_{i}}\right),    
\end{align*}
where $\tilde{h}(p(\tilde{x}_{i}),p_{0}(\tilde{x}_{i}))$ is the Hellinger distance between the categorical distributions $Categorical(p(\tilde{x}_{i}))$ and $Categorical(p_{0}(\tilde{x}_{i}))$. Since $\tilde{h}(p(\tilde{x}_{i}),p_{0}(\tilde{x}_{i}))\in [0,1]$ and $m_{i} \geq 1$,
\begin{align*}
    \left(1-\tilde{h}^{2}(p(\tilde{x}_{i}),p_{0}(\tilde{x}_{i}))\right)^{m_{i}} \leq 1-\tilde{h}^{2}(p(\tilde{x}_{i}),p_{0}(\tilde{x}_{i})),
\end{align*}
which, combined with the above, implies that
\begin{align*}
    h^{2}(p(\tilde{x}_{i}),p_{0}(\tilde{x}_{i}))\geq 2 \tilde{h}^{2}(p(\tilde{x}_{i}),p_{0}(\tilde{x}_{i})) \geq 2 c ||p(\tilde{x}_{i})-p_{0}(\tilde{x}_{i})||_{2}^{2}
\end{align*}
for some universal constant $c>0$, where the last inequality uses the equivalence between Hellinger and Total Variation distance (i.e., taxicab norm for discrete distributions) and then the equivalence between norms in Euclidean space. By averaging over the $\tilde{x}_{i}$s and taking the square root, it follows there exists a constant $C \geq 0$ independent of $\{x_{i}\}_{i \geq 1}$ and $K$ such that $||p-p_{0}||_{n,2} \leq C h_{n}(p,p_{0})$, which completes the extension of Proposition \ref{prop:L2} to the multinomial setting.
\paragraph{Step 2.} We now show that, for $P_{0,X}^{(\infty)}$-almost every fixed realization $\{x_{i}\}_{i \geq 1}$ of $\{X_{i}\}_{i \geq 1}$, the following holds: for every $\varepsilon > 0$,
\begin{align}\label{eq:supmultinomial}
\Pi_{n}\left(||p-p_{0}||_{n,\infty} \geq \varepsilon |Y^{(n)}\right) \overset{P_{0}^{(n)}}{\longrightarrow} 0    
\end{align} 
as $n\rightarrow \infty$. Following the same argument as Step 1 of Proposition \ref{prop:post_sup_norm}, we can show that
\begin{align*}
&\Pi_{n}\left(||p-p_{0}||_{n,\infty} \geq \varepsilon |Y^{(n)}\right)\\
&\quad = \Pi_{n}\left(||p-p_{0}||_{n,\infty} \geq \varepsilon |Y^{(n)}, \mathcal{P}_{n} \cap B_{n,2}(p_{0},M\tilde{\delta}_{n})\right)\left(1+ o_{P_{0}^{(n)}}(1)\right) + o_{P_{0}^{(n)}}(1),   
\end{align*}
where $\mathcal{P}_{n} = \{(p_{1},...,p_{K}): \max_{1 \leq k \leq K-1}||B_{k}||_{a_{k}} \leq M\sqrt{n}\tilde{\delta}_{n}\}$ with $a_{1},...,a_{K-1}$ chosen so that Assumption \ref{as:wavelet} is satisfied. Moreover, since in Steps 2 and 3, we only make use of the event $\mathcal{P}_{n} \cap B_{n,2}(p_{0},M\tilde{\delta}_{n})$, Assumption \ref{as:wavelet}, and that $b_{0,k} \in C^{\alpha_{0,k}}([0,1]^{d_{x}},\mathbb{R})$ with $\alpha_{0,k}> d_{x}/2$ for each $k=1,...,K-1$, the argument automatically extends to the multinomial model, thereby leading to conclusion that, for $P_{0,X}^{(\infty)}$-almost every fixed realization $\{x_{i}\}_{i \geq 1}$ of $\{X_{i}\}_{i \geq 1}$, (\ref{eq:supmultinomial}) holds. This completes the proof.
    
\end{proof}
\begin{assumption}\label{as:continuousDGP}
The following conditions hold:
\begin{enumerate}
    \item $\Gamma$ is a compact, convex subset of $\mathbb{R}^{d_{\gamma}}$
    \item The functions $f(x,p_{Y|X}(\cdot|x),\gamma)$ are differentiable and concave in $\gamma$ for all $(x,p_{Y|X})$,
\end{enumerate}
and, for $P_{0,X}^{(\infty)}$-almost every fixed realization $\{x_{i}\}_{i \geq 1}$ of $\{X_{i}\}_{i \geq 1}$,
\begin{enumerate}
\setcounter{enumi}{2}
\item The criterion $Q_{n,r}(\cdot,p_{0})$ satisfies $\sup_{\gamma \in \Gamma}Q_{n,r}(\gamma,p_{0,Y|X})< \infty$,
\item There exists constants $0 < \underline{\lambda}_{W} \leq \overline{\lambda}_{W} < \infty$ such that 
\begin{align*}
\underline{\lambda}_{W} \leq  \inf_{\gamma \in \Gamma} \inf_{n \geq 1}\min_{1 \leq i \leq n} \lambda_{min}(W(x_{i},p_{0,Y|X}(x_{i}),\gamma))     
\end{align*}
and
\begin{align*}
\sup_{\gamma \in \Gamma} \sup_{n \geq 1}\max_{1 \leq i \leq n} \lambda_{max}(W(x_{i},p_{0,Y|X}(x_{i}),\gamma)) \leq \overline{\lambda}_{W}.    
\end{align*}
\item There exists $\gamma^{\circ} \in \Gamma$, $\bar{N} \geq 1$, $\xi >0$ and $\eta >0$ such that $\Gamma_{n,I}(p_{0,Y|X}) \subseteq \operatorname{int}(\Gamma)$, 
\begin{align*}
\sup_{\gamma \in \partial \Gamma}\max_{1\leq i \leq n}\max_{1 \leq j \leq d_{f}}f_{j}(x_{i},p_{0,Y|X}(x_{i}),\gamma) \leq -\xi,    
\end{align*} 
and 
\begin{align*}
\min_{1\leq i \leq n}\min_{1 \leq j \leq d_{f}}f_{j}(x_{i},p_{0,Y|X}(x_{i}),\gamma^{\circ}) \geq \eta 
\end{align*}
for all $n \geq \bar{N}$. 
\end{enumerate}
\end{assumption}
\begin{proof}[Proof of Proposition \ref{prop:continuousY}]
Since the first two conditions is the continuous counterpart of Assumption \ref{as:prior}.1, it suffices to show that Assumptions \ref{as:prior}.2 and Assumptions \ref{as:prior}.3 extend to the case with $p_{Y|X}$ instead of $p$. This is done in Steps 1 and 2.
\paragraph{Step 1.} To verify Assumption \ref{as:prior}.2, we show that, for $P_{0,X}^{(\infty)}$-almost every fixed realization $\{x_{i}\}_{i \geq 1}$ of $\{X_{i}\}_{i \geq 1}$, the following is true:
\begin{align}\label{eq:cond1continuous}
    \sup_{(\gamma,p_{Y|X}) \in \Gamma \times \mathcal{P}_{n,Y|X}: e_{n}(p_{Y|X},p_{0,Y|X})<\delta_{n}}\left|Q_{n,r}(\gamma,p_{Y|X})-Q_{n,r}(\gamma,p_{0,Y|X})\right|\longrightarrow 0
\end{align}
as $n\rightarrow \infty$. Under Assumption \ref{as:continuousDGP}.3 and \ref{as:continuousDGP}.4 and the uniform convergence conditions, the same argument as Proposition \ref{prop:ucon}, except replacing $p$ with $p_{Y|X}$, yields (\ref{eq:cond1continuous}).
\paragraph{Step 2.} To show Assumption \ref{as:prior}.3, we show that, for $P_{0,X}^{(\infty)}$-almost every fixed realization $\{x_{i}\}_{i \geq 1}$ of $\{X_{i}\}_{i \geq 1}$, the following is true: for every $\varepsilon > 0$, there exists $N \geq 1$ such that $p_{Y|X} \in \mathcal{P}_{n,Y|X} \cap \{p_{Y|X}: e_{n}(p_{Y|X},p_{0,Y|X})< \delta_{n}\}$ and $n \geq N$ implies $\Gamma_{n,I}(p) \cap B(\gamma,\varepsilon) \neq \emptyset$ for all $\gamma \in \Gamma_{n,I}(p_{0,Y|X})$. Repeating an identical argument to Proposition \ref{prop:hemi} (except with $p$ replaced by $p_{Y|X}$), a sufficient condition is that, for $P_{0,X}^{(\infty)}$-almost every fixed realization $\{x_{i}\}_{i \geq 1}$ of $\{X_{i}\}_{i \geq 1}$, there exists a constant $c \in (0,\infty)$ and $\tilde{N} \geq 1$, such that, for each $\gamma \in \Gamma$,
\begin{align*}
    \max_{1 \leq i \leq n}||(f(x_{i},p_{Y|X}(\cdot|x_{i}),\gamma))_{-}||_{2} \geq c \cdot \text{dist}(\gamma,\Gamma_{n,I}(p_{Y|X}))
    \end{align*}
holds for all $p_{Y|X} \in \mathcal{P}_{n,Y|X} \cap \{p_{Y|X}: e_{n}(p_{Y|X},p_{0,Y|X})< \delta_{n}\}$. Under the conditions of the Proposition and Assumption \ref{as:continuousDGP}, we can repeat the argument of Proposition \ref{prop:convexhemi}, except replacing $p$ with $p_{Y|X}$, to that the sufficient condition is satisfied.
\end{proof}
\end{document}